\documentclass{article}
\usepackage{iclr2027_conference,times}
\usepackage{amsmath,amssymb,amsthm,mathtools,bm}
\usepackage{booktabs}
\usepackage{algorithm}
\usepackage{algpseudocode}
\usepackage{microtype}
\usepackage{enumitem}
\usepackage[colorlinks=true,linkcolor=blue,citecolor=blue,urlcolor=blue]{hyperref}
\usepackage{url}
\usepackage{threeparttable}
\usepackage{graphicx}
\newcommand{\safeincludegraphics}[2][]{%
\IfFileExists{#2}{\includegraphics[#1]{#2}}{%
\fbox{\begin{minipage}[c][0.28\textheight][c]{0.95\textwidth}\centering Missing figure file: \texttt{\detokenize{#2}}\end{minipage}}}}
\theoremstyle{definition}

\newtheorem{assumption}{Assumption}[section]

\theoremstyle{plain}
\newtheorem{lemma}{Lemma}[section]
\newtheorem{proposition}{Proposition}[section]
\newtheorem{theorem}{Theorem}[section]
\newtheorem{corollary}{Corollary}[section]
\theoremstyle{remark}

\makeatletter
\def\section{\@startsection{section}{1}{\z@}
  {-0.5ex plus -0.4ex minus -.2ex}
  {0.2ex plus .12ex minus .08ex}
  {\large\sc\raggedright}}
\def\subsection{\@startsection{subsection}{2}{\z@}
  {-0.5ex plus -0.4ex minus -.2ex}
  {0.1ex plus .10ex minus .06ex}
  {\normalsize\sc\raggedright}}
\def\subsubsection{\@startsection{subsubsection}{3}{\z@}
  {-0.5ex plus -0.35ex minus -.15ex}
  {0.1ex plus .06ex minus .04ex}
  {\normalsize\sc\raggedright}}
\makeatother
\title{Simultaneous Inference for Proximal Causal Learning with Weak Proxies}
\author{
Pengyun Wang\\
Data Science Institute\\
University of Chicago
\And
Qi Huang\\
Committee on Computational and Applied Mathematics\\
University of Chicago
\AND
Jia Huang\\
Columbia University
\And
Yutong Zhou\\
Independent Researcher
}
\date{}
\iclrfinalcopy
\begin{document}
\maketitle

\begin{abstract}
Proximal causal learning uses the treatment-inducing proxy and the outcome-inducing proxy through the outcome bridge and the treatment bridge to provide a way to address unobserved confounding in observational studies.
This paper studies simultaneous inference for the marginal counterfactual distribution and quantile treatment effect under weak proxies.
First, we propose a measure of weak-proxy strength and tightens the population-bias bound for the counterfactual CDF estimate.
The baseline Cauchy--Schwarz bound introduces the inverse of weak-proxy strength twice, whereas the proposed bound introduces the inverse of weak-proxy strength only once and is sharp for the bridge estimators, thereby relaxing the restriction on the decay rate of weak-proxy strength.
Second, we prove that, even when outcome thresholds vary continuously, the number of distinct estimated outcome bridges is at most one greater than the number of observations in the training sample, thereby establishing a uniform expansion across treatment arms and outcome thresholds.
Finally, we prove that the critical value remains valid after bridge estimation, thereby obtaining simultaneous CDF bands, and constructs quantile treatment effect bands through confidence-band inversion.
The proposed method provides deterministic band-width bounds and does not require counterfactual density estimation.
In the numerical study, the baseline Cauchy--Schwarz bound overestimates the true population bias by a factor of approximately 50, whereas the proposed bound is very close to the true population bias.
The quantile treatment effect bands constructed by the proposed method achieve coverage of 99.8\%--100.0\% while being only approximately 20\% wider than the oracle 95\% bands, compared with baseline coverage of 68.4\%--89.2\%.
\end{abstract}

\section{Introduction}
\label{sec:introduction}

Proximal causal learning uses the indirect information provided by the treatment-inducing proxy and the outcome-inducing proxy to estimate causal effects under appropriate identification conditions by estimating the outcome bridge and the treatment bridge, providing a way to address confounding in observational studies where the unobserved latent variable cannot be directly observed \citep{cui2024semiparametric,bozkurt2026doubly}.
This paper focuses on the marginal counterfactual distribution and quantile treatment effect to characterize the effects of treatment at different positions of the outcome distribution, rather than being limited to mean effects \citep{chernozhukov2020generic}.
\par
Moving from bridge estimation to simultaneous inference for the entire marginal counterfactual distribution requires addressing three challenges:

\textit{First}, weak proxies amplify bridge estimation errors \citep{bennett2023source}; directly multiplying the two bridge estimation error bounds to obtain a Cauchy--Schwarz bound overestimates the population bias of the counterfactual CDF estimate, thereby imposing overly stringent requirements on weak-proxy strength.
\textit{Second}, the treatment bridge is the same across all outcome thresholds, whereas the outcome bridge varies with outcome thresholds, requiring control of errors other than population bias across treatment arms and outcome thresholds \citep{tchetgen2024introduction}.
\textit{Third}, commonly used methods for constructing confidence bands for the quantile treatment effect require estimating the counterfactual density \citep{chernozhukov2013inference}, and may be unreliable when the counterfactual density is close to zero or does not satisfy smoothness conditions \citep{firpo2007efficient}.
\par
First, this paper complements existing proximal doubly robust analyses \citep{cui2024semiparametric,bennett2023source} by retaining the stabilizing effect of ridge and ridge bias in the bridge estimation error bounds (Proposition~\ref{prop:method-bridge-rates}) and tightening the population-bias bound (Corollary~\ref{cor:method-weak-proxy-bias}).
The population-bias bound in this paper is sharp and cannot generally be further reduced.
Consequently, weak-proxy strength is only required to decay more slowly than the inverse square root of the sample size, a condition of the same order as that required by the bridge estimation error bounds themselves.
\par
Second, this paper complements the existing proximal inference literature \citep{ying2022survival,wu2024doubly} by proving that, even when outcome thresholds vary continuously, the number of distinct estimated outcome bridges is at most one greater than the number of observations in the training sample (Lemma~\ref{lem:app-bridge-class-sizemain}).
Consequently, the conditional entropy of the estimated outcome bridge class is at most of the same order as the logarithm of the number of observations in the training sample, so the complexity cost introduced by estimated outcome bridges is controlled.
This paper establishes a uniform expansion (Theorem~\ref{thm:uniform-dr}), showing that, under the stated conditions, bridge estimation does not alter the leading sampling behavior of the marginal counterfactual distribution.
\par
Finally, this paper establishes simultaneous coverage of simultaneous CDF bands across treatment arms and outcome thresholds (Theorem~\ref{thm:method-simultaneous-inference}).
The contribution of this paper lies in ensuring the validity of the underlying simultaneous CDF bands under weak proxies and estimated bridges, thereby extending existing proximal inference that primarily focuses on mean effects and counterfactual survival curves \citep{cui2024semiparametric,ying2022survival} to simultaneous inference for the quantile treatment effect.
The resulting quantile treatment effect bands admit deterministic band-width bounds and require neither counterfactual density estimation nor conditions on the counterfactual density.
\par
The numerical results support the three contributions. In Figure~\ref{fig:simulation1-bridges-bias}, the true population bias is 11.00, the Bound underlying Corollary~\ref{cor:method-weak-proxy-bias} is 20.97, and the baseline Cauchy--Schwarz bound is 571.75; the baseline overestimates it by a factor of approximately 50.
In Figure~\ref{fig:simulation2}, as the sample size increases, the median of the errors other than population bias decreases progressively in every setting of proxies, falling by approximately half in the strongest proxies setting considered.
Across the 12 settings in Table~\ref{tab:simulation3-qte}, the simultaneous coverage of Proposed is 99.8\%--100.0\%, whereas that of Estimated-density delta is 68.4\%--89.2\%, below the nominal 95\% in all settings; the mean band width of Proposed is only 0.06--0.27 wider than that of Oracle 95\% and decreases as the sample size increases.

\section{Background and Notation}
\label{sec:notation}
Consider an observational study with $n$ independent units. Let $P_0$ denote the
true distribution of the observed data. For a generic unit, write
$O=(Y,A,Z,W,X)\sim P_0$, and let $O_1,\ldots,O_n\overset{\mathrm{iid}}{\sim}P_0$. Here, $Y\in\mathbb{R}$ is the outcome, $A\in\{0,1\}$ is a binary observed treatment indicator, and $X$ collects observed covariates. We allow treatment assignment and the outcome to be confounded by an unobserved latent variable $U$. Both $Z$ and $W$ serve as proxies that indirectly capture information about $U$. Specifically, $Z$ is a treatment-inducing proxy and $W$ is an outcome-inducing proxy.

For each treatment arm $a\in\{0,1\}$, let $Y(a)$ denote the potential outcome under treatment $a$. For the unobserved latent variable $U$, we assume $Y(a)\perp A\mid U,X$, and $Y\perp Z\mid U,A,X$, and $W\perp(A,Z)\mid U,X$. We also assume latent positivity $0<P(A=a\mid U,X)<1$.

Our object is the marginal counterfactual distribution
$F_a(y)=P\{Y(a)\le y\}$, from which the counterfactual quantile function and the quantile treatment effect are $Q_a(\tau)=\inf\{y\in\mathbb{R}:F_a(y)\ge\tau\}$ and
$\Delta_Q(\tau)=Q_1(\tau)-Q_0(\tau)$.

\subsection{Proximal Identification}
\label{sec:background-identification}

\noindent\textit{The outcome bridge $h_{a,y}(W,X)$.} Because $U$ is unobserved,
to identify $F_a(y)$ from the observed data, we introduce an outcome
bridge $h_{a,y}(W,X)$. For each $a\in\{0,1\}$ and $y\in\mathbb{R}$, suppose
there exists a measurable function $h_{a,y}(W,X)$ satisfying
\begin{equation}
E\bigl[\mathbf{1}\{Y\le y\}-h_{a,y}(W,X)\mid A=a,Z,X\bigr]=0.
\label{eq:outcome-bridge}
\end{equation}

Under these assumptions, the marginal counterfactual distribution $F_a(y)$ is
identified as
\begin{equation}
F_a(y)=P\{Y(a)\le y\}=E\{h_{a,y}(W,X)\}.
\label{eq:identification-Fa}
\end{equation}
Thus, the outcome bridge $h_{a,y}(W,X)$ uses the outcome-inducing proxy $W$ to recover the outcome information that would otherwise be unavailable because $U$ is unobserved \citep{xu2021deep}.

\noindent\textit{The treatment bridge $q_a(Z,X)$.} For each $a\in\{0,1\}$, a
treatment bridge is a measurable function $q_a(Z,X)$ satisfying
$E\{\mathbf{1}\{A=a\}q_a^2(Z,X)\}<\infty$ and
\begin{equation}
E\bigl[\mathbf{1}\{A=a\}q_a(Z,X)-1\mid W,X\bigr]=0.
\label{eq:treatment-bridge}
\end{equation}
When a treatment bridge $q_a(Z,X)$ exists, the marginal counterfactual
distribution $F_a(y)$ also satisfies
\[
F_a(y)=E\bigl[\mathbf{1}\{A=a\}q_a(Z,X)\mathbf{1}\{Y\le y\}\bigr].
\]
The same treatment bridge $q_a(Z,X)$ can be used for every threshold $y$.
Combining the equations for the outcome bridge $h_{a,y}(W,X)$ and the treatment
bridge $q_a(Z,X)$ gives
\begin{equation}\label{eq:double}
F_a(y)=E\bigl[h_{a,y}(W,X)+\mathbf{1}\{A=a\}q_a(Z,X)
\times\bigl\{\mathbf{1}\{Y\le y\}-h_{a,y}(W,X)\bigr\}\bigr].
\end{equation}
Thus, proximal identification replaces direct adjustment for the unobserved
latent variable $U$ with the outcome bridge $h_{a,y}(W,X)$ and the treatment
bridge $q_a(Z,X)$ using observed $Z$ and $W$.

\subsection{Related Work and Challenges}
\label{sec:related-work-challenges}
Existing studies have shown that, under appropriate conditions, observed proxies can be used to identify causal effects in the presence of an unobserved latent variable $U$, and that outcome bridges and treatment bridges can be used for estimation \citep{miao2018identifying,cui2024semiparametric}. Subsequent work pointed out that identification of causal effects does not necessarily imply that these causal effects can be estimated stably, and that the existence of the treatment bridge is closely related to such estimation stability \citep{zhang2023nonunique}. Related studies have also shown when bridge equations are difficult to solve, regularization and bridge estimation error can directly affect inference for the final target \citep{bennett2023source}. Meanwhile, existing proximal causal inference methods have studied estimation and inference for entire counterfactual survival curves \citep{ying2022survival}.

Building on the above literature, we study three challenges in simultaneous inference when the outcome bridge $h_{a,y}(W,X)$ and treatment bridge $q_a(Z,X)$ are estimated from data:

\noindent\textit{Challenge 1: Bridge estimation errors under weak proxies.} Weak proxies $Z$ and $W$ amplify sampling errors in $\widehat h_{a,y}(W,X)$ and $\widehat q_a(Z,X)$ \citep{kompa2022deep}. Multiplying separate bridge error bounds can overstate the population bias of $\widehat F_a(y)$ because both bridges use the same sample matrix.  The challenge is to \textit{quantify} the bridge estimation errors and determine when using $\widehat h_{a,y}(W,X)$ and $\widehat q_a(Z,X)$ introduces only negligible population bias into the counterfactual CDF estimate $\widehat F_a(y)$.

\noindent\textit{Challenge 2: Keeping bridge estimation errors small across treatment arms and outcome thresholds.} The counterfactual CDF estimate $\widehat F_a(y)$ combines the estimated outcome bridge $\widehat h_{a,y}(W,X)$ and the estimated treatment bridge $\widehat q_a(Z,X)$ (see equation~\eqref{eq:double}). $\widehat q_a(Z,X)$ is the same for all outcome thresholds $y$, whereas $\widehat h_{a,y}(W,X)$ changes with $y$. Bridge estimation can create additional sampling variation that must be controlled over the entire outcome threshold $y$ \citep[cf.][]{byambadalai2024estimating}. The challenge is therefore to \textit{understand} how the estimated bridges vary with $y$ and $a\in\{0,1\}$, while controlling both the bridge estimation errors and how large the estimated bridges can become.

\noindent\textit{Challenge 3: Reliable confidence bands with estimated bridges.} Errors in the estimated bridges can affect both the CDF estimate $\widehat F_a(y)$ and the width of its confidence band. The challenge is to show that the confidence bands still have the desired coverage simultaneously for both treatment arms $a\in\{0,1\}$ and all outcome thresholds $y$. The usual method for the counterfactual quantile $Q_a(\tau)$ requires additional smoothness and positivity conditions on the counterfactual density near $Q_a(\tau)$ \citep{chernozhukov2013inference}. The challenge is to \textit{construct} simultaneous confidence bands for the quantile treatment effect $\Delta_Q(\tau)$ without introducing density estimation or density conditions.

\section{Proximal Estimation and Simultaneous Inference}
\label{sec:methodology}
\textit{To address Challenge~1}, we quantify weak-proxy strength in the chosen basis spaces by $\kappa_{a,n}$. Proposition~\ref{prop:method-bridge-rates} bounds the estimation errors of the outcome bridge $h_{a,y}(W,X)$ and treatment bridge $q_a(Z,X)$. Using the bridge moment equations rather than multiplying separate bridge error bounds, Corollary~\ref{cor:method-weak-proxy-bias} establishes the sharper population-bias bound uniformly over $y\in\mathcal Y$ and gives explicit sufficient conditions for negligible population bias at the $n^{-1/2}$ scale. Corollary~\ref{cor:method-weak-proxy-bias} further shows that the sharpness of $\kappa_{a,n}^{-1}$. We complement existing proximal doubly robust \citep{cui2024semiparametric,bennett2023source} by quantifying the population bias of $\widehat F_a(y)$ across outcome thresholds $y$ and establishing the sharp $\kappa_{a,n}^{-1}$ dependence for the specified bridge estimators.

\textit{To address Challenge~2}, we exploit the shared estimating structure of the outcome bridges across outcome thresholds. For a training sample of size $N_k$, we prove that the estimated outcome bridges across outcome thresholds $y$ contains at most $N_k+1$ distinct functions in Lemma~\ref{lem:app-bridge-class-sizemain}. Theorem~\ref{thm:uniform-dr} then shows that estimating the bridges does not change the leading sampling behavior of $\widehat F_a(y)$. Compared with existing proximal inference literature \citep{ying2022survival}, our result gives an explicit control of the estimated bridge across outcome thresholds $y$ for $\widehat h_{a,y}(W,X)$ and $\widehat q_a(Z,X)$.

\textit{To address Challenge~3}, Theorem~\ref{thm:method-simultaneous-inference} establishes CDF coverage converging to $1-\alpha_{\mathrm{sig}}$ simultaneously over both treatment arms $a\in\{0,1\}$ and all outcome thresholds $y$. Existing proximal inference has focused primarily on mean effects and counterfactual survival curves \citep{cui2024semiparametric,ying2022survival}. We also extend this line of work to simultaneous inference for the quantile treatment effect $\Delta_Q(\tau)$ in Theorem~\ref{thm:method-simultaneous-inference} and derive deterministic band-width bounds for it. Unlike standard method for $\Delta_Q(\tau)$  \citep{chernozhukov2013inference}, require neither estimation of the counterfactual density nor additional density continuity conditions.

Estimating $F_a(y)$ over an entire outcome threshold range requires estimating the outcome bridge $h_{a,y}(W,X)$ for each treatment arm $a\in\{0,1\}$, every outcome threshold $y$ and estimating the treatment bridge $q_a(Z,X)$ for each treatment arm $a\in\{0,1\}$. Under weak proxies $W$ and $Z$, directly solving sample versions of \eqref{eq:outcome-bridge} and \eqref{eq:treatment-bridge} can amplify sampling errors \citep{kallus2021causal}.

Let $\mathcal Y=[y_L,y_U]$ be a fixed compact threshold interval. Independently of $O_1,\ldots,O_n$, partition $\{1,\ldots,n\}$ into a fixed number $K\ge2$ of disjoint folds $I_1,\ldots,I_K$, with $|I_k|\asymp n$. Define $N_k=n-|I_k|$ and $P_{n,-k}f=N_k^{-1}\sum_{i\notin I_k}f(O_i)$. Repeating this procedure over \(k=1,\ldots,K\) yields cross-fitted bridge evaluations for all observations \(O_1,\ldots,O_n\), with each \(O_i\in I_k\) evaluated using \(\widehat h_{a,y}^{(-k)}\) and \(\widehat q_a^{(-k)}\) estimated without the observations in \(I_k\).

\subsection{Regularized Bridge Estimation and Weak-Proxy Strength}
\label{sec:method-bridge-estimation}

\noindent\textit{Finite-dimensional bridge approximation.} The bridge equations \eqref{eq:outcome-bridge} and \eqref{eq:treatment-bridge} characterize the unknown functions $h_{a,y}(W,X)$ and $q_a(Z,X)$ but do not directly provide estimators for them. We therefore approximate the two bridges $h_{a,y}(W,X)$ and $q_a(Z,X)$ by $b_W(W,X)^\top\theta$ and $b_Z(Z,X)^\top\alpha$, where the prespecified basis vectors $b_W(W,X),b_Z(Z,X)\in\mathbb R^{d_n}$ may have dimension $d_n$ increasing with $n$. $b_W(W,X)$ and $b_Z(Z,X)$ can be constructed by
evaluating prespecified polynomial or 
including interaction terms, at $(W,X)$ and $(Z,X)$,
respectively \citep{tchetgen2024introduction}. Estimating the bridges $h_{a,y}(W,X)$ and $q_a(Z,X)$ then reduces to estimating the coefficient vectors $\theta$ and $\alpha$.

Substituting $h_{a,y}(W,X)\approx b_W(W,X)^\top\theta$ and $q_a(Z,X)\approx b_Z(Z,X)^\top\alpha$ into \eqref{eq:outcome-bridge} and \eqref{eq:treatment-bridge} gives the sample moment conditions for the outcome and treatment bridges,
\begin{equation}
P_{n,-k}\bigl[\mathbf{1}\{A=a\}b_Z(Z,X)\{\mathbf{1}\{Y\le y\}-b_W(W,X)^\top\theta\}\bigr]=0,
\label{eq:sample-outcome-bridge}
\end{equation}
\begin{equation}
P_{n,-k}\bigl[b_W(W,X)\{\mathbf{1}\{A=a\}b_Z(Z,X)^\top\alpha-1\}\bigr]=0.
\label{eq:sample-treatment-bridge}
\end{equation}
Define
\begin{gather*}
\widehat\Sigma_a^{(-k)}=P_{n,-k}\{\mathbf{1}\{A=a\}b_Z(Z,X)b_W(W,X)^\top\},\qquad
\widehat\gamma_a^{(-k)}(y)=P_{n,-k}\{\mathbf{1}\{A=a\}b_Z(Z,X)\mathbf{1}\{Y\le y\}\},\\
\widehat\mu_W^{(-k)}=P_{n,-k}\{b_W(W,X)\},\qquad
(\widehat\Sigma_a^{(-k)})^\top=P_{n,-k}\{\mathbf{1}\{A=a\}b_W(W,X)b_Z(Z,X)^\top\}.
\end{gather*}
Expanding \eqref{eq:sample-outcome-bridge} and \eqref{eq:sample-treatment-bridge} then gives the \textit{outcome and treatment bridge moment conditions}
\begin{equation}
\widehat\Sigma_a^{(-k)}\theta=\widehat\gamma_a^{(-k)}(y),
\qquad
(\widehat\Sigma_a^{(-k)})^\top\alpha=\widehat\mu_W^{(-k)}.
\label{eq:method-bridge-moments}
\end{equation}
For treatment arm $a\in\{0,1\}$, all outcome thresholds $y$ share $\widehat\Sigma_a^{(-k)}$, only $\widehat\gamma_a^{(-k)}(y)$ changes with $y$.

\noindent\textit{Regularization under weak proxies.} When $\widehat\Sigma_a^{(-k)}$ is nearly singular, directly solving $\widehat\Sigma_a^{(-k)}\theta=\widehat\gamma_a^{(-k)}(y)$ and $(\widehat\Sigma_a^{(-k)})^\top\alpha=\widehat\mu_W^{(-k)}$ can make the coefficient estimates highly sensitive to sampling errors. We therefore choose $\theta$ and $\alpha$ so that $\widehat\Sigma_a^{(-k)}\theta$ is close to $\widehat\gamma_a^{(-k)}(y)$ and $(\widehat\Sigma_a^{(-k)})^\top\alpha$ is close to $\widehat\mu_W^{(-k)}$, while preventing $\theta$ and $\alpha$ from becoming excessively large \citep{carrasco2007linear}.

For $\lambda_h,\lambda_q>0$, which may depend on $n$ but are fixed across outcome thresholds $y$, define
\begin{equation}
\begin{aligned}
\widehat\theta_{a,\lambda_h}^{(-k)}(y)
&=\operatorname*{arg\,min}_{\theta\in\mathbb R^{d_n}}
\left\{
\left\|\widehat\Sigma_a^{(-k)}\theta-\widehat\gamma_a^{(-k)}(y)\right\|_2^2
+\lambda_h\|\theta\|_2^2
\right\},\\
\widehat\alpha_{a,\lambda_q}^{(-k)}
&=\operatorname*{arg\,min}_{\alpha\in\mathbb R^{d_n}}
\left\{
\left\|(\widehat\Sigma_a^{(-k)})^\top\alpha-\widehat\mu_W^{(-k)}\right\|_2^2
+\lambda_q\|\alpha\|_2^2
\right\}.
\end{aligned}
\label{eq:method-regularized-bridges}
\end{equation}
The optimization problems are convex quadratic problems with unique solutions.

The estimated bridges are therefore
\begin{equation}
\begin{aligned}
\widehat h_{a,y}^{(-k)}(W,X)
&=b_W(W,X)^\top\widehat\theta_{a,\lambda_h}^{(-k)}(y),\quad
\widehat q_a^{(-k)}(Z,X)
=b_Z(Z,X)^\top\widehat\alpha_{a,\lambda_q}^{(-k)}.
\end{aligned}
\label{eq:method-estimated-bridges}
\end{equation}

\noindent\textit{Quantifying weak-proxy strength.} We now quantify how sampling errors can affect estimates of $\theta$ and $\alpha$. Replacing the training-sample average $P_{n,-k}$ in $\widehat\Sigma_a^{(-k)}$ by expectation $E$ under $P_0$, define
\begin{equation}
\Sigma_a
=
E\!\left\{\mathbf{1}\{A=a\}b_Z(Z,X)b_W(W,X)^\top\right\}.
\label{eq:method-population-matrix}
\end{equation}
We measure weak-proxy strength in the chosen basis spaces by
\begin{equation}
\kappa_{a,n}=\sigma_{\min}(\Sigma_a)>0,
\label{eq:method-weak-proxy-strength}
\end{equation}
where $\sigma_{\min}(\Sigma_a)$ denotes the smallest singular value of $\Sigma_a$. A small $\kappa_{a,n}$ means that substantially different values of $\theta$ can produce nearly identical values of $\Sigma_a\theta$, making $\theta$ difficult to distinguish using noisy sample moments. Because $\Sigma_a^\top$ has the same smallest singular value as $\Sigma_a$, a small $\kappa_{a,n}$ can also make substantially different values of $\alpha$ difficult to distinguish through $\Sigma_a^\top\alpha$. Proposition~\ref{prop:method-bridge-rates} quantifies how $\kappa_{a,n}$ affects the estimation error bounds for $h_{a,y}(W,X)$ and $q_a(Z,X)$.

\begin{proposition}
\label{prop:method-bridge-rates}
Under Assumption~\ref{ass:finite-rank}. For every $\lambda_h,\lambda_q>0$,
\begin{equation}
\begin{aligned}
r_{h,a,n}=\max_k\sup_{y\in\mathcal Y}
\left\|\widehat h_{a,y}^{(-k)}-h_{a,y}\right\|_W
&=O_p\!\left(
\frac{1+M_{h,a,n}}{\sqrt{\kappa_{a,n}^2+\lambda_h}}\sqrt{\frac{d_n}{n}}
+b_{h,a,n}
+\frac{\lambda_hM_{h,a,n}}{\kappa_{a,n}^2+\lambda_h}
\right),\\
r_{q,a,n}=\max_k\left\|\widehat q_a^{(-k)}-q_a\right\|_{a,Z}
&=O_p\!\left(
\frac{1+M_{q,a,n}}{\sqrt{\kappa_{a,n}^2+\lambda_q}}\sqrt{\frac{d_n}{n}}
+b_{q,a,n}
+\frac{\lambda_qM_{q,a,n}}{\kappa_{a,n}^2+\lambda_q}
\right).
\end{aligned}
\label{eq:method-uniform-bridge-rates}
\end{equation}
The norms $M_{h,a,n}$, $M_{q,a,n}$, approximation errors, and coefficient norm bounds are defined in Appendix~\ref{app:bridge-error-definitions}. Appendix~\ref{app:bridge-rates-proof} gives the exact factors involving the sample singular values.
\end{proposition}

\subsection{Controlling $\widehat F_a(y)-F_a(y)$ simultaneously Across treatment arms and Outcome Thresholds}

Motivated by the \(F_a(y)\) in \eqref{eq:double}, we replace \(h_{a,y}(W,X)\) and \(q_a(Z,X)\) by their estimators. Define
\begin{equation}
\begin{aligned}
\widehat F_a(y)
=\frac{1}{n}\sum_{k=1}^K\sum_{i\in I_k}\Big[\,
&\widehat h_{a,y}^{(-k)}(W_i,X_i)+\mathbf{1}\{A_i=a\}\widehat q_a^{(-k)}(Z_i,X_i)
\{\mathbf{1}\{Y_i\le y\}-\widehat h_{a,y}^{(-k)}(W_i,X_i)\}\Big].
\end{aligned}
\label{eq:method-cdf-estimator}
\end{equation}

\noindent\textit{How bridge estimation errors enter $\widehat F_a(y)$.} For each fold $I_k$, expectations containing estimated bridges integrate over $O\sim P_0$ with the training sample outside $I_k$ held fixed. The outcome bridge equation \eqref{eq:outcome-bridge}, the treatment bridge equation \eqref{eq:treatment-bridge}, and \eqref{eq:method-cdf-estimator} give
\begin{equation}
\begin{aligned}
&E\Big[\widehat h_{a,y}^{(-k)}(W,X)
+\mathbf{1}\{A=a\}\widehat q_a^{(-k)}(Z,X)\times\{\mathbf{1}\{Y\le y\}-\widehat h_{a,y}^{(-k)}(W,X)\}\Big]-F_a(y)\\
&\quad=-E\Big[\mathbf{1}\{A=a\}
\{\widehat q_a^{(-k)}(Z,X)-q_a(Z,X)\}
\{\widehat h_{a,y}^{(-k)}(W,X)-h_{a,y}(W,X)\}\Big].
\end{aligned}
\label{eq:method-bridge-product-bias}
\end{equation}
Thus the bias of \eqref{eq:method-bridge-product-bias} depends on the product of the two bridge estimation errors rather than on either bridge estimation error alone \citep{bozkurt2025density}. Recall Proposition~\ref{prop:method-bridge-rates}, by Cauchy--Schwarz, the absolute population bias in \eqref{eq:method-bridge-product-bias} is at most $r_{h,a,n}r_{q,a,n}$.
The Cauchy--Schwarz bound 
\begin{equation}
\max_k\sup_{y\in\mathcal Y}
\left|
E\!\left[
\mathbf{1}\{A=a\}
\{\widehat q_a^{(-k)}-q_a\}
\{\widehat h_{a,y}^{(-k)}-h_{a,y}\}
\right]
\right|
\le r_{h,a,n}r_{q,a,n}.
\label{eq:Cauchy--Schwarz}
\end{equation} incurs the $\kappa_{a,n}^{-1}$ sampling-error factor twice.

\begin{corollary}
\label{cor:method-weak-proxy-bias}
Under the conditions of Proposition~\ref{prop:method-bridge-rates}, suppose $M_{h,a,n},M_{q,a,n}=O(1)$ and
\[
\lambda_hM_{h,a,n}+\lambda_qM_{q,a,n}
=o\!\left(\kappa_{a,n}\sqrt{\frac{d_n}{n}}\right)
\]
for each treatment arm $a\in\{0,1\}$.
Then, for each treatment arm $a\in\{0,1\}$,
\begin{equation}
\begin{aligned}
&\max_k\sup_{y\in\mathcal Y}
\left|E\!\left[\mathbf{1}\{A=a\}
\{\widehat q_a^{(-k)}(Z,X)-q_a(Z,X)\}
\{\widehat h_{a,y}^{(-k)}(W,X)-h_{a,y}(W,X)\}\right]\right|=O_p\!\left(\frac{d_n}{n\kappa_{a,n}}+b_{h,a,n}b_{q,a,n}\right).
\label{eq:method-sharper-bias-rate}
\end{aligned}
\end{equation}
Consequently,
\[
\max_{a\in\{0,1\}}
\left\{\frac{d_n}{\kappa_{a,n}\sqrt n}+\sqrt n\,b_{h,a,n}b_{q,a,n}\right\}\to0
\]
is sufficient for the population bias in \eqref{eq:method-bridge-product-bias} to be negligible at the $n^{-1/2}$ scale.

Even in the ideal case where the true outcome bridge and treatment bridge can be exactly represented by the finite-dimensional bases $b_W(W,X)$ and $b_Z(Z,X)$, the population bias can still satisfy
\begin{equation}
\max_k\sup_{y\in\mathcal Y}
\left|E\!\left[\mathbf{1}\{A=1\}
\{\widehat q_1^{(-k)}(Z,X)-q_1(Z,X)\}
\{\widehat h_{1,y}^{(-k)}(W,X)-h_{1,y}(W,X)\}\right]\right|
\neq o_p\!\left(\frac{1}{n\kappa_{1,n}}\right).
\label{eq:method-bias-sharpness}
\end{equation}
\end{corollary}

Appendix~\ref{app:population-bias} gives the proof. Compared with the Cauchy--Schwarz bound, whose leading sampling term scales as $O_p\{d_n/(n\kappa_{a,n}^2)\}$, Corollary~\ref{cor:method-weak-proxy-bias} improves this dependence to $O_p\{d_n/(n\kappa_{a,n}) + b_{h,a,n}b_{q,a,n}\}$, and \eqref{eq:method-bias-sharpness} shows that the remaining $\kappa_{a,n}^{-1}$ is sharp, and cannot generally be improved in order for the specified bridge estimators.

\noindent\textit{Controlling bridge estimation errors across outcome thresholds $y$.}
Corollary~\ref{cor:method-weak-proxy-bias} controls population bias of $\widehat F_a(y)$, but simultaneous inference also requires control of the sampling fluctuation caused by estimated bridges across $y\in\mathcal Y$. For each treatment arm $a\in\{0,1\}$, the treatment bridge estimate $\widehat q_a^{(-k)}(Z,X)$ is fixed across outcome thresholds $y$. The outcome bridge estimate $\widehat h_{a,y}^{(-k)}(W,X)$ varies with $y$, so we examine how many distinct functions occur in $\{\widehat h_{a,y}^{(-k)}(W,X):y\in\mathcal Y\}$.
\begin{lemma}
\label{lem:app-bridge-class-sizemain}
Fix a treatment arm $a\in\{0,1\}$ and a fold $I_k$. Conditionally on the training sample $\{O_i:i\notin I_k\}$, the class
\[
\{\widehat h_{a,y}^{(-k)}(W,X):y\in\mathcal Y\}
\]
contains at most $N_k+1$ distinct functions.
\end{lemma}
Appendix~\ref{app:bridge-class-size} gives the proof. Thus, although $y$ ranges over a continuum, the estimated outcome bridge class has conditional entropy of order $\log N_k$, providing an explicit simple complexity control for simultaneous inference with estimated bridges below.

\noindent\textit{Control of $\widehat F_a(y)-F_a(y)$ across both treatment arms and all outcome thresholds.} To describe the estimation error that would remain if true $h_{a,y}(W,X)$ and $q_a(Z,X)$ were known, define
\begin{equation}
\begin{aligned}
\varphi_{a,y}(O)
={}&h_{a,y}(W,X)-F_a(y)
+\mathbf{1}\{A=a\}q_a(Z,X)
\{\mathbf{1}\{Y\le y\}-h_{a,y}(W,X)\}.
\end{aligned}
\label{eq:method-influence-function}
\end{equation}
Equation~\eqref{eq:double} gives $E\{\varphi_{a,y}(O)\}=0$.

\begin{theorem}
\label{thm:uniform-dr}
Assume Assumptions~\ref{ass:process} and \ref{ass:nuisance} hold. Corollary~\ref{cor:method-weak-proxy-bias} gives sufficient population-bias conditions, and Lemma~\ref{lem:app-bridge-class-sizemain} verifies the estimated outcome bridge covering bound, these results give
\begin{equation}
\begin{aligned}
\max_{a\in\{0,1\}}\sup_{y\in\mathcal Y}
\left|
\sqrt n\{\widehat F_a(y)-F_a(y)\}
-\frac{1}{\sqrt n}\sum_{i=1}^n\varphi_{a,y}(O_i)
\right|
=o_p(1).
\end{aligned}
\label{eq:method-uniform-expansion}
\end{equation}
\end{theorem}
Appendix~\ref{app:uniform-expansion-proof} gives the proof. Theorem~\ref{thm:uniform-dr} connects the bridge estimation error bounds in Proposition~\ref{prop:method-bridge-rates} to simultaneous inference over treatment arms $a\in\{0,1\}$ and outcome thresholds $y\in\mathcal Y$.

\subsection{Simultaneous Confidence Bands for $F_a(y)$ and $\Delta_Q(\tau)$}
\label{sec:method-simultaneous-inference}
Theorem~\ref{thm:uniform-dr} describes $\widehat F_a(y)-F_a(y)$ through $\varphi_{a,y}(O)$. Because $\varphi_{a,y}(O)$ contains unknown bridges, the critical value must remain valid after bridge estimation. The usual method for constructing confidence bands for the quantile treatment effect $\Delta_Q(\tau)$ requires estimating the counterfactual density at $Q_a(\tau)$ for both treatment arms \citep{chernozhukov2013inference}. The counterfactual density is the density of the potential outcome $Y(a)$ and describes how quickly the marginal counterfactual distribution $F_a(y)$ increases with $y$. Estimating this density adds an additional estimation step, and can be unreliable when the counterfactual density near $Q_a(\tau)$ is close to zero or does not satisfy the required smoothness conditions \citep{firpo2007efficient}.

We instead construct simultaneous CDF bands using Theorem~\ref{thm:uniform-dr}, validating the multiplier critical value after estimating the outcome bridge $h_{a,y}(W,X)$ and treatment bridge $q_a(Z,X)$ and replacing $F_a(y)$ by $\widehat F_a(y)$. Confidence-band inversion then gives simultaneous confidence bands for $\Delta_Q(\tau)$ and bounds the QTE band width by changes in the counterfactual quantile function $Q_a(\tau)$ without counterfactual density estimation.

\noindent\textit{Calculating the critical value from estimated bridges.} To calculate the confidence-band width, we first estimate $\varphi_{a,y}(O_i)$ for each $i\in I_k$. In \eqref{eq:method-influence-function}, replace $h_{a,y}(W_i,X_i)$, $q_a(Z_i,X_i)$, and $F_a(y)$ by $\widehat h_{a,y}^{(-k)}(W_i,X_i)$, $\widehat q_a^{(-k)}(Z_i,X_i)$, and $\widehat F_a(y)$, respectively, to define
\begin{equation}
\begin{aligned}
\widehat\varphi_{a,y}^{(-k)}(O_i)
={}&\widehat h_{a,y}^{(-k)}(W_i,X_i)-\widehat F_a(y)+\mathbf{1}\{A_i=a\}\widehat q_a^{(-k)}(Z_i,X_i)
\{\mathbf{1}\{Y_i\le y\}-\widehat h_{a,y}^{(-k)}(W_i,X_i)\}.
\end{aligned}
\label{eq:method-estimated-influence-function}
\end{equation}

We approximate the simultaneous distribution of $\sqrt n\{\widehat F_a(y)-F_a(y)\}$ over both treatment arms $a\in\{0,1\}$ and outcome thresholds $y\in\mathcal Y$ using a Rademacher multiplier bootstrap \citep{vandervaartwellner1996}. Let $\xi_1,\ldots,\xi_n$ be independent Rademacher multipliers, with $P(\xi_i=1)=P(\xi_i=-1)=1/2$, and define
\begin{equation}
\widehat{\mathbb G}_a^\xi(y)
=
\frac{1}{\sqrt n}\sum_{k=1}^K\sum_{i\in I_k}
\xi_i\widehat\varphi_{a,y}^{(-k)}(O_i).
\label{eq:method-multiplier-process}
\end{equation}
We estimate $\widehat c_{1-\alpha_{\mathrm{sig}}}$ by the conditional $(1-\alpha_{\mathrm{sig}})$ quantile of
\begin{equation}
\max_{a\in\{0,1\}}\sup_{y\in\mathcal Y}
|\widehat{\mathbb G}_a^\xi(y)|
\label{eq:method-multiplier-maximum}
\end{equation}
over repeated multiplier draws.

Using $\widehat c_{1-\alpha_{\mathrm{sig}}}/\sqrt n$ as the half-width around $\widehat F_a(y)$, define
\begin{equation}
\begin{aligned}
L_a(y)=\max\{0,\widehat F_a(y)-\widehat c_{1-\alpha_{\mathrm{sig}}}/\sqrt n\},\quad
U_a(y)=\min\{1,\widehat F_a(y)+\widehat c_{1-\alpha_{\mathrm{sig}}}/\sqrt n\}.
\end{aligned}
\label{eq:method-cdf-bands}
\end{equation}

\noindent\textit{Confidence bands for $\Delta_Q(\tau)$ without density estimation.} Let $\mathcal T\subset(0,1)$ be a fixed compact set of quantile levels. We apply the confidence-band inversion method of \citet{chernozhukov2020generic} to $[L_a(y),U_a(y)]$. The upper CDF bound $U_a(y)$ determines a lower quantile bound, and the lower CDF bound $L_a(y)$ determines an upper quantile bound:
\begin{equation}
\begin{aligned}
Q_a^L(\tau)=\inf\{y\in\mathcal Y:U_a(y)\ge\tau\},\quad
Q_a^U(\tau)=\inf\{y\in\mathcal Y:L_a(y)\ge\tau\},
\end{aligned}
\label{eq:method-quantile-bands}
\end{equation}
If $L_a(y)\le F_a(y)\le U_a(y)$ for both treatment arms $a\in\{0,1\}$ and every outcome thresholds $y\in\mathcal Y$, then $Q_a^L(\tau)\le Q_a(\tau)\le Q_a^U(\tau)$ for every target quantile $Q_a(\tau)$ in the interior of $\mathcal Y$.

Since $\Delta_Q(\tau)=Q_1(\tau)-Q_0(\tau)$, define
\begin{equation}
\begin{aligned}
\Delta_Q^L(\tau)=Q_1^L(\tau)-Q_0^U(\tau),\quad
\Delta_Q^U(\tau)=Q_1^U(\tau)-Q_0^L(\tau).
\end{aligned}
\label{eq:method-qte-bands}
\end{equation}
Coverage of $F_0(y)$ and $F_1(y)$ over $y\in\mathcal Y$ therefore implies coverage of $\Delta_Q(\tau)$ simultaneously over $\tau\in\mathcal T$. Thus, the construction in \eqref{eq:method-quantile-bands} and \eqref{eq:method-qte-bands} uses only $L_a(y)$ and $U_a(y)$, without estimating a counterfactual density or differentiating $Q_a(\tau)$.

\begin{theorem}
\label{thm:method-simultaneous-inference}
Under the conditions of Theorem~\ref{thm:uniform-dr}, suppose
\begin{equation}
\max_{a\in\{0,1\}}\sup_{y\in\mathcal Y}E\{\varphi_{a,y}^2(O)\}>0
\label{eq:method-nondegenerate-variance}
\end{equation}
and $Q_a(\tau)\in(y_L,y_U)$ for every $a\in\{0,1\}$ and $\tau\in\mathcal T$.
Then:
\begin{enumerate}[label=(\roman*),leftmargin=*]
\item The CDF bands satisfy
\begin{equation}
P_0\!\left\{
L_a(y)\le F_a(y)\le U_a(y)
\text{ for all }a\in\{0,1\},\ y\in\mathcal Y
\right\}
\longrightarrow1-\alpha_{\mathrm{sig}}.
\label{eq:method-cdf-coverage}
\end{equation}
\item The QTE bands satisfy
\begin{equation}
\liminf_{n\to\infty}P_0\!\left\{
\Delta_Q^L(\tau)\le\Delta_Q(\tau)\le\Delta_Q^U(\tau)
\text{ for all }\tau\in\mathcal T
\right\}
\ge1-\alpha_{\mathrm{sig}}.
\label{eq:method-qte-coverage}
\end{equation}
\item On the simultaneous CDF-coverage event in \eqref{eq:method-cdf-coverage}, for both treatment arms,
\begin{equation}
\Delta_Q^U(\tau)-\Delta_Q^L(\tau)
\le\sum_{a=0}^1\left[
Q_a\!\left(\tau+\frac{2\widehat c_{1-\alpha_{\mathrm{sig}}}}{\sqrt n}\right)
-Q_a\!\left(\tau-\frac{2\widehat c_{1-\alpha_{\mathrm{sig}}}}{\sqrt n}\right)
\right]
\label{eq:method-qte-bound}
\end{equation}
for every $\tau\in\mathcal T$ with $\tau\pm2\widehat c_{1-\alpha_{\mathrm{sig}}}/\sqrt n\in(0,1)$ and $Q_a(\tau\pm2\widehat c_{1-\alpha_{\mathrm{sig}}}/\sqrt n)\in(y_L,y_U)$.
\end{enumerate}
\end{theorem}
Appendix~\ref{app:c-coverage-proof} gives the proof. Equation~\eqref{eq:method-qte-coverage} establishes simultaneous coverage of $\Delta_Q(\tau)$,  \eqref{eq:method-qte-bound} bounds the QTE band width, without counterfactual density estimation.

\section{Numerical Studies}
\label{sec:numerical}
We design three simulations to examine bridge estimation errors and population bias under weak proxies (\textit{Challenge~1}), estimated-bridge sampling fluctuation and uniform expansion (\textit{Challenge~2}), and simultaneous confidence bands for $\Delta_Q(\tau)$ without density estimation (\textit{Challenge~3}), respectively.

In all three simulations, the observed covariates $X=(X_1,X_2)$ comprise two independent Bernoulli$(1/2)$ variables. The unobserved latent variable $U$, treatment-inducing proxy $Z$, outcome-inducing proxy $W$, and treatment indicator $A$ are binary and generated from conditional Bernoulli distributions. The coefficient $\rho$ controls $Z$ and $W$ dependence on $U$; the resulting $\kappa_{a,n}$ is calculated from the population matrix in \eqref{eq:method-population-matrix} for the chosen basis. For each treatment arm $a\in\{0,1\}$, we generate the potential outcome $Y(a)$ by adding standardized noise. We use the prespecified basis vectors $b_W(W,X)$ and $b_Z(Z,X)$ with $d_n=8$ and $K=3$ folds. We use 500 replications for each setting. See Appendix~\ref{app:simulation-common} for details.

\noindent\textit{Simulation 1: Bridge Estimation Errors and Population Bias under Weak Proxies.}

We examine \textit{Challenge~1} by connecting Proposition~\ref{prop:method-bridge-rates} with Corollary~\ref{cor:method-weak-proxy-bias}. Further details are provided in Appendix~\ref{app:simulation1-dgp}. Figure~\ref{fig:simulation1-bridges-bias} fixes $n=3200$ and uses $\min_{a\in\{0,1\}}\kappa_{a,n}$ as the horizontal axis in all four panels. The first two panels plot $\max_{a\in\{0,1\}}r_{h,a,n}$ and $\max_{a\in\{0,1\}}r_{q,a,n}$, and the last two compare the $\sqrt n$-scaled population bias and its bounds without and with ridge. Population bias denotes the true absolute population bias in \eqref{eq:method-bridge-product-bias}, maximized over both treatment arms, training folds, and outcome thresholds, and multiplied by $\sqrt n$. Cauchy--Schwarz bound denotes $\sqrt n\max_{a\in\{0,1\}}r_{h,a,n}r_{q,a,n}$ from \eqref{eq:Cauchy--Schwarz}. Bound underlying Corollary~\ref{cor:method-weak-proxy-bias} denotes the right-hand side of \eqref{eq:app-projected-product-bias-bound}, multiplied by $\sqrt n$. The additional bound in \eqref{eq:app-ridge-product-bias-bound} is not plotted.

\begin{figure}[htbp]
    \centering
    \includegraphics[width=\linewidth]{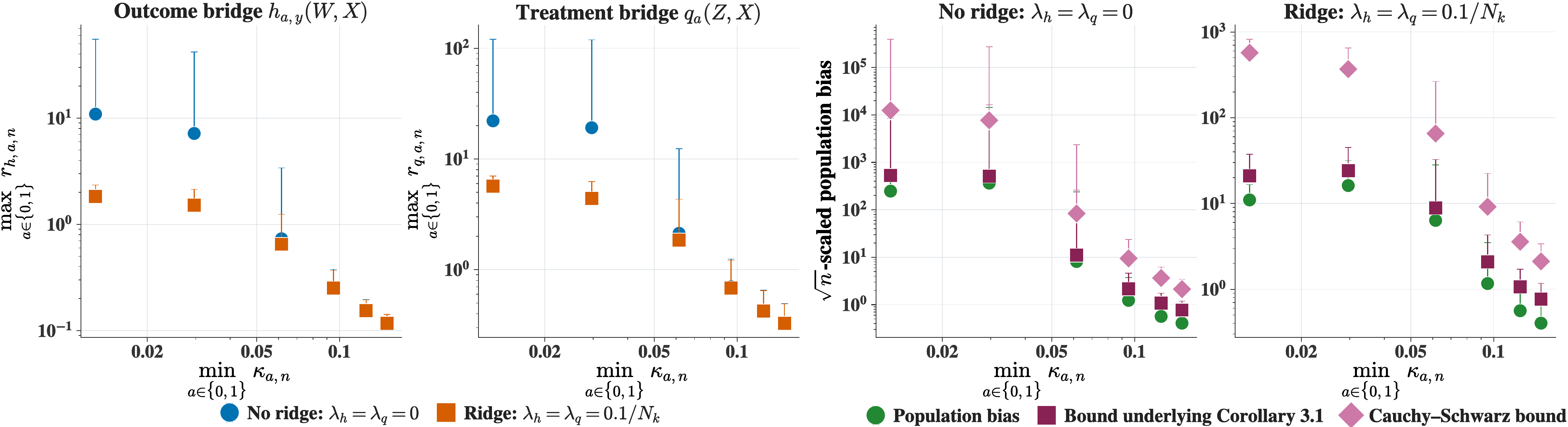}
    \caption{Bridge estimation errors and population bias under weak proxies at $n=3200$. Markers show medians and upper bars show 90th percentiles among successful fits.}
    \label{fig:simulation1-bridges-bias}
\end{figure}

The first two panels of Figure~\ref{fig:simulation1-bridges-bias} show that smaller $\min_{a\in\{0,1\}}\kappa_{a,n}$ is associated with larger bridge estimation errors, qualitatively consistent with Proposition~\ref{prop:method-bridge-rates}. At the weakest-proxy setting $\rho=0.2$, ridge in \eqref{eq:method-regularized-bridges} reduces the median outcome- and treatment-bridge errors from approximately $10.90$ to $1.84$ ($-83\%$) and from $22.09$ to $5.68$ ($-74\%$), and their 90th percentiles from $55.12$ to $2.35$ and from $120.29$ to $7.00$. In the last two panels, the median population bias with ridge ($11.00$) stays close to the bound underlying Corollary~\ref{cor:method-weak-proxy-bias} ($20.97$) but far below the Cauchy--Schwarz bound ($571.75$), supporting its sharper $\kappa_{a,n}^{-1}$ population-bias control; without ridge, the median population bias rises to $247.61$, so ridge reduces it by about $96\%$.

\noindent\textit{Simulation 2: Estimated-Bridge Sampling Fluctuation and Uniform Expansion.}

We examine \textit{Challenge~2} by assessing Theorem~\ref{thm:uniform-dr}, and Lemma~\ref{lem:app-bridge-class-sizemain}. The estimated-bridge sampling fluctuation $V_n=\sqrt n\max_{a\in\{0,1\}}\sup_{y\in\mathcal Y}\left|\sum_{k=1}^K(n_k/n)(P_{n,k}-E)d_{a,y}^{(-k)}\right|$ measures the sampling-fluctuation component controlled in the proof of Theorem~\ref{thm:uniform-dr} (see \eqref{eq:app-fluctuation-negligible}). The uniform expansion remainder $R_n=\max_{a\in\{0,1\}}\sup_{y\in\mathcal Y}\left|\sqrt n\{\widehat F_a(y)-F_a(y)\}-n^{-1/2}\sum_{i=1}^n\varphi_{a,y}(O_i)\right|$ measures the full remainder in Theorem~\ref{thm:uniform-dr}, including both population bias and sampling fluctuation, and equals the left-hand side of \eqref{eq:method-uniform-expansion}. We use $n\in\{1600,3200,4800,5400\}$ and $\rho\in\{0.75,0.60,0.45\}$. At every value of $\rho$, we compare Walsh, ReLU, and GB-leaf constructions of $b_W(W,X)$ and $b_Z(Z,X)$.
Details of the three basis constructions are provided in Appendix~\ref{app:simulation2-bases}.

\begin{figure}[htbp]
    \centering
    \includegraphics[width=\linewidth]{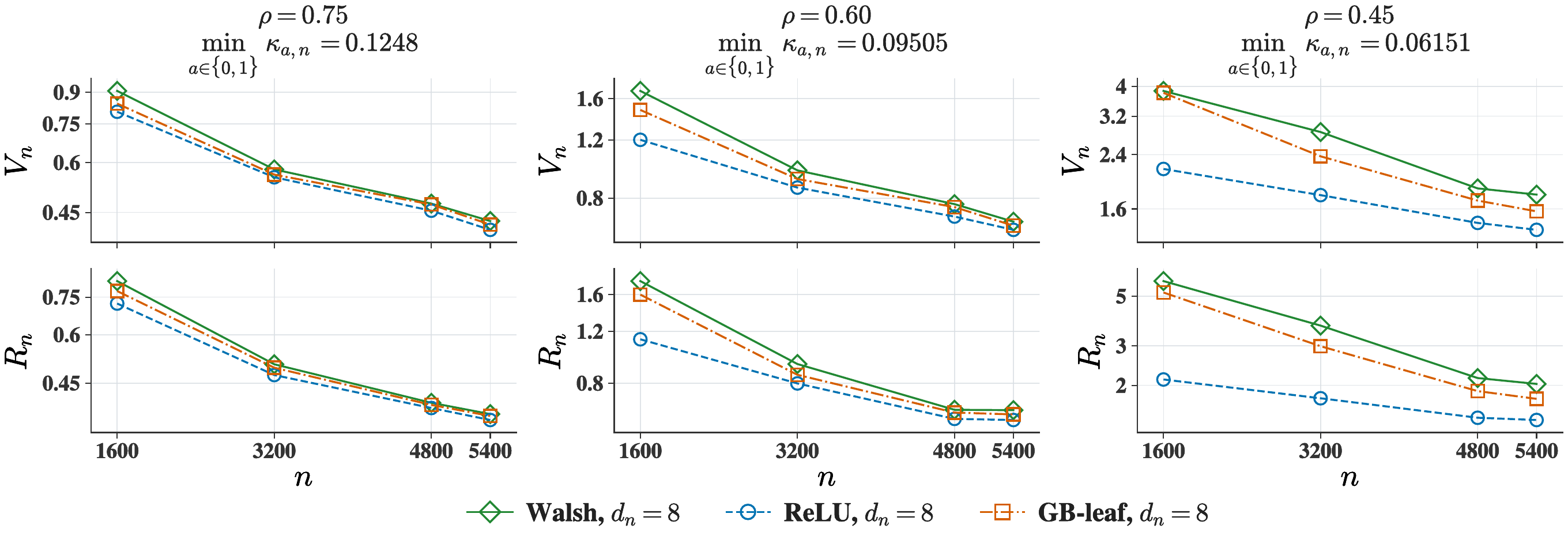}
    \caption{Estimated-bridge sampling fluctuation (top) and uniform expansion remainder (bottom).}
    \label{fig:simulation2}
\end{figure}

At $\rho=0.75$, Figure~\ref{fig:simulation2} shows that increasing $n$ from $1600$ to $5400$ reduces the median $V_n$ from approximately $0.80$--$0.91$ to $0.41$--$0.43$ and the median $R_n$ from $0.72$--$0.82$ to $0.36$--$0.37$ across the three bases, i.e., reductions of $49\%$--$53\%$ and $50\%$--$55\%$; these nearly identical decreases across Walsh, ReLU, and GB-leaf indicate robustness to the choice of basis construction. Within the fixed Walsh basis, $\rho=0.75$, $0.60$, and $0.45$ correspond to $\min_{a\in\{0,1\}}\kappa_{a,n}$ of approximately $0.1248$, $0.0951$, and $0.0615$. Over the same range of $n$, the median $V_n$ and $R_n$ decrease from $1.20$--$1.69$ to $0.64$--$0.68$ and from $1.13$--$1.78$ to $0.60$--$0.65$ at $\rho=0.60$, and from $2.16$--$3.87$ to $1.37$--$1.78$ and from $2.12$--$5.84$ to $1.40$--$2.03$ at $\rho=0.45$. Both medians decrease at each successive sample size for every basis and proxy setting, consistent with Theorem~\ref{thm:uniform-dr} and Lemma~\ref{lem:app-bridge-class-sizemain} with bounds on bridge estimation errors simultaneously over treatment arms and outcome thresholds.

\noindent\textit{Simulation 3: Simultaneous Confidence Bands for $\Delta_Q(\tau)$ without Density Estimation.}

We examine \textit{Challenge~3} by figuring whether confidence-band inversion maintains simultaneous coverage of the quantile treatment effect $\Delta_Q(\tau)$ when the usual normal approximation for the counterfactual quantile $Q_a(\tau)$ can be unreliable. The noise distribution changed to Double Weibull$(1/2)$, Beta$(1/4,1/4)$, Beta$(1/4,3/4)$, or $\chi^2_1$. Further details are provided in Appendix~\ref{app:simulation3-dgp}. We set $\lambda_h=\lambda_q=0.01/\sqrt{N_k}$ on each training sample, $\alpha_{\mathrm{sig}}=0.05$ and
$\mathcal T=\{0.10,0.12,\ldots,0.90\}$. We compare the simultaneous confidence bands
$[\Delta_Q^L(\tau),\Delta_Q^U(\tau)]$, obtained by confidence-band inversion
of the CDF bands $[L_a(y),U_a(y)]$, with Estimated-density delta, which uses the same counterfactual CDF estimate $\widehat F_a(y)$ as \eqref{eq:method-cdf-estimator} and additionally estimates the counterfactual density.

Table~\ref{tab:simulation3-qte} supports the Proposed method in addressing \textit{Challenge~3}. Across all 12 combinations of noise distribution and sample size, Proposed achieves simultaneous coverage of $99.8\%$--$100.0\%$, whereas Estimated-density delta achieves $68.4\%$--$89.2\%$, below the nominal $95\%$ level in every combination. The coverage advantage is 10.8--31.4 percentage points. Across 12 combinations,
Proposed achieves mean band widths of $0.38$--$1.04$, only exceeding
the corresponding Oracle 95\% widths by $0.06$--$0.27$ (an average increase of approximately $20\%$), while maintaining simultaneous
coverage of $99.8\%$--$100.0\%$.
Its mean band width decreases in every setting as $n$ increases
from $1600$ to $4800$, with the range across noise distributions
decreasing from $0.76$--$1.04$ to $0.38$--$0.57$. These results indicate that the proposed method achieves conservative simultaneous coverage without an excessive increase in mean bandwidth, thereby striking a balance between efficiency and robustness.

\begin{table}[!htbp]
\centering
\caption{Simultaneous quantile treatment effects coverage and mean band width.}
\label{tab:simulation3-qte}
\small
\setlength{\tabcolsep}{3pt}
\renewcommand{\arraystretch}{0.94}
\setlength{\aboverulesep}{0.25ex}
\setlength{\belowrulesep}{0.25ex}
\begin{tabular*}{\linewidth}{@{\extracolsep{\fill}}lr*{5}{r}@{}}
\toprule
 & & \multicolumn{2}{c}{Coverage (\%)}
   & \multicolumn{3}{c}{Mean band width} \\
\cmidrule(lr){3-4}\cmidrule(l){5-7}
Noise & $n$ & \textbf{Proposed} & Est.-density $\delta$
      & \textbf{Proposed} & Est.-density $\delta$ & Oracle 95\% \\
\midrule
Double Weibull$(1/2)$
 & 1600 & \textbf{100.0} & 81.6 & \textbf{0.76} & 0.41 & 0.58 \\
 & 3200 & \textbf{100.0} & 81.4 & \textbf{0.47} & 0.26 & 0.39 \\
 & 4800 & \textbf{100.0} & 74.6 & \textbf{0.38} & 0.21 & 0.32 \\
\addlinespace[2pt]
Beta$(1/4,1/4)$
 & 1600 & \textbf{100.0} & 80.4 & \textbf{1.04} & 0.58 & 0.77 \\
 & 3200 & \textbf{100.0} & 80.2 & \textbf{0.70} & 0.39 & 0.54 \\
 & 4800 & \textbf{100.0} & 75.2 & \textbf{0.57} & 0.31 & 0.46 \\
\addlinespace[2pt]
Beta$(1/4,3/4)$
 & 1600 & \textbf{99.8}  & 80.2 & \textbf{1.00} & 0.55 & 0.81 \\
 & 3200 & \textbf{100.0} & 68.8 & \textbf{0.66} & 0.36 & 0.54 \\
 & 4800 & \textbf{99.8}  & 68.4 & \textbf{0.53} & 0.29 & 0.45 \\
\addlinespace[2pt]
$\chi^2_1$
 & 1600 & \textbf{100.0} & 89.2 & \textbf{0.91} & 0.48 & 0.73 \\
 & 3200 & \textbf{100.0} & 86.6 & \textbf{0.58} & 0.31 & 0.49 \\
 & 4800 & \textbf{100.0} & 85.4 & \textbf{0.47} & 0.25 & 0.41 \\
\bottomrule
\end{tabular*}
\begin{minipage}{\linewidth}
\footnotesize\raggedright
\textit{Notes:} Est.-density $\delta$ denotes Estimated-density delta.
\end{minipage}
\end{table}

\section{Discussion}
\label{sec:discussion}
We develop simultaneous inference for the marginal counterfactual distribution $F_a(y)$ and the quantile treatment effect $\Delta_Q(\tau)$ under weak proxies: using the bridge moment equations, Corollary~\ref{cor:method-weak-proxy-bias} establishes a population-bias bound for $\widehat F_a(y)$ with a sharp $\kappa_{a,n}^{-1}$ dependence rather than incurring the $\kappa_{a,n}^{-1}$ sampling-error factor twice as in the Cauchy--Schwarz bound, and Lemma~\ref{lem:app-bridge-class-sizemain} and Theorem~\ref{thm:uniform-dr} control the estimated bridges across treatment arms $a\in\{0,1\}$ and outcome thresholds $y\in\mathcal Y$. Confidence-band inversion of the simultaneous CDF bands then yields simultaneous confidence bands for $\Delta_Q(\tau)$ without counterfactual density estimation. Future work includes inference when $\kappa_{a,n}$ is of the same order as the sample matrix error, data-driven choice of the basis vectors $b_W(W,X)$, $b_Z(Z,X)$ and the ridge parameters $\lambda_h,\lambda_q$, and less conservative confidence bands for $\Delta_Q(\tau)$.

\clearpage
\section*{AI Use Statement}
In this work, a generative AI assistant was used to check compliance with the ICLR 2027 template, edit and condense the manuscript language, and refine the \LaTeX{} formatting, including the presentation of cited baseline methods and reported results.

\section*{Reproducibility Statement}
The proposed implementation is available in our
\href{https://anonymous.4open.science/r/Simultaneous-Inference-for-Proximal-Causal-Learning-with-Weak-Proxies-126A/proximal_method.py}{anonymous code repository}.
\section*{Ethics Statement}
The work does not recruit participants or use individual-level human-subject data.

\bibliographystyle{plainnat}
\bibliography{related_work_references}

\clearpage
\appendix
\section{Derivation and Proof for Regularized Finite-Dimensional Bridge Estimation}
\label{app:finite-dimensional-bridges}
We derive the moment equations and regularized estimates in Section~\ref{sec:method-bridge-estimation}, define the finite-dimensional approximation errors and coefficient norm bounds, and prove Proposition~\ref{prop:method-bridge-rates}.
All expectations $E$ are under the true distribution of the observed data $P_0$.
\subsection{Derivation of the Finite-Dimensional Moment Equations}
\label{app:bridge-moment-derivation}
The folds $I_1,\ldots,I_K$ partition the observation indices $\{1,\ldots,n\}$ independently of $O_1,\ldots,O_n$.
The number of folds $K$ is fixed, and $|I_k|\asymp n$ for every $k$.
The evaluation sample for fold $I_k$ is $\{O_i:i\in I_k\}$, and the training sample is $\{O_i:i\notin I_k\}$.
Define the training sample size and training sample average by
\begin{equation}
N_k=n-|I_k|,
\qquad
P_{n,-k}f=\frac{1}{N_k}\sum_{i\notin I_k}f(O_i).
\label{eq:app-training-average}
\end{equation}
Because $K$ is fixed and every fold contains a number of observations proportional to $n$, $N_k\asymp n$.
The superscript $(-k)$ indicates that estimation uses only the training sample outside $I_k$.
\par
Fix a treatment arm $a\in\{0,1\}$.
The population counterparts of the sample moments in \eqref{eq:sample-outcome-bridge}--\eqref{eq:sample-treatment-bridge} are
\begin{equation}
\begin{aligned}
\Sigma_a
&=E\!\left\{\mathbf{1}\{A=a\}b_Z(Z,X)b_W(W,X)^\top\right\},\\
\gamma_a(y)
&=E\!\left\{\mathbf{1}\{A=a\}b_Z(Z,X)\mathbf{1}\{Y\le y\}\right\},\\
\mu_W
&=E\{b_W(W,X)\}.
\end{aligned}
\label{eq:app-population-moments}
\end{equation}
The matrix $\Sigma_a$ has dimension $d_n\times d_n$, and $\gamma_a(y)$ and $\mu_W$ have dimension $d_n$.
\par
\subsubsection{From Proximal Identification to a Primal--Dual Inverse Problem}
\label{sec:Primal--Dual}
Although the representation in Section~\ref{sec:background-identification} identifies $F_a(y)$ using the outcome bridge $h_{a,y}(W,X)$ and the treatment bridge $q_a(Z,X)$, it does not by itself guarantee stable estimation. The outcome bridge and treatment bridge are defined indirectly through conditional expectation equations. In particular, recovering $h_{a,y}(W,X)$ can be unstable when the proxies $Z$ and $W$ provide limited information about the unobserved latent variable $U$. 

We rewrite the two bridge equations in operator form. For each treatment arm $a\in\{0,1\}$, define the conditional expectation operator $T_a$ through
$$
(T_a h_{a,y})(Z,X)
=
E\{h_{a,y}(W,X)\mid A=a,Z,X\}.
$$
Also define
$$
g_{a,y}(Z,X)
=
E\{\mathbf{1}\{Y\le y\}\mid A=a,Z,X\}.
$$
The equation for the outcome bridge $h_{a,y}(W,X)$ can then be written as
$$
(T_a h_{a,y})(Z,X)
=
g_{a,y}(Z,X),
$$
or, more compactly,
$$
T_a h_{a,y}=g_{a,y}.
$$
For the outcome bridge $h_{a,y}(W,X)$ and the treatment bridge $q_a(Z,X)$,
$$
\begin{aligned}
&E\bigl[\mathbf{1}\{A=a\}(T_a h_{a,y})(Z,X)q_a(Z,X)\bigr]=
E\bigl[\mathbf{1}\{A=a\}h_{a,y}(W,X)q_a(Z,X)\bigr]=
E\bigl[h_{a,y}(W,X)(T_a^*q_a)(W,X)\bigr],
\end{aligned}
$$
where the adjoint operator $T_a^*$ acts on the treatment bridge $q_a(Z,X)$ as
$$
(T_a^*q_a)(W,X)
=
E\{\mathbf{1}\{A=a\}q_a(Z,X)\mid W,X\}.
$$
By the definition of the treatment bridge $q_a(Z,X)$,
$$
E\{\mathbf{1}\{A=a\}q_a(Z,X)-1\mid W,X\}=0.
$$
Therefore,
$$
(T_a^*q_a)(W,X)
=
E\{\mathbf{1}\{A=a\}q_a(Z,X)\mid W,X\}
=
1,
$$
which we write compactly as
$$
T_a^*q_a=1.
$$
Together with the outcome bridge equation, the two bridge equations can now be written as the primal--dual pair
$$
T_a h_{a,y}=g_{a,y},
\qquad
T_a^*q_a=1.
$$
The two equations show how functions of $(W,X)$ and functions of $(Z,X)$ are
linked in opposite directions through $T_a$ and its adjoint operator $T_a^{*}$: $T_a$
maps the outcome bridge $h_{a,y}(W,X)$ to the observed quantity $g_{a,y}(Z,X)$, while $T_a^{*}$ maps the treatment bridge $q_a(Z,X)$ back to the constant function $(T_a^{*}q_a)(W,X)=1$.

Recalling~\eqref{eq:identification-Fa}, the primal--dual pair is useful because the marginal counterfactual distribution $F_a(y)$ depends on the outcome bridge only through $F_a(y)=E\{h_{a,y}(W,X)\}$. Using the primal--dual pair $T_a^{*}q_a=1$ and $T_a h_{a,y}=g_{a,y}$,
\begin{align}
F_a(y)
&=E\{h_{a,y}(W,X)\}=E\bigl[h_{a,y}(W,X)\,(T_a^{*}q_a)(W,X)\bigr]\nonumber\\
&=E\bigl[\mathbf{1}\{A=a\}(T_a h_{a,y})(Z,X)\,q_a(Z,X)\bigr]\nonumber\\
&=E\bigl[\mathbf{1}\{A=a\}g_{a,y}(Z,X)\,q_a(Z,X)\bigr].
\label{eq:primal-dual-Fa}
\end{align}
Thus, even when recovering the entire outcome bridge $h_{a,y}(W,X)$ is unstable, the marginal counterfactual distribution $F_a(y)$ may still be estimated stably through the primal--dual pair.
\subsubsection{Finite-Dimensional Bridge Moment Equations}
\noindent\textit{Outcome bridge moment equation.}
The outcome bridge equation $T_a h_{a,y}=g_{a,y}$ gives
\[
E\!\left\{\mathbf{1}\{Y\le y\}-h_{a,y}(W,X)\mid A=a,Z,X\right\}=0.
\]
Multiplying the conditional equation by $b_Z(Z,X)$ and integrating over observations with $A=a$ yields
\begin{equation}
\begin{aligned}
&E\!\left[\mathbf{1}\{A=a\}b_Z(Z,X)
\{\mathbf{1}\{Y\le y\}-h_{a,y}(W,X)\}\right]\\
&\quad=E\!\left[\mathbf{1}\{A=a\}b_Z(Z,X)
E\!\left\{\mathbf{1}\{Y\le y\}-h_{a,y}(W,X)\mid A,Z,X\right\}\right]
=0.
\end{aligned}
\label{eq:app-outcome-tested-moment}
\end{equation}
For the finite-dimensional representation $b_W(W,X)^\top\theta$, the tested moment becomes
\begin{equation}
\begin{aligned}
&E\!\left[\mathbf{1}\{A=a\}b_Z(Z,X)
\{\mathbf{1}\{Y\le y\}-b_W(W,X)^\top\theta\}\right]\\
&\quad=E\!\left\{\mathbf{1}\{A=a\}b_Z(Z,X)\mathbf{1}\{Y\le y\}\right\}
-E\!\left\{\mathbf{1}\{A=a\}b_Z(Z,X)b_W(W,X)^\top\right\}\theta\\
&\quad=\gamma_a(y)-\Sigma_a\theta.
\end{aligned}
\label{eq:app-outcome-matrix-moment}
\end{equation}
Setting the tested moment in \eqref{eq:app-outcome-matrix-moment} equal to zero gives the finite-dimensional equation $\Sigma_a\theta=\gamma_a(y)$.
The finite-dimensional equation imposes the moments tested by $b_Z(Z,X)$, not every conditional moment in $T_a h_{a,y}=g_{a,y}$.
\par
\noindent\textit{Treatment bridge moment equation.}
The treatment bridge equation $T_a^*q_a=1$ gives
\[
E\!\left\{\mathbf{1}\{A=a\}q_a(Z,X)-1\mid W,X\right\}=0.
\]
Multiplying the conditional equation by $b_W(W,X)$ and taking expectation yields
\begin{equation}
\begin{aligned}
&E\!\left[b_W(W,X)\{\mathbf{1}\{A=a\}q_a(Z,X)-1\}\right]\\
&\quad=E\!\left[b_W(W,X)
E\!\left\{\mathbf{1}\{A=a\}q_a(Z,X)-1\mid W,X\right\}\right]
=0.
\end{aligned}
\label{eq:app-treatment-tested-moment}
\end{equation}
For the finite-dimensional representation $b_Z(Z,X)^\top\alpha$, the tested moment becomes
\begin{equation}
\begin{aligned}
&E\!\left[b_W(W,X)\{\mathbf{1}\{A=a\}b_Z(Z,X)^\top\alpha-1\}\right]\\
&\quad=E\!\left\{\mathbf{1}\{A=a\}b_W(W,X)b_Z(Z,X)^\top\right\}\alpha
-E\{b_W(W,X)\}\\
&\quad=\Sigma_a^\top\alpha-\mu_W.
\end{aligned}
\label{eq:app-treatment-matrix-moment}
\end{equation}
Setting the tested moment in \eqref{eq:app-treatment-matrix-moment} equal to zero gives the finite-dimensional equation $\Sigma_a^\top\alpha=\mu_W$.
\par
\noindent\textit{Population finite-dimensional representations.}
When $\Sigma_a$ is nonsingular, define
\begin{equation}
\begin{aligned}
\theta_a^{(d_n)}(y)&=\Sigma_a^{-1}\gamma_a(y),\\
\alpha_a^{(d_n)}&=(\Sigma_a^\top)^{-1}\mu_W,\\
h_{a,y}^{(d_n)}(W,X)&=b_W(W,X)^\top\theta_a^{(d_n)}(y),\\
q_a^{(d_n)}(Z,X)&=b_Z(Z,X)^\top\alpha_a^{(d_n)}.
\end{aligned}
\label{eq:app-population-bridges}
\end{equation}
The coefficients $\theta_a^{(d_n)}(y)$ and $\alpha_a^{(d_n)}$ solve
\begin{equation}
\Sigma_a\theta_a^{(d_n)}(y)=\gamma_a(y),
\qquad
\Sigma_a^\top\alpha_a^{(d_n)}=\mu_W.
\label{eq:app-population-coefficient-equations}
\end{equation}
The population finite-dimensional representations $h_{a,y}^{(d_n)}(W,X)$ and $q_a^{(d_n)}(Z,X)$ need not equal the true outcome bridge $h_{a,y}(W,X)$ and treatment bridge $q_a(Z,X)$.
\par
\noindent\textit{Sample moment equations.}
Replacing $E$ with $P_{n,-k}$ in \eqref{eq:app-population-moments} gives
\begin{equation}
\begin{aligned}
\widehat\Sigma_a^{(-k)}
&=\frac{1}{N_k}\sum_{i\notin I_k}
\mathbf{1}\{A_i=a\}b_Z(Z_i,X_i)b_W(W_i,X_i)^\top,\\
\widehat\gamma_a^{(-k)}(y)
&=\frac{1}{N_k}\sum_{i\notin I_k}
\mathbf{1}\{A_i=a\}b_Z(Z_i,X_i)\mathbf{1}\{Y_i\le y\},\\
\widehat\mu_W^{(-k)}
&=\frac{1}{N_k}\sum_{i\notin I_k}b_W(W_i,X_i).
\end{aligned}
\label{eq:app-sample-moments}
\end{equation}
The sample versions of the tested moments are therefore
\begin{equation}
\begin{aligned}
&P_{n,-k}\!\left[\mathbf{1}\{A=a\}b_Z(Z,X)
\{\mathbf{1}\{Y\le y\}-b_W(W,X)^\top\theta\}\right]\\
&\quad=\widehat\gamma_a^{(-k)}(y)-\widehat\Sigma_a^{(-k)}\theta,\\
&P_{n,-k}\!\left[b_W(W,X)
\{\mathbf{1}\{A=a\}b_Z(Z,X)^\top\alpha-1\}\right]\\
&\quad=(\widehat\Sigma_a^{(-k)})^\top\alpha-\widehat\mu_W^{(-k)}.
\end{aligned}
\label{eq:app-sample-tested-moments}
\end{equation}
Exact sample moment estimation would set both expressions in \eqref{eq:app-sample-tested-moments} equal to zero.
The regularized estimates in \eqref{eq:method-regularized-bridges} instead minimize squared moment errors plus coefficient penalties.
The regularized estimates need not satisfy the unregularized sample moment equations exactly.
\subsection{Derivation of the Regularized Estimates}
\label{app:bridge-regularized-derivation}
For the outcome bridge, differentiating the objective in \eqref{eq:method-regularized-bridges} with respect to $\theta$ gives
\begin{equation}
\begin{aligned}
&\nabla_\theta\!\left\{
\|\widehat\Sigma_a^{(-k)}\theta-\widehat\gamma_a^{(-k)}(y)\|_2^2
+\lambda_h\|\theta\|_2^2\right\}\\
&\quad=2(\widehat\Sigma_a^{(-k)})^\top
\{\widehat\Sigma_a^{(-k)}\theta-\widehat\gamma_a^{(-k)}(y)\}
+2\lambda_h\theta.
\end{aligned}
\label{eq:app-outcome-gradient}
\end{equation}
For the treatment bridge, differentiating the objective with respect to $\alpha$ gives
\begin{equation}
\begin{aligned}
&\nabla_\alpha\!\left\{
\|(\widehat\Sigma_a^{(-k)})^\top\alpha-\widehat\mu_W^{(-k)}\|_2^2
+\lambda_q\|\alpha\|_2^2\right\}\\
&\quad=2\widehat\Sigma_a^{(-k)}
\{(\widehat\Sigma_a^{(-k)})^\top\alpha-\widehat\mu_W^{(-k)}\}
+2\lambda_q\alpha.
\end{aligned}
\label{eq:app-treatment-gradient}
\end{equation}
Let $I_{d_n}$ denote the $d_n\times d_n$ identity matrix.
The Hessians of the outcome bridge and treatment bridge objectives are twice the respective matrices
\[
(\widehat\Sigma_a^{(-k)})^\top\widehat\Sigma_a^{(-k)}+\lambda_h I_{d_n},
\qquad
\widehat\Sigma_a^{(-k)}(\widehat\Sigma_a^{(-k)})^\top+\lambda_q I_{d_n}.
\]
For every nonzero vector $v\in\mathbb R^{d_n}$,
\begin{equation}
\begin{aligned}
&v^\top\{(\widehat\Sigma_a^{(-k)})^\top\widehat\Sigma_a^{(-k)}
+\lambda_h I_{d_n}\}v
=\|\widehat\Sigma_a^{(-k)}v\|_2^2+\lambda_h\|v\|_2^2>0,\\
&v^\top\{\widehat\Sigma_a^{(-k)}(\widehat\Sigma_a^{(-k)})^\top
+\lambda_q I_{d_n}\}v
=\|(\widehat\Sigma_a^{(-k)})^\top v\|_2^2+\lambda_q\|v\|_2^2>0.
\end{aligned}
\label{eq:app-positive-hessians}
\end{equation}
Positive $\lambda_h$ and $\lambda_q$ therefore give strictly convex quadratic objectives and unique minimizers, even when $\widehat\Sigma_a^{(-k)}$ is singular.
Setting the gradients in \eqref{eq:app-outcome-gradient} and \eqref{eq:app-treatment-gradient} equal to zero gives
\begin{equation}
\begin{aligned}
&\{(\widehat\Sigma_a^{(-k)})^\top\widehat\Sigma_a^{(-k)}+\lambda_h I_{d_n}\}
\widehat\theta_{a,\lambda_h}^{(-k)}(y)
=(\widehat\Sigma_a^{(-k)})^\top\widehat\gamma_a^{(-k)}(y),\\
&\{\widehat\Sigma_a^{(-k)}(\widehat\Sigma_a^{(-k)})^\top+\lambda_q I_{d_n}\}
\widehat\alpha_{a,\lambda_q}^{(-k)}
=\widehat\Sigma_a^{(-k)}\widehat\mu_W^{(-k)}.
\end{aligned}
\label{eq:app-regularized-normal-equations}
\end{equation}
Inverting the positive definite matrices in \eqref{eq:app-regularized-normal-equations} yields
\begin{equation}
\begin{aligned}
\widehat\theta_{a,\lambda_h}^{(-k)}(y)
&=\{(\widehat\Sigma_a^{(-k)})^\top\widehat\Sigma_a^{(-k)}+\lambda_h I_{d_n}\}^{-1}
(\widehat\Sigma_a^{(-k)})^\top\widehat\gamma_a^{(-k)}(y),\\
\widehat\alpha_{a,\lambda_q}^{(-k)}
&=\{\widehat\Sigma_a^{(-k)}(\widehat\Sigma_a^{(-k)})^\top+\lambda_q I_{d_n}\}^{-1}
\widehat\Sigma_a^{(-k)}\widehat\mu_W^{(-k)}.
\end{aligned}
\label{eq:app-regularized-solutions}
\end{equation}
Substituting the coefficients in \eqref{eq:app-regularized-solutions} into the basis representations gives exactly the estimated bridges in \eqref{eq:method-estimated-bridges}.
Within each training sample, $b_W(W,X)$, $b_Z(Z,X)$, $\widehat\Sigma_a^{(-k)}$, $\lambda_h$, and $\lambda_q$ are fixed across thresholds.
Consequently, the estimated outcome bridge depends on $y$ only through $\widehat\gamma_a^{(-k)}(y)$, and the estimated treatment bridge does not depend on $y$.
\subsection{Finite-Dimensional Approximation Errors and Assumptions}
\label{app:bridge-error-definitions}
For functions $u(W,X)$ and $v(Z,X)$, the norms used in the main text are $\|u\|_W^2=E\{u(W,X)^2\}$ and $\|v\|_{a,Z}^2=E\{\mathbf{1}\{A=a\}v(Z,X)^2\}$.
Using $\|\cdot\|_W$ and $\|\cdot\|_{a,Z}$, define
\begin{equation}
\begin{aligned}
b_{h,a,n}
&=\sup_{y\in\mathcal Y}
\|h_{a,y}^{(d_n)}(W,X)-h_{a,y}(W,X)\|_W,\\
b_{q,a,n}
&=\|q_a^{(d_n)}(Z,X)-q_a(Z,X)\|_{a,Z},\\
M_{h,a,n}
&=\sup_{y\in\mathcal Y}\|\theta_a^{(d_n)}(y)\|_2,\\
M_{q,a,n}
&=\|\alpha_a^{(d_n)}\|_2.
\end{aligned}
\label{eq:app-bridge-error-definitions}
\end{equation}
The finite-dimensional approximation errors $b_{h,a,n}$ and $b_{q,a,n}$ concern the population finite-dimensional representations in \eqref{eq:app-population-bridges}.
The coefficient norm bounds $M_{h,a,n}$ and $M_{q,a,n}$ concern population coefficients.
The quantities in \eqref{eq:app-bridge-error-definitions} may depend on $n$ through the basis dimension and the population moment equations.
No zero finite-dimensional approximation error or uniformly bounded coefficient norm is assumed.
\begin{assumption}[Finite-dimensional concentration and weak-proxy strength]
\label{ass:finite-rank}
The following conditions hold for both treatment arms $a\in\{0,1\}$.
\begin{enumerate}[label=(\roman*)]
\item The basis vectors $b_W(W,X)$ and $b_Z(Z,X)$ are prespecified, have the same dimension $d_n$, and are bounded or sub-Gaussian.
There are constants $0<c_b\le C_b<\infty$, independent of $n$ and $a$, such that for every $v\in\mathbb R^{d_n}$,
\begin{equation}
\begin{aligned}
c_b\|v\|_2^2
&\le v^\top E\{b_W(W,X)b_W(W,X)^\top\}v
\le C_b\|v\|_2^2,\\
c_b\|v\|_2^2
&\le v^\top E\!\left\{\mathbf{1}\{A=a\}b_Z(Z,X)b_Z(Z,X)^\top\right\}v
\le C_b\|v\|_2^2.
\end{aligned}
\label{eq:app-basis-moment-bounds}
\end{equation}
\item The cross-moment matrix $\Sigma_a$ is nonsingular, with $\kappa_{a,n}=\sigma_{\min}(\Sigma_a)>0$.
The finite-dimensional approximation errors and coefficient norm bounds in \eqref{eq:app-bridge-error-definitions} are finite.
\item With the fixed folds and training sample sizes in \eqref{eq:app-training-average}, the cross-moment matrix and basis mean satisfy
\begin{equation}
\begin{aligned}
\max_k\|\widehat\Sigma_a^{(-k)}-\Sigma_a\|_{\mathrm{op}}
&=O_p\!\left(\sqrt{\frac{d_n}{n}}\right),\\
\max_k\|\widehat\mu_W^{(-k)}-\mu_W\|_2
&=O_p\!\left(\sqrt{\frac{d_n}{n}}\right).
\end{aligned}
\label{eq:app-moment-concentration}
\end{equation}
Here $\|\cdot\|_{\mathrm{op}}$ denotes the matrix operator norm induced by the Euclidean norm $\|\cdot\|_2$.
\end{enumerate}
\end{assumption}
The moment-concentration conditions in \eqref{eq:app-moment-concentration} are assumptions on observed matrix and vector averages.
The moment-concentration conditions are not deduced from boundedness or sub-Gaussianity alone.
\par
\noindent\textit{Uniform concentration of the threshold moments.}
The upper Gram bound in Assumption~\ref{ass:finite-rank}(i) implies
\[
E\!\left\{\mathbf{1}\{A=a\}\|b_Z(Z,X)\|_2^2\right\}
=\operatorname{tr}E\!\left\{\mathbf{1}\{A=a\}b_Z(Z,X)b_Z(Z,X)^\top\right\}
\le C_b d_n.
\]
Fix a treatment arm $a$ and a fold $I_k$, and condition on the fold partition.
The training observations remain iid from $P_0$, and $N_k$ is fixed in the following calculation.
Let $\{\varepsilon_i:i\notin I_k\}$ be independent Rademacher signs, independent of the observations and the fold partition, with $P(\varepsilon_i=1)=P(\varepsilon_i=-1)=1/2$.
Write $E_\varepsilon$ for expectation over the signs conditional on the training observations and the fold partition.
All suprema below are measurable, since the threshold sums and $\gamma_a(y)$ are right-continuous and a fixed countable dense subset of $\mathcal Y$ together with $y_U$ determines the suprema.
Symmetrization gives
\[
\begin{aligned}
&E\!\left[\sup_{y\in\mathcal Y}
\|\widehat\gamma_a^{(-k)}(y)-\gamma_a(y)\|_2\right]\\
&\quad\le\frac{2}{N_k}E E_\varepsilon
\sup_{y\in\mathcal Y}\left\|
\sum_{i\notin I_k}\varepsilon_i\mathbf{1}\{A_i=a\}
 b_Z(Z_i,X_i)\mathbf{1}\{Y_i\le y\}
\right\|_2.
\end{aligned}
\]
Conditionally on the training observations, order the training indices as $i_1,\ldots,i_{N_k}$ so that $Y_{i_1}\le\cdots\le Y_{i_{N_k}}$, breaking ties by the original observation indices.
Define
\[
S_0=0,\qquad
S_m=\sum_{j=1}^{m}\varepsilon_{i_j}\mathbf{1}\{A_{i_j}=a\}b_Z(Z_{i_j},X_{i_j}),
\qquad m=1,\ldots,N_k.
\]
Every threshold sum is one of the partial sums $S_m$, including when training outcomes are tied.
The sequence $S_m$ is a square-integrable vector martingale with respect to the successive Rademacher signs, and $\|S_m\|_2$ is a nonnegative submartingale.
Doob's $L_2$ maximal inequality therefore gives
\[
\begin{aligned}
&E_\varepsilon\sup_{y\in\mathcal Y}\left\|
\sum_{i\notin I_k}\varepsilon_i\mathbf{1}\{A_i=a\}
 b_Z(Z_i,X_i)\mathbf{1}\{Y_i\le y\}
\right\|_2^2\\
&\quad\le E_\varepsilon\max_{0\le m\le N_k}\|S_m\|_2^2
\le4E_\varepsilon\|S_{N_k}\|_2^2\\
&\quad=4\sum_{i\notin I_k}\mathbf{1}\{A_i=a\}\|b_Z(Z_i,X_i)\|_2^2.
\end{aligned}
\]
Combining the symmetrization and maximal bounds, then applying Jensen's inequality, yields
\[
\begin{aligned}
&E\!\left[\sup_{y\in\mathcal Y}
\|\widehat\gamma_a^{(-k)}(y)-\gamma_a(y)\|_2\right]\\
&\quad\le\frac{4}{N_k}E\!\left[
\left\{\sum_{i\notin I_k}\mathbf{1}\{A_i=a\}\|b_Z(Z_i,X_i)\|_2^2\right\}^{1/2}
\right]\\
&\quad\le4\sqrt{\frac{E\{\mathbf{1}\{A=a\}\|b_Z(Z,X)\|_2^2\}}{N_k}}
\le4\sqrt{\frac{C_b d_n}{N_k}}.
\end{aligned}
\]
Markov's inequality, $N_k\asymp n$, and the fixed number of folds imply
\begin{equation}
\max_k\sup_{y\in\mathcal Y}
\|\widehat\gamma_a^{(-k)}(y)-\gamma_a(y)\|_2
=O_p\!\left(\sqrt{\frac{d_n}{n}}\right).
\label{eq:app-threshold-moment-rate}
\end{equation}
The rate in \eqref{eq:app-threshold-moment-rate} follows from the upper Gram bound and the threshold ordering, without an additional logarithmic factor.
The argument applies to both treatment arms and does not require overlapping training samples to be independent.
The matrix concentration condition in \eqref{eq:app-moment-concentration} remains a separate assumption.
\subsection{Proof of the Uniform Bridge Estimation Error Bounds}
\label{app:bridge-rates-proof}
\begin{proof}[Proof of Proposition~\ref{prop:method-bridge-rates}]
\noindent\textit{Step 1: Control the sample matrices.}
By \eqref{eq:app-moment-concentration} and $\sqrt{d_n/n}=o(\kappa_{a,n})$,
\[
\frac{\max_k\|\widehat\Sigma_a^{(-k)}-\Sigma_a\|_{\mathrm{op}}}
{\kappa_{a,n}}
=O_p\!\left(\frac{1}{\kappa_{a,n}}\sqrt{\frac{d_n}{n}}\right)
=o_p(1).
\]
Because there are two treatment arms, the event
\begin{equation}
\mathcal E_n
=\bigcap_{a\in\{0,1\}}
\left\{\max_k\|\widehat\Sigma_a^{(-k)}-\Sigma_a\|_{\mathrm{op}}
\le\frac{\kappa_{a,n}}{2}\right\}
\label{eq:app-matrix-event}
\end{equation}
has probability tending to one.
On $\mathcal E_n$, for every $a$, $k$, and $v\in\mathbb R^{d_n}$,
\begin{equation}
\begin{aligned}
\|\widehat\Sigma_a^{(-k)}v\|_2
&\ge\|\Sigma_a v\|_2
-\|(\widehat\Sigma_a^{(-k)}-\Sigma_a)v\|_2\\
&\ge\left\{\kappa_{a,n}
-\|\widehat\Sigma_a^{(-k)}-\Sigma_a\|_{\mathrm{op}}\right\}\|v\|_2\\
&\ge\frac{\kappa_{a,n}}{2}\|v\|_2.
\end{aligned}
\label{eq:app-sample-singular-value}
\end{equation}
Taking the infimum over unit vectors gives
\[
\sigma_{\min}(\widehat\Sigma_a^{(-k)})\ge\frac{\kappa_{a,n}}{2},
\qquad
\|\{\widehat\Sigma_a^{(-k)}\}^{-1}\|_{\mathrm{op}}
\le\frac{2}{\kappa_{a,n}}.
\]
The eigenvalues of both $(\widehat\Sigma_a^{(-k)})^\top\widehat\Sigma_a^{(-k)}$ and $\widehat\Sigma_a^{(-k)}(\widehat\Sigma_a^{(-k)})^\top$ are squared singular values of $\widehat\Sigma_a^{(-k)}$.
Every eigenvalue is therefore at least $\kappa_{a,n}^2/4$ on $\mathcal E_n$.
Adding the positive regularization parameters and inverting gives
\begin{equation}
\begin{aligned}
\left\|\{(\widehat\Sigma_a^{(-k)})^\top\widehat\Sigma_a^{(-k)}
+\lambda_h I_{d_n}\}^{-1}\right\|_{\mathrm{op}}
&\le\frac{1}{\kappa_{a,n}^2/4+\lambda_h},\\
\left\|\{\widehat\Sigma_a^{(-k)}(\widehat\Sigma_a^{(-k)})^\top
+\lambda_q I_{d_n}\}^{-1}\right\|_{\mathrm{op}}
&\le\frac{1}{\kappa_{a,n}^2/4+\lambda_q}.
\end{aligned}
\label{eq:app-regularized-inverse-bounds}
\end{equation}
Let $\sigma_j(\widehat\Sigma_a^{(-k)})$, $j=1,\ldots,d_n$, denote the singular values of the sample matrix, including zero singular values. Singular value decomposition gives, for every $\lambda_h,\lambda_q>0$,
\begin{equation}
\begin{aligned}
&\left\|\{(\widehat\Sigma_a^{(-k)})^\top\widehat\Sigma_a^{(-k)}
+\lambda_h I_{d_n}\}^{-1}(\widehat\Sigma_a^{(-k)})^\top\right\|_{\mathrm{op}}
=\max_j\frac{\sigma_j(\widehat\Sigma_a^{(-k)})}
{\sigma_j(\widehat\Sigma_a^{(-k)})^2+\lambda_h},\\
&\left\|\{\widehat\Sigma_a^{(-k)}(\widehat\Sigma_a^{(-k)})^\top
+\lambda_q I_{d_n}\}^{-1}\widehat\Sigma_a^{(-k)}\right\|_{\mathrm{op}}
=\max_j\frac{\sigma_j(\widehat\Sigma_a^{(-k)})}
{\sigma_j(\widehat\Sigma_a^{(-k)})^2+\lambda_q}.
\end{aligned}
\label{eq:app-moment-factorization}
\end{equation}
The identities in \eqref{eq:app-moment-factorization} do not require an invertible sample matrix. On $\mathcal E_n$, the singular values are at least $\kappa_{a,n}/2$. For any $\lambda>0$, elementary maximization gives
\begin{equation}
\sup_{t\ge\kappa_{a,n}/2}\frac{t}{t^2+\lambda}
=\begin{cases}
\displaystyle\frac{\kappa_{a,n}/2}{\kappa_{a,n}^2/4+\lambda},
&0<\lambda\le\kappa_{a,n}^2/4,\\[4pt]
\displaystyle\frac{1}{2\sqrt\lambda},&\lambda>\kappa_{a,n}^2/4.
\end{cases}
\label{eq:app-exact-ridge-factor}
\end{equation}
In particular,
\begin{equation}
\begin{aligned}
\left\|\{(\widehat\Sigma_a^{(-k)})^\top\widehat\Sigma_a^{(-k)}
+\lambda_h I_{d_n}\}^{-1}(\widehat\Sigma_a^{(-k)})^\top\right\|_{\mathrm{op}}
&\le\sup_{t\ge\kappa_{a,n}/2}\frac{t}{t^2+\lambda_h}
\le\frac{2}{\sqrt{\kappa_{a,n}^2+\lambda_h}},\\
\left\|\{\widehat\Sigma_a^{(-k)}(\widehat\Sigma_a^{(-k)})^\top
+\lambda_q I_{d_n}\}^{-1}\widehat\Sigma_a^{(-k)}\right\|_{\mathrm{op}}
&\le\sup_{t\ge\kappa_{a,n}/2}\frac{t}{t^2+\lambda_q}
\le\frac{2}{\sqrt{\kappa_{a,n}^2+\lambda_q}}.
\end{aligned}
\label{eq:app-regularized-moment-bounds}
\end{equation}
The last inequalities follow from $t/(t^2+\lambda)\le(t^2+\lambda)^{-1/2}$ and $t\ge\kappa_{a,n}/2$. The proposition uses this simpler expression up to numerical constants, while \eqref{eq:app-moment-factorization}--\eqref{eq:app-exact-ridge-factor} retain the exact factors.
\noindent\textit{Step 2: Bound the outcome bridge estimation error.}
Subtracting $\{(\widehat\Sigma_a^{(-k)})^\top\widehat\Sigma_a^{(-k)}+\lambda_h I_{d_n}\}\theta_a^{(d_n)}(y)$ from both sides of the outcome bridge normal equation in \eqref{eq:app-regularized-normal-equations} gives
\begin{equation}
\begin{aligned}
&\{(\widehat\Sigma_a^{(-k)})^\top\widehat\Sigma_a^{(-k)}+\lambda_h I_{d_n}\}
\{\widehat\theta_{a,\lambda_h}^{(-k)}(y)-\theta_a^{(d_n)}(y)\}\\
&\quad=(\widehat\Sigma_a^{(-k)})^\top
\{\widehat\gamma_a^{(-k)}(y)-\widehat\Sigma_a^{(-k)}\theta_a^{(d_n)}(y)\}
-\lambda_h\theta_a^{(d_n)}(y).
\end{aligned}
\label{eq:app-outcome-coefficient-difference}
\end{equation}
The population equation $\Sigma_a\theta_a^{(d_n)}(y)=\gamma_a(y)$ gives the exact identity
\[
\begin{aligned}
\widehat\gamma_a^{(-k)}(y)-\widehat\Sigma_a^{(-k)}\theta_a^{(d_n)}(y)
&=\widehat\gamma_a^{(-k)}(y)-\gamma_a(y)\\
&\quad+(\Sigma_a-\widehat\Sigma_a^{(-k)})\theta_a^{(d_n)}(y).
\end{aligned}
\]
Substitution into \eqref{eq:app-outcome-coefficient-difference}, followed by multiplication by the regularized inverse, gives
\begin{equation}
\begin{aligned}
&\widehat\theta_{a,\lambda_h}^{(-k)}(y)-\theta_a^{(d_n)}(y)\\
&\quad=\{(\widehat\Sigma_a^{(-k)})^\top\widehat\Sigma_a^{(-k)}+\lambda_h I_{d_n}\}^{-1}
(\widehat\Sigma_a^{(-k)})^\top\\
&\qquad\times\left\{\widehat\gamma_a^{(-k)}(y)-\gamma_a(y)
+(\Sigma_a-\widehat\Sigma_a^{(-k)})\theta_a^{(d_n)}(y)\right\}\\
&\qquad-\lambda_h
\{(\widehat\Sigma_a^{(-k)})^\top\widehat\Sigma_a^{(-k)}+\lambda_h I_{d_n}\}^{-1}
\theta_a^{(d_n)}(y).
\end{aligned}
\label{eq:app-outcome-error-decomposition}
\end{equation}
Applying \eqref{eq:app-regularized-inverse-bounds} and \eqref{eq:app-regularized-moment-bounds}, then taking the supremum over thresholds and the maximum over folds, yields
\begin{equation}
\begin{aligned}
&\max_k\sup_{y\in\mathcal Y}
\|\widehat\theta_{a,\lambda_h}^{(-k)}(y)-\theta_a^{(d_n)}(y)\|_2\\
&\quad\le\left(\sup_{t\ge\kappa_{a,n}/2}\frac{t}{t^2+\lambda_h}\right)
\left\{\max_k\sup_{y\in\mathcal Y}
\|\widehat\gamma_a^{(-k)}(y)-\gamma_a(y)\|_2
+M_{h,a,n}\max_k\|\widehat\Sigma_a^{(-k)}-\Sigma_a\|_{\mathrm{op}}\right\}\\
&\qquad+\frac{\lambda_hM_{h,a,n}}{\kappa_{a,n}^2/4+\lambda_h}.
\end{aligned}
\label{eq:app-outcome-coefficient-bound}
\end{equation}
To convert coefficient error into outcome bridge estimation error, integrate over $(W,X)$ with the estimated coefficients held fixed.
By \eqref{eq:method-estimated-bridges}, \eqref{eq:app-population-bridges}, and \eqref{eq:app-basis-moment-bounds},
\begin{equation}
\begin{aligned}
&\|\widehat h_{a,y}^{(-k)}(W,X)-h_{a,y}^{(d_n)}(W,X)\|_W^2\\
&\quad=\{\widehat\theta_{a,\lambda_h}^{(-k)}(y)-\theta_a^{(d_n)}(y)\}^\top
E\{b_W(W,X)b_W(W,X)^\top\}\\
&\qquad\times\{\widehat\theta_{a,\lambda_h}^{(-k)}(y)-\theta_a^{(d_n)}(y)\}\\
&\quad\le C_b\|\widehat\theta_{a,\lambda_h}^{(-k)}(y)-\theta_a^{(d_n)}(y)\|_2^2.
\end{aligned}
\label{eq:app-outcome-norm-conversion}
\end{equation}
The triangle inequality and the definition of $b_{h,a,n}$ imply
\begin{equation}
\begin{aligned}
&\max_k\sup_{y\in\mathcal Y}
\|\widehat h_{a,y}^{(-k)}(W,X)-h_{a,y}(W,X)\|_W\\
&\quad\le\sqrt{C_b}\max_k\sup_{y\in\mathcal Y}
\|\widehat\theta_{a,\lambda_h}^{(-k)}(y)-\theta_a^{(d_n)}(y)\|_2
+b_{h,a,n}.
\end{aligned}
\label{eq:app-outcome-total-error}
\end{equation}
\noindent\textit{Step 3: Bound the treatment bridge estimation error.}
Subtracting $\{\widehat\Sigma_a^{(-k)}(\widehat\Sigma_a^{(-k)})^\top+\lambda_q I_{d_n}\}\alpha_a^{(d_n)}$ from both sides of the treatment bridge normal equation in \eqref{eq:app-regularized-normal-equations} gives
\begin{equation}
\begin{aligned}
&\{\widehat\Sigma_a^{(-k)}(\widehat\Sigma_a^{(-k)})^\top+\lambda_q I_{d_n}\}
\{\widehat\alpha_{a,\lambda_q}^{(-k)}-\alpha_a^{(d_n)}\}\\
&\quad=\widehat\Sigma_a^{(-k)}
\{\widehat\mu_W^{(-k)}-(\widehat\Sigma_a^{(-k)})^\top\alpha_a^{(d_n)}\}
-\lambda_q\alpha_a^{(d_n)}.
\end{aligned}
\label{eq:app-treatment-coefficient-difference}
\end{equation}
Because $\Sigma_a^\top\alpha_a^{(d_n)}=\mu_W$,
\[
\begin{aligned}
\widehat\mu_W^{(-k)}-(\widehat\Sigma_a^{(-k)})^\top\alpha_a^{(d_n)}
&=\widehat\mu_W^{(-k)}-\mu_W\\
&\quad+\{\Sigma_a^\top-(\widehat\Sigma_a^{(-k)})^\top\}\alpha_a^{(d_n)}.
\end{aligned}
\]
Substituting the population equation into \eqref{eq:app-treatment-coefficient-difference} and multiplying by the regularized inverse gives
\begin{equation}
\begin{aligned}
&\widehat\alpha_{a,\lambda_q}^{(-k)}-\alpha_a^{(d_n)}\\
&\quad=\{\widehat\Sigma_a^{(-k)}(\widehat\Sigma_a^{(-k)})^\top+\lambda_q I_{d_n}\}^{-1}
\widehat\Sigma_a^{(-k)}\\
&\qquad\times\left\{\widehat\mu_W^{(-k)}-\mu_W
+\{\Sigma_a^\top-(\widehat\Sigma_a^{(-k)})^\top\}\alpha_a^{(d_n)}\right\}\\
&\qquad-\lambda_q
\{\widehat\Sigma_a^{(-k)}(\widehat\Sigma_a^{(-k)})^\top+\lambda_q I_{d_n}\}^{-1}
\alpha_a^{(d_n)}.
\end{aligned}
\label{eq:app-treatment-error-decomposition}
\end{equation}
Transposition preserves the matrix operator norm.
Applying \eqref{eq:app-regularized-inverse-bounds} and \eqref{eq:app-regularized-moment-bounds} therefore gives
\begin{equation}
\begin{aligned}
&\max_k\|\widehat\alpha_{a,\lambda_q}^{(-k)}-\alpha_a^{(d_n)}\|_2\\
&\quad\le\left(\sup_{t\ge\kappa_{a,n}/2}\frac{t}{t^2+\lambda_q}\right)
\left\{\max_k\|\widehat\mu_W^{(-k)}-\mu_W\|_2
+M_{q,a,n}\max_k\|\widehat\Sigma_a^{(-k)}-\Sigma_a\|_{\mathrm{op}}\right\}\\
&\qquad+\frac{\lambda_qM_{q,a,n}}{\kappa_{a,n}^2/4+\lambda_q}.
\end{aligned}
\label{eq:app-treatment-coefficient-bound}
\end{equation}
To convert coefficient error into treatment bridge estimation error, use the weighted norm $\|\cdot\|_{a,Z}$ with the estimated coefficients held fixed.
By \eqref{eq:app-basis-moment-bounds},
\begin{equation}
\begin{aligned}
&\|\widehat q_a^{(-k)}(Z,X)-q_a^{(d_n)}(Z,X)\|_{a,Z}^2\\
&\quad=\{\widehat\alpha_{a,\lambda_q}^{(-k)}-\alpha_a^{(d_n)}\}^\top
E\!\left\{\mathbf{1}\{A=a\}b_Z(Z,X)b_Z(Z,X)^\top\right\}\\
&\qquad\times\{\widehat\alpha_{a,\lambda_q}^{(-k)}-\alpha_a^{(d_n)}\}\\
&\quad\le C_b\|\widehat\alpha_{a,\lambda_q}^{(-k)}-\alpha_a^{(d_n)}\|_2^2.
\end{aligned}
\label{eq:app-treatment-norm-conversion}
\end{equation}
The triangle inequality and the definition of $b_{q,a,n}$ give
\begin{equation}
\begin{aligned}
&\max_k\|\widehat q_a^{(-k)}(Z,X)-q_a(Z,X)\|_{a,Z}\\
&\quad\le\sqrt{C_b}\max_k
\|\widehat\alpha_{a,\lambda_q}^{(-k)}-\alpha_a^{(d_n)}\|_2
+b_{q,a,n}.
\end{aligned}
\label{eq:app-treatment-total-error}
\end{equation}
For later use, the same normal equations and singular value decomposition give the following bounds on every sample, without conditioning on $\mathcal E_n$:
\begin{equation}
\begin{aligned}
&\|\widehat\theta_{a,\lambda_h}^{(-k)}(y)-\theta_a^{(d_n)}(y)\|_2\\
&\quad\le\max_j\frac{\sigma_j(\widehat\Sigma_a^{(-k)})}
{\sigma_j(\widehat\Sigma_a^{(-k)})^2+\lambda_h}
\|\widehat\gamma_a^{(-k)}(y)-\widehat\Sigma_a^{(-k)}\theta_a^{(d_n)}(y)\|_2
+\frac{\lambda_h\|\theta_a^{(d_n)}(y)\|_2}
{\sigma_{\min}(\widehat\Sigma_a^{(-k)})^2+\lambda_h},\\
&\|\widehat\alpha_{a,\lambda_q}^{(-k)}-\alpha_a^{(d_n)}\|_2\\
&\quad\le\max_j\frac{\sigma_j(\widehat\Sigma_a^{(-k)})}
{\sigma_j(\widehat\Sigma_a^{(-k)})^2+\lambda_q}
\|\widehat\mu_W^{(-k)}-(\widehat\Sigma_a^{(-k)})^\top\alpha_a^{(d_n)}\|_2
+\frac{\lambda_q\|\alpha_a^{(d_n)}\|_2}
{\sigma_{\min}(\widehat\Sigma_a^{(-k)})^2+\lambda_q}.
\end{aligned}
\label{eq:app-all-ridge-coefficient-bounds}
\end{equation}
Both denominators remain positive for singular sample matrices. The ridge-bias factor is at most one.
\noindent\textit{Step 4: Substitute the moment-concentration rates.}
Substituting \eqref{eq:app-moment-concentration} and \eqref{eq:app-threshold-moment-rate} into \eqref{eq:app-outcome-coefficient-bound} and then \eqref{eq:app-outcome-total-error} gives
\begin{equation}
\begin{aligned}
&\max_k\sup_{y\in\mathcal Y}
\|\widehat h_{a,y}^{(-k)}(W,X)-h_{a,y}(W,X)\|_W\\
&\quad=O_p\!\left(
\frac{1+M_{h,a,n}}{\sqrt{\kappa_{a,n}^2+\lambda_h}}\sqrt{\frac{d_n}{n}}
+b_{h,a,n}
+\frac{\lambda_hM_{h,a,n}}{\kappa_{a,n}^2+\lambda_h}
\right).
\end{aligned}
\label{eq:app-final-outcome-rate}
\end{equation}
Substituting \eqref{eq:app-moment-concentration} into \eqref{eq:app-treatment-coefficient-bound} and then \eqref{eq:app-treatment-total-error} gives
\begin{equation}
\begin{aligned}
&\max_k\|\widehat q_a^{(-k)}(Z,X)-q_a(Z,X)\|_{a,Z}\\
&\quad=O_p\!\left(
\frac{1+M_{q,a,n}}{\sqrt{\kappa_{a,n}^2+\lambda_q}}\sqrt{\frac{d_n}{n}}
+b_{q,a,n}
+\frac{\lambda_qM_{q,a,n}}{\kappa_{a,n}^2+\lambda_q}
\right).
\end{aligned}
\label{eq:app-final-treatment-rate}
\end{equation}
The factor $1+M_{q,a,n}$ combines the error in $\widehat\mu_W^{(-k)}$ with the error in $\widehat\Sigma_a^{(-k)}$ multiplied by $\alpha_a^{(d_n)}$.
The bounds hold on $\mathcal E_n$, whose probability tends to one, so the bounds in \eqref{eq:app-final-outcome-rate} and \eqref{eq:app-final-treatment-rate} hold in probability without conditioning on $\mathcal E_n$.
The inequalities $\lambda_h/(\kappa_{a,n}^2/4+\lambda_h)\le4\lambda_h/(\kappa_{a,n}^2+\lambda_h)$ and its treatment-bridge counterpart give the displayed ridge terms. No small-ridge condition has been used. These upper bounds do not imply consistency when a right-hand side fails to vanish.
The finite-dimensional approximation errors and coefficient norm bounds remain explicit in both rates.
The argument uses concentration uniformly over the fixed number of folds and does not require overlapping training samples to be independent.
\end{proof}
\subsection{Number of Estimated Outcome Bridges across Outcome Thresholds}
\label{app:bridge-class-size}
The error bounds in Proposition~\ref{prop:method-bridge-rates} control estimation accuracy.
The following lemma records a separate consequence of using the same sample matrix and regularization parameter across thresholds.
\begin{lemma}[Number of estimated outcome bridges]
\label{lem:app-bridge-class-size}
Fix a treatment arm $a$ and a fold $I_k$.
Suppose the basis vectors and $\lambda_h>0$ are fixed across $y\in\mathcal Y$ within the training sample, and use the estimated outcome bridge in \eqref{eq:method-estimated-bridges}.
Conditionally on the training sample, the class
$\{\widehat h_{a,y}^{(-k)}(W,X):y\in\mathcal Y\}$
contains at most $N_k+1$ distinct functions.
The $L_2$ covering number of the estimated outcome bridge class is therefore at most $N_k+1$ for every positive radius and every finitely supported probability measure.
On the envelope event in \eqref{eq:app-estimated-bridge-bounds}, for every finitely supported probability measure $\mathsf Q$ and $0<\epsilon\le1$,
\[
N\!\left\{\epsilon H_n,
\{\widehat h_{a,y}^{(-k)}(W,X):y\in\mathcal Y\},L_2(\mathsf Q)\right\}
\le\min\left\{N_k+1,\left(1+\frac4\epsilon\right)^{d_n}\right\}.
\]
\end{lemma}
\begin{proof}
Condition on the training sample $\{O_i:i\notin I_k\}$.
As $y$ increases, each coordinate of the indicator vector
$\bigl(\mathbf{1}\{Y_i\le y\}\bigr)_{i\notin I_k}$
changes from zero to one at most once.
The indicator vector changes only at a distinct training outcome, and there are at most $N_k$ distinct training outcomes.
The indicator vector therefore takes at most $N_k+1$ values.
\par
By \eqref{eq:app-sample-moments}, $\widehat\gamma_a^{(-k)}(y)$ is a fixed weighted sum of the coordinates of the indicator vector.
The vector $\widehat\gamma_a^{(-k)}(y)$ consequently takes at most $N_k+1$ values.
The outcome bridge coefficient in \eqref{eq:app-regularized-solutions} is obtained by multiplying $\widehat\gamma_a^{(-k)}(y)$ by a matrix that does not depend on $y$.
The coefficient $\widehat\theta_{a,\lambda_h}^{(-k)}(y)$ therefore takes at most $N_k+1$ values.
Since $b_W(W,X)$ is fixed across thresholds, the class $\{\widehat h_{a,y}^{(-k)}(W,X):y\in\mathcal Y\}$ contains at most $N_k+1$ distinct functions.
\par
Use every distinct estimated outcome bridge as a covering center.
Every estimated outcome bridge has zero distance from a covering center, so at most $N_k+1$ centers cover the class at any positive $L_2$ radius under any finitely supported probability measure.
\par
By \eqref{eq:method-estimated-bridges}, the same class lies in the linear span of the $d_n$ components of $b_W(W,X)$.
After identifying functions at zero $L_2(\mathsf Q)$ distance, this span has dimension at most $d_n$.
On the envelope event, every estimated outcome bridge has $L_2(\mathsf Q)$ norm at most $H_n$.
For $H_n>0$, a maximal $\epsilon H_n$-separated subset of the class is an $\epsilon H_n$-net.
The disjoint balls of radius $\epsilon H_n/2$ around its elements lie in a ball of radius $(1+\epsilon/2)H_n$ in that finite-dimensional space.
Comparing volumes bounds the number of elements by $(1+2/\epsilon)^{d_n}\le(1+4/\epsilon)^{d_n}$.
If $H_n=0$, the class consists of the zero function on the corresponding support and the bound is immediate.
Taking the smaller of this bound and $N_k+1$ proves the stated covering bound.
\end{proof}
The bound on the number of estimated outcome bridges does not bound the magnitude of the estimated bridges.
Envelope conditions are imposed separately when establishing the uniform expansion and simultaneous inference.

\section{Derivation and Proof for Uniform Expansion with Estimated Bridges}
\label{app:uniform-expansion}
This section derives the uniform expansion in Theorem~\ref{thm:uniform-dr}. The derivation separates the population bias caused by bridge estimation errors from the sampling fluctuation caused by using estimated bridges across outcome thresholds. The assumptions below concern bridge estimation errors, boundedness, and covering entropy.
\par
For each fold $I_k$, let $n_k=|I_k|$ and define the evaluation-sample average
\begin{equation}
P_{n,k}f=\frac{1}{n_k}\sum_{i\in I_k}f(O_i).
\label{eq:app-evaluation-average}
\end{equation}
Since the folds partition $\{1,\ldots,n\}$,
\begin{equation}
\frac{1}{n}\sum_{i=1}^n f(O_i)
=\sum_{k=1}^K\frac{n_k}{n}P_{n,k}f.
\label{eq:app-fold-average}
\end{equation}
For each fixed $k$, condition on the training sample $\{O_i:i\notin I_k\}$. Then $\widehat h_{a,y}^{(-k)}(W,X)$ and $\widehat q_a^{(-k)}(Z,X)$ are fixed functions, while the observations $\{O_i:i\in I_k\}$ remain independent observations from $P_0$.
\subsection{Derivation of the Population Bias}
\label{app:population-bias}
Fix $a\in\{0,1\}$, $y\in\mathcal Y$, and a fold $I_k$. Starting from the expression used in \eqref{eq:method-cdf-estimator},
\begin{align}
&E\Big[
\widehat h_{a,y}^{(-k)}(W,X)
+\mathbf{1}\{A=a\}\widehat q_a^{(-k)}(Z,X)
\{\mathbf{1}\{Y\le y\}-\widehat h_{a,y}^{(-k)}(W,X)\}
\Big]-F_a(y)\nonumber\\
&\quad=
E\{\widehat h_{a,y}^{(-k)}(W,X)-h_{a,y}(W,X)\}\nonumber\\
&\qquad+
E\Big[
\mathbf{1}\{A=a\}\widehat q_a^{(-k)}(Z,X)
\{\mathbf{1}\{Y\le y\}-h_{a,y}(W,X)\}
\Big]\nonumber\\
&\qquad-
E\Big[
\mathbf{1}\{A=a\}\widehat q_a^{(-k)}(Z,X)
\{\widehat h_{a,y}^{(-k)}(W,X)-h_{a,y}(W,X)\}
\Big].
\label{eq:app-bias-step1}
\end{align}
The middle term in \eqref{eq:app-bias-step1} is zero. Because $\widehat q_a^{(-k)}(Z,X)$ is fixed after conditioning on the training sample and is a function of $(Z,X)$,
\begin{align}
&E\Big[
\mathbf{1}\{A=a\}\widehat q_a^{(-k)}(Z,X)
\{\mathbf{1}\{Y\le y\}-h_{a,y}(W,X)\}
\Big]\nonumber\\
&\quad=
E\Big[
\mathbf{1}\{A=a\}\widehat q_a^{(-k)}(Z,X)
E\{\mathbf{1}\{Y\le y\}-h_{a,y}(W,X)\mid A,Z,X\}
\Big]
=0,
\label{eq:app-outcome-bridge-zero}
\end{align}
where the last equality follows from the outcome bridge equation \eqref{eq:outcome-bridge}.
\par
The treatment bridge equation \eqref{eq:treatment-bridge} gives
\[
E\{\mathbf{1}\{A=a\}q_a(Z,X)\mid W,X\}=1.
\]
Therefore,
\begin{align}
&E\{\widehat h_{a,y}^{(-k)}(W,X)-h_{a,y}(W,X)\}\nonumber\\
&\quad=
E\Big[
\{\widehat h_{a,y}^{(-k)}(W,X)-h_{a,y}(W,X)\}
E\{\mathbf{1}\{A=a\}q_a(Z,X)\mid W,X\}
\Big]\nonumber\\
&\quad=
E\Big[
\mathbf{1}\{A=a\}q_a(Z,X)
\{\widehat h_{a,y}^{(-k)}(W,X)-h_{a,y}(W,X)\}
\Big].
\label{eq:app-treatment-bridge-correction}
\end{align}
Substituting \eqref{eq:app-outcome-bridge-zero} and \eqref{eq:app-treatment-bridge-correction} into \eqref{eq:app-bias-step1} gives
\begin{align}
&E\Big[
\widehat h_{a,y}^{(-k)}(W,X)
+\mathbf{1}\{A=a\}\widehat q_a^{(-k)}(Z,X)
\{\mathbf{1}\{Y\le y\}-\widehat h_{a,y}^{(-k)}(W,X)\}
\Big]-F_a(y)\nonumber\\
&\quad=
-E\Big[
\mathbf{1}\{A=a\}
\{\widehat q_a^{(-k)}(Z,X)-q_a(Z,X)\}
\{\widehat h_{a,y}^{(-k)}(W,X)-h_{a,y}(W,X)\}
\Big].
\label{eq:app-exact-product-bias}
\end{align}
Equation~\eqref{eq:app-exact-product-bias} is the population bias identity stated in \eqref{eq:method-bridge-product-bias}. Applying Cauchy--Schwarz gives
\begin{align}
&\left|
E\Big[
\mathbf{1}\{A=a\}
\{\widehat q_a^{(-k)}(Z,X)-q_a(Z,X)\}
\{\widehat h_{a,y}^{(-k)}(W,X)-h_{a,y}(W,X)\}
\Big]
\right|\nonumber\\
&\quad\le
\|\widehat q_a^{(-k)}(Z,X)-q_a(Z,X)\|_{a,Z}
\left[
E\Big\{
\mathbf{1}\{A=a\}
\{\widehat h_{a,y}^{(-k)}(W,X)-h_{a,y}(W,X)\}^2
\Big\}
\right]^{1/2}\nonumber\\
&\quad\le
\|\widehat q_a^{(-k)}(Z,X)-q_a(Z,X)\|_{a,Z}
\|\widehat h_{a,y}^{(-k)}(W,X)-h_{a,y}(W,X)\|_W.
\label{eq:app-product-bias-bound}
\end{align}
Taking the maximum over $k\in\{1,\ldots,K\}$ and the supremum over $y\in\mathcal Y$ yields
\begin{equation}
\max_k\sup_{y\in\mathcal Y}
\left|
E\Big[
\widehat h_{a,y}^{(-k)}(W,X)
+\mathbf{1}\{A=a\}\widehat q_a^{(-k)}(Z,X)
\{\mathbf{1}\{Y\le y\}-\widehat h_{a,y}^{(-k)}(W,X)\}
\Big]-F_a(y)
\right|
\le r_{h,a,n}r_{q,a,n}.
\label{eq:app-uniform-product-bias}
\end{equation}
The condition $r_{h,a,n}r_{q,a,n}=o_p(n^{-1/2})$ remains sufficient, but is not required by the sharper bound below.
\par
\noindent\textit{Retaining the finite-dimensional moment equations.}
The refinement uses the mixed-bias property and a tested moment error rather than two strong-norm errors, as in mixed-bias and projected-error analyses \citep{rotnitzky2021mixed,kallus2021causal,ghassami2022minimax,bennett2023source}.
Related weak-norm guarantees for neural proxy estimators are developed by \citet{bozkurt2026doubly}.
For the population finite-dimensional bridges in \eqref{eq:app-population-bridges}, the bridge equations and \eqref{eq:app-population-coefficient-equations} imply
\begin{equation}
\begin{aligned}
E\!\left[\mathbf{1}\{A=a\}b_Z(Z,X)
\{h_{a,y}^{(d_n)}(W,X)-h_{a,y}(W,X)\}\right]&=0,\\
E\!\left[\mathbf{1}\{A=a\}b_W(W,X)
\{q_a^{(d_n)}(Z,X)-q_a(Z,X)\}\right]&=0.
\end{aligned}
\label{eq:app-tested-approximation-orthogonality}
\end{equation}
Expanding each estimated bridge around the corresponding population finite-dimensional bridge, the cross terms between a coefficient estimation error and an approximation error vanish by \eqref{eq:app-tested-approximation-orthogonality}.
Consequently, conditionally on the training sample,
\begin{equation}
\begin{aligned}
&E\!\left[\mathbf{1}\{A=a\}
\{\widehat q_a^{(-k)}(Z,X)-q_a(Z,X)\}
\{\widehat h_{a,y}^{(-k)}(W,X)-h_{a,y}(W,X)\}\right]\\
&\quad=
\{\widehat\alpha_{a,\lambda_q}^{(-k)}-\alpha_a^{(d_n)}\}^{\top}
\Sigma_a
\{\widehat\theta_{a,\lambda_h}^{(-k)}(y)-\theta_a^{(d_n)}(y)\}\\
&\qquad+
E\!\left[\mathbf{1}\{A=a\}
\{q_a^{(d_n)}(Z,X)-q_a(Z,X)\}
\{h_{a,y}^{(d_n)}(W,X)-h_{a,y}(W,X)\}\right].
\end{aligned}
\label{eq:app-finite-dimensional-product-bias}
\end{equation}
The population bias in \eqref{eq:app-exact-product-bias} is the negative of \eqref{eq:app-finite-dimensional-product-bias}.
Cauchy--Schwarz applied after retaining $\Sigma_a$ gives
\begin{equation}
\begin{aligned}
&\max_k\sup_{y\in\mathcal Y}
\left|E\!\left[\mathbf{1}\{A=a\}
\{\widehat q_a^{(-k)}(Z,X)-q_a(Z,X)\}
\{\widehat h_{a,y}^{(-k)}(W,X)-h_{a,y}(W,X)\}\right]\right|\\
&\quad\le
\max_k\left\|\Sigma_a^\top
\{\widehat\alpha_{a,\lambda_q}^{(-k)}-\alpha_a^{(d_n)}\}\right\|_2
\max_k\sup_{y\in\mathcal Y}
\left\|\widehat\theta_{a,\lambda_h}^{(-k)}(y)-\theta_a^{(d_n)}(y)\right\|_2
+b_{h,a,n}b_{q,a,n}.
\end{aligned}
\label{eq:app-projected-product-bias-bound}
\end{equation}
No independence between the outcome and treatment coefficient estimates is used in \eqref{eq:app-projected-product-bias-bound}.
\par
\noindent\textit{Bounding the tested treatment coefficient error.}
The population equation $\Sigma_a^\top\alpha_a^{(d_n)}=\mu_W$ gives the exact identity
\begin{equation}
\begin{aligned}
\Sigma_a^\top\{\widehat\alpha_{a,\lambda_q}^{(-k)}-\alpha_a^{(d_n)}\}
&=(\Sigma_a-\widehat\Sigma_a^{(-k)})^\top\widehat\alpha_{a,\lambda_q}^{(-k)}\\
&\quad+(\widehat\Sigma_a^{(-k)})^\top\widehat\alpha_{a,\lambda_q}^{(-k)}-\widehat\mu_W^{(-k)}
+\widehat\mu_W^{(-k)}-\mu_W.
\end{aligned}
\label{eq:app-tested-treatment-error-identity}
\end{equation}
For a bound that retains the factor $M_{q,a,n}$ in the ridge term, multiply \eqref{eq:app-treatment-error-decomposition} by $(\widehat\Sigma_a^{(-k)})^\top$.
On the event $\mathcal E_n$ in \eqref{eq:app-matrix-event},
\[
\left\|(\widehat\Sigma_a^{(-k)})^\top
\{\widehat\Sigma_a^{(-k)}(\widehat\Sigma_a^{(-k)})^\top+\lambda_q I_{d_n}\}^{-1}
\widehat\Sigma_a^{(-k)}\right\|_{\mathrm{op}}\le1,
\]
and the transpose of the treatment matrix bound in \eqref{eq:app-regularized-moment-bounds} has norm at most $2/\kappa_{a,n}$.
Thus
\begin{equation}
\begin{aligned}
&\left\|\Sigma_a^\top
\{\widehat\alpha_{a,\lambda_q}^{(-k)}-\alpha_a^{(d_n)}\}\right\|_2\\
&\quad\le
\|\widehat\mu_W^{(-k)}-\mu_W\|_2
+M_{q,a,n}\|\widehat\Sigma_a^{(-k)}-\Sigma_a\|_{\mathrm{op}}
+\frac{2\lambda_qM_{q,a,n}}{\kappa_{a,n}}\\
&\qquad+
\|\widehat\Sigma_a^{(-k)}-\Sigma_a\|_{\mathrm{op}}
\|\widehat\alpha_{a,\lambda_q}^{(-k)}-\alpha_a^{(d_n)}\|_2.
\end{aligned}
\label{eq:app-tested-treatment-error-bound}
\end{equation}
Using \eqref{eq:app-moment-concentration}, \eqref{eq:app-treatment-coefficient-bound}, and $\sqrt{d_n/n}=o(\kappa_{a,n})$, the last term in \eqref{eq:app-tested-treatment-error-bound} is of smaller order than the preceding sampling and ridge bounds.
Since $P_0(\mathcal E_n)\to1$,
\begin{equation}
\max_k\left\|\Sigma_a^\top
\{\widehat\alpha_{a,\lambda_q}^{(-k)}-\alpha_a^{(d_n)}\}\right\|_2
=O_p\!\left(
(1+M_{q,a,n})\sqrt{\frac{d_n}{n}}
+\frac{\lambda_qM_{q,a,n}}{\kappa_{a,n}}
\right).
\label{eq:app-tested-treatment-error-rate}
\end{equation}
The sampling term in \eqref{eq:app-tested-treatment-error-rate} has no factor $\kappa_{a,n}^{-1}$.
\par
\noindent\textit{Proof of Corollary~\ref{cor:method-weak-proxy-bias}.}
Combining \eqref{eq:app-tested-treatment-error-rate} with the outcome coefficient bound in \eqref{eq:app-outcome-coefficient-bound} and the moment concentration rates gives
\begin{equation}
\begin{aligned}
&\max_k\left\|\Sigma_a^\top
\{\widehat\alpha_{a,\lambda_q}^{(-k)}-\alpha_a^{(d_n)}\}\right\|_2
\max_k\sup_{y\in\mathcal Y}
\left\|\widehat\theta_{a,\lambda_h}^{(-k)}(y)-\theta_a^{(d_n)}(y)\right\|_2\\
&\quad=O_p\!\left[
\left\{(1+M_{q,a,n})\sqrt{\frac{d_n}{n}}
+\frac{\lambda_qM_{q,a,n}}{\kappa_{a,n}}\right\}
\left\{\frac{1+M_{h,a,n}}{\kappa_{a,n}}\sqrt{\frac{d_n}{n}}
+\frac{\lambda_hM_{h,a,n}}{\kappa_{a,n}^2}\right\}
\right].
\end{aligned}
\label{eq:app-projected-product-rate}
\end{equation}
Under the bounded coefficient norm condition and the $o\{\kappa_{a,n}\sqrt{d_n/n}\}$ ridge condition in Corollary~\ref{cor:method-weak-proxy-bias}, the first factor in \eqref{eq:app-projected-product-rate} is $O_p(\sqrt{d_n/n})$ and the second is $O_p\{\sqrt{d_n/n}/\kappa_{a,n}\}$. Thus their product is $O_p\{d_n/(n\kappa_{a,n})\}$.
Equation~\eqref{eq:app-projected-product-bias-bound} then proves \eqref{eq:method-sharper-bias-rate}.
Multiplication by $\sqrt n$ and the fixed number of treatment arms prove the stated sufficient condition for negligible population bias.
Approximation errors enter \eqref{eq:app-projected-product-bias-bound} only through $b_{h,a,n}b_{q,a,n}$, but individual approximation errors still enter the strong bridge errors used in the empirical-process bounds.
The argument retains $\sqrt{d_n/n}=o(\kappa_{a,n})$ and does not establish inference when $\kappa_{a,n}$ is of the same order as the sample matrix error.
\par
\noindent\textit{Retaining both ridge parameters in the population-bias bound.}
The coefficient product in \eqref{eq:app-finite-dimensional-product-bias} can also be bounded before separating the two regularized matrix factors. This calculation uses exactly the estimates in \eqref{eq:method-regularized-bridges}. In particular,
\begin{equation}
\begin{aligned}
&\{\widehat\alpha_{a,\lambda_q}^{(-k)}-\alpha_a^{(d_n)}\}^{\top}
\Sigma_a\{\widehat\theta_{a,\lambda_h}^{(-k)}(y)-\theta_a^{(d_n)}(y)\}\\
&=\Big[\widehat\Sigma_a^{(-k)}
\{\widehat\mu_W^{(-k)}-(\widehat\Sigma_a^{(-k)})^\top\alpha_a^{(d_n)}\}
-\lambda_q\alpha_a^{(d_n)}\Big]^\top\\
&\quad\times\{\widehat\Sigma_a^{(-k)}(\widehat\Sigma_a^{(-k)})^\top+\lambda_q I_{d_n}\}^{-1}
\widehat\Sigma_a^{(-k)}
\{(\widehat\Sigma_a^{(-k)})^\top\widehat\Sigma_a^{(-k)}+\lambda_h I_{d_n}\}^{-1}\\
&\quad\times\Big[(\widehat\Sigma_a^{(-k)})^\top
\{\widehat\gamma_a^{(-k)}(y)-\widehat\Sigma_a^{(-k)}\theta_a^{(d_n)}(y)\}
-\lambda_h\theta_a^{(d_n)}(y)\Big]\\
&\quad-\{\widehat\alpha_{a,\lambda_q}^{(-k)}-\alpha_a^{(d_n)}\}^{\top}
(\widehat\Sigma_a^{(-k)}-\Sigma_a)
\{\widehat\theta_{a,\lambda_h}^{(-k)}(y)-\theta_a^{(d_n)}(y)\}.
\end{aligned}
\label{eq:app-ridge-product-identity}
\end{equation}
Equation~\eqref{eq:app-ridge-product-identity} follows by substituting \eqref{eq:app-outcome-error-decomposition} and \eqref{eq:app-treatment-error-decomposition} into the product containing $\widehat\Sigma_a^{(-k)}$, and then subtracting the product containing $\widehat\Sigma_a^{(-k)}-\Sigma_a$. Expanding the two square brackets gives four matrix products. Singular value decomposition shows that their operator norms are the maxima of the fractions with numerators $\sigma_j^3$, $\sigma_j^2$, $\sigma_j^2$, and $\sigma_j$, respectively, and denominator $(\sigma_j^2+\lambda_h)(\sigma_j^2+\lambda_q)$, with every singular value evaluated at $\widehat\Sigma_a^{(-k)}$. Thus
\begin{equation}
\begin{aligned}
&\max_k\sup_{y\in\mathcal Y}
\left|E\!\left[\mathbf{1}\{A=a\}
\{\widehat q_a^{(-k)}(Z,X)-q_a(Z,X)\}
\{\widehat h_{a,y}^{(-k)}(W,X)-h_{a,y}(W,X)\}\right]\right|\\
&\le b_{h,a,n}b_{q,a,n}+\max_k\sup_{y\in\mathcal Y}\Bigg[\\
&\quad\max_j\frac{\sigma_j(\widehat\Sigma_a^{(-k)})^3}
{\{\sigma_j(\widehat\Sigma_a^{(-k)})^2+\lambda_h\}
 \{\sigma_j(\widehat\Sigma_a^{(-k)})^2+\lambda_q\}}\\
&\qquad\times\|\widehat\mu_W^{(-k)}-(\widehat\Sigma_a^{(-k)})^\top\alpha_a^{(d_n)}\|_2
\|\widehat\gamma_a^{(-k)}(y)-\widehat\Sigma_a^{(-k)}\theta_a^{(d_n)}(y)\|_2\\
&\quad+\max_j\frac{\sigma_j(\widehat\Sigma_a^{(-k)})^2}
{\{\sigma_j(\widehat\Sigma_a^{(-k)})^2+\lambda_h\}
 \{\sigma_j(\widehat\Sigma_a^{(-k)})^2+\lambda_q\}}\\
&\qquad\times\Big\{
\lambda_h\|\theta_a^{(d_n)}(y)\|_2
\|\widehat\mu_W^{(-k)}-(\widehat\Sigma_a^{(-k)})^\top\alpha_a^{(d_n)}\|_2\\
&\hspace{29mm}+\lambda_q\|\alpha_a^{(d_n)}\|_2
\|\widehat\gamma_a^{(-k)}(y)-\widehat\Sigma_a^{(-k)}\theta_a^{(d_n)}(y)\|_2\Big\}\\
&\quad+\lambda_h\lambda_q\max_j
\frac{\sigma_j(\widehat\Sigma_a^{(-k)})}
{\{\sigma_j(\widehat\Sigma_a^{(-k)})^2+\lambda_h\}
 \{\sigma_j(\widehat\Sigma_a^{(-k)})^2+\lambda_q\}}
\|\theta_a^{(d_n)}(y)\|_2\|\alpha_a^{(d_n)}\|_2\\
&\quad+\|\widehat\Sigma_a^{(-k)}-\Sigma_a\|_{\mathrm{op}}
\|\widehat\alpha_{a,\lambda_q}^{(-k)}-\alpha_a^{(d_n)}\|_2
\|\widehat\theta_{a,\lambda_h}^{(-k)}(y)-\theta_a^{(d_n)}(y)\|_2
\Bigg].
\end{aligned}
\label{eq:app-ridge-product-bias-bound}
\end{equation}
The coefficient errors in the last line can be replaced directly by \eqref{eq:app-all-ridge-coefficient-bounds}. No favorable alignment of the two errors or independence of the two estimated bridges is required.
\par
Both \eqref{eq:app-ridge-product-identity} and \eqref{eq:app-ridge-product-bias-bound} hold for all positive ridge parameters, including singular sample matrices. The bound \eqref{eq:app-projected-product-bias-bound} also remains valid. Taking the smaller of the right-hand sides of \eqref{eq:app-projected-product-bias-bound} and \eqref{eq:app-ridge-product-bias-bound} therefore never enlarges the population-bias bound in \eqref{eq:app-projected-product-bias-bound}. The bound \eqref{eq:app-ridge-product-bias-bound} alone need not be smaller on every sample.
\par
For example, when $\lambda_h=\lambda_q$, elementary differentiation gives
\[
\sup_{t\ge0}\frac{t^3}{(t^2+\lambda_h)^2}
=\frac{3\sqrt3}{16\sqrt{\lambda_h}},
\]
with the maximum attained at $t^2=3\lambda_h$. On $\mathcal E_n$,
\begin{equation}
\max_j\frac{\sigma_j(\widehat\Sigma_a^{(-k)})^3}
{\{\sigma_j(\widehat\Sigma_a^{(-k)})^2+\lambda_h\}^2}
\le\min\left\{\frac{2}{\kappa_{a,n}},\frac{3\sqrt3}{16\sqrt{\lambda_h}}\right\}.
\label{eq:app-ridge-sampling-factor}
\end{equation}
The sampling-error product can consequently have a smaller amplification factor when ridge is not negligible. The ridge-bias and matrix-error terms in \eqref{eq:app-ridge-product-bias-bound} cannot be omitted. These deterministic bounds do not establish root-$n$ inference when $\kappa_{a,n}$ is of the same order as the sample matrix error.
\par
\noindent\textit{A second-order expansion of the population bias.}
Consider the subcase of Corollary~\ref{cor:method-weak-proxy-bias} satisfying
\begin{equation}
\lambda_h,\lambda_q=o(\kappa_{a,n}^2),
\qquad
\lambda_hM_{h,a,n}+\lambda_qM_{q,a,n}
=o\!\left(\kappa_{a,n}\sqrt{\frac{d_n}{n}}\right).
\label{eq:app-bias-expansion-ridge}
\end{equation}
The ridge conditions in \eqref{eq:app-bias-expansion-ridge} allow ridge terms to be omitted from the leading signed term.
\par
The population moment equations give the exact centered-moment identities
\begin{equation}
\begin{aligned}
\widehat\gamma_a^{(-k)}(y)-\widehat\Sigma_a^{(-k)}\theta_a^{(d_n)}(y)
&=(P_{n,-k}-E)\!\left[\mathbf{1}\{A=a\}b_Z(Z,X)
\{\mathbf{1}\{Y\le y\}-h_{a,y}^{(d_n)}(W,X)\}\right],\\
\widehat\mu_W^{(-k)}-(\widehat\Sigma_a^{(-k)})^\top\alpha_a^{(d_n)}
&=(P_{n,-k}-E)\!\left[b_W(W,X)
\{1-\mathbf{1}\{A=a\}q_a^{(d_n)}(Z,X)\}\right].
\end{aligned}
\label{eq:app-centered-bridge-moments}
\end{equation}
By \eqref{eq:app-moment-concentration}, \eqref{eq:app-threshold-moment-rate}, and bounded $M_{h,a,n},M_{q,a,n}$, both centered vectors are $O_p(\sqrt{d_n/n})$, uniformly over folds and, for the outcome vector, thresholds.
\par
On $\mathcal E_n$, subtracting the sample inverse from the population inverse gives
\[
\|\{\widehat\Sigma_a^{(-k)}\}^{-1}-\Sigma_a^{-1}\|_{\mathrm{op}}
\le\frac{2}{\kappa_{a,n}^2}
\|\widehat\Sigma_a^{(-k)}-\Sigma_a\|_{\mathrm{op}}.
\]
The difference between the outcome ridge matrix factor and the sample inverse has operator norm at most $8\lambda_h/\kappa_{a,n}^3$. Applying these two bounds to \eqref{eq:app-outcome-error-decomposition} gives
\begin{equation}
\begin{aligned}
&\max_k\sup_{y\in\mathcal Y}
\Big\|\widehat\theta_{a,\lambda_h}^{(-k)}(y)-\theta_a^{(d_n)}(y)
-\Sigma_a^{-1}\{\widehat\gamma_a^{(-k)}(y)
-\widehat\Sigma_a^{(-k)}\theta_a^{(d_n)}(y)\}\Big\|_2\\
&\quad=O_p\!\left(
\frac{d_n}{n\kappa_{a,n}^2}
+\frac{\lambda_h}{\kappa_{a,n}^3}\sqrt{\frac{d_n}{n}}
+\frac{\lambda_hM_{h,a,n}}{\kappa_{a,n}^2}\right)
=o_p\!\left(\frac{1}{\kappa_{a,n}}\sqrt{\frac{d_n}{n}}\right).
\end{aligned}
\label{eq:app-outcome-second-order-coefficient}
\end{equation}
For the treatment coefficient error, multiply \eqref{eq:app-treatment-error-decomposition} by $(\widehat\Sigma_a^{(-k)})^\top$ and use
\[
(\widehat\Sigma_a^{(-k)})^\top
\{\widehat\Sigma_a^{(-k)}(\widehat\Sigma_a^{(-k)})^\top+\lambda_qI_{d_n}\}^{-1}
\widehat\Sigma_a^{(-k)}
=I_{d_n}-\lambda_q\{(\widehat\Sigma_a^{(-k)})^\top\widehat\Sigma_a^{(-k)}+\lambda_qI_{d_n}\}^{-1}.
\]
The remaining difference between $\Sigma_a^\top$ and $(\widehat\Sigma_a^{(-k)})^\top$ multiplies the coefficient error. Consequently,
\begin{equation}
\begin{aligned}
&\max_k\Big\|\Sigma_a^\top\{\widehat\alpha_{a,\lambda_q}^{(-k)}-\alpha_a^{(d_n)}\}
-\{\widehat\mu_W^{(-k)}-(\widehat\Sigma_a^{(-k)})^\top\alpha_a^{(d_n)}\}\Big\|_2\\
&\quad=O_p\!\left(
\frac{\lambda_q}{\kappa_{a,n}^2}\sqrt{\frac{d_n}{n}}
+\frac{\lambda_qM_{q,a,n}}{\kappa_{a,n}}
+\frac{d_n}{n\kappa_{a,n}}\right)
=o_p\!\left(\sqrt{\frac{d_n}{n}}\right).
\end{aligned}
\label{eq:app-treatment-second-order-coefficient}
\end{equation}
The last equalities in \eqref{eq:app-outcome-second-order-coefficient} and \eqref{eq:app-treatment-second-order-coefficient} follow from \eqref{eq:app-bias-expansion-ridge} and $\sqrt{d_n/n}=o(\kappa_{a,n})$. Substituting both expansions into \eqref{eq:app-finite-dimensional-product-bias} yields
\begin{equation}
\begin{aligned}
&E\!\left[\mathbf{1}\{A=a\}
\{\widehat q_a^{(-k)}(Z,X)-q_a(Z,X)\}
\{\widehat h_{a,y}^{(-k)}(W,X)-h_{a,y}(W,X)\}\right]\\
&=\{\widehat\mu_W^{(-k)}-(\widehat\Sigma_a^{(-k)})^\top\alpha_a^{(d_n)}\}^\top
\Sigma_a^{-1}
\{\widehat\gamma_a^{(-k)}(y)-\widehat\Sigma_a^{(-k)}\theta_a^{(d_n)}(y)\}\\
&\quad+E\!\left[\mathbf{1}\{A=a\}
\{q_a^{(d_n)}(Z,X)-q_a(Z,X)\}
\{h_{a,y}^{(d_n)}(W,X)-h_{a,y}(W,X)\}\right]
+o_p\!\left(\frac{d_n}{n\kappa_{a,n}}\right).
\end{aligned}
\label{eq:app-bias-second-order-expansion}
\end{equation}
The remainder is uniform over thresholds and the fixed number of folds. Both treatment arms satisfy the same statement. The signed population bias in \eqref{eq:method-bridge-product-bias} is the negative of the right-hand side of \eqref{eq:app-bias-second-order-expansion}. The leading term pairs the two centered training-sample moment errors through exactly one $\Sigma_a^{-1}$. The two centered errors use the same training observations and are not assumed independent.
\par
\noindent\textit{Sharpness of the $\kappa_{a,n}^{-1}$ factor for the specified estimates.}
We construct a sequence of data distributions satisfying the conditions of Corollary~\ref{cor:method-weak-proxy-bias}. Let $X$ be constant, let $U\in\{-1,1\}$ have equal probabilities, and, for $n>16$, generate $W,Z\in\{-1,1\}$ independently conditional on $U$ with
\[
P(W=w\mid U=u)=\frac{1+uw/2}{2},
\qquad
P(Z=z\mid U=u)=\frac{1+2n^{-1/4}uz}{2}.
\]
Then
\[
P(W=w,Z=z)=\frac{1+n^{-1/4}wz}{4}.
\]
Generate $A\sim\operatorname{Bernoulli}(1/2)$ independently of $(U,W,Z)$. Independently generate $Y(0)$ and $Y(1)$ from the uniform distribution on $[0,1]$, and set $Y=AY(1)+(1-A)Y(0)$. Latent positivity and all three conditional independence statements in Section~\ref{sec:notation} follow from this construction. Also, the two conditional distributions of $U$ given $Z$ are distinct, so the corresponding binary conditional expectation matrix is invertible at every $n>16$. The marginal counterfactual distribution is $F_a(y)=y$ on $[0,1]$.
\par
Choose $\mathcal Y=[1/4,3/4]$, $d_n=2$, $b_W(W,X)=(1,W)^\top$, and $b_Z(Z,X)=(1,Z)^\top$. For both treatment arms,
\begin{equation}
\Sigma_a=\frac12\begin{pmatrix}1&0\\0&n^{-1/4}\end{pmatrix},
\qquad
\kappa_{a,n}=\frac{n^{-1/4}}2,
\qquad
\theta_a^{(2)}(y)=\begin{pmatrix}y\\0\end{pmatrix},
\qquad
\alpha_a^{(2)}=\begin{pmatrix}2\\0\end{pmatrix}.
\label{eq:app-sharpness-model}
\end{equation}
Thus $h_{a,y}(W,X)=y$, $q_a(Z,X)=2$, $b_{h,a,n}=b_{q,a,n}=0$, and the population coefficient norms remain bounded. The two Gram matrices are $I_2$ and $I_2/2$. Bounded fixed-dimensional observations give \eqref{eq:app-moment-concentration}, and $\sqrt{d_n/n}=o(\kappa_{a,n})$. With $\lambda_h=\lambda_q=0.1/N_k$ in each training sample, \eqref{eq:app-bias-expansion-ridge} holds because $N_k\asymp n$.
\par
Fix $a=1$, $y=1/2$, and one fold $I_k$. The second coordinates of the centered vectors in \eqref{eq:app-centered-bridge-moments} are
\[
P_{n,-k}\{W(1-2A)\},
\qquad
P_{n,-k}\!\left[A Z\{\mathbf{1}\{Y\le1/2\}-1/2\}\right].
\]
Both summands have mean zero. Their variances are $1$ and $1/8$, and their covariance is zero because $Y$ is independent of $(A,W,Z)$ and $E\{\mathbf{1}\{Y\le1/2\}-1/2\}=0$. The bounded triangular-array central limit theorem gives
\begin{equation}
\sqrt{N_k}
\begin{pmatrix}
P_{n,-k}\{W(1-2A)\}\\
P_{n,-k}[A Z\{\mathbf{1}\{Y\le1/2\}-1/2\}]
\end{pmatrix}
\rightsquigarrow
\begin{pmatrix}G_q\\G_h\end{pmatrix}
\sim\mathcal N\!\left(0,\begin{pmatrix}1&0\\0&1/8\end{pmatrix}\right).
\label{eq:app-sharpness-moment-limit}
\end{equation}
Here $G_q$ and $G_h$ denote the two coordinates of this limiting vector. The first coordinates of the two moment vectors contribute $O_p(N_k^{-1})$ to \eqref{eq:app-bias-second-order-expansion}, which is negligible after multiplication by $N_k\kappa_{1,n}$. The second diagonal entry of $\Sigma_1^{-1}$ equals $\kappa_{1,n}^{-1}$. Hence
\begin{equation}
N_k\kappa_{1,n}
E\!\left[\mathbf{1}\{A=1\}
\{\widehat q_1^{(-k)}(Z,X)-q_1(Z,X)\}
\{\widehat h_{1,1/2}^{(-k)}(W,X)-h_{1,1/2}(W,X)\}\right]
\rightsquigarrow G_qG_h.
\label{eq:app-sharpness-bias-limit}
\end{equation}
The limit is not concentrated at zero. Thus the absolute population bias is not $o_p\{(N_k\kappa_{1,n})^{-1}\}$, even for this fixed dimension and zero approximation error. Since $N_k\asymp n$ and the uniform maximum dominates the fixed fold and threshold, the factor $\kappa_{a,n}^{-1}$ in Corollary~\ref{cor:method-weak-proxy-bias} cannot be replaced uniformly by a smaller order for these estimates. This example does not establish optimality of the factor $d_n$ when $d_n$ grows, or a lower bound over all estimators.
\par
The true bridges in the construction also satisfy Assumption~\ref{ass:process}. Proposition~\ref{prop:method-bridge-rates} gives bridge errors $O_p(n^{-1/4})$ and, because the basis vectors are bounded, fixed estimated-bridge envelopes are valid with probability tending to one. The population-bias bound is $O_p(n^{-3/4})$. Taking $\delta_n=n^{-1/8}$ and fixed sufficiently large $H_n,Q_n$ verifies the remaining rate conditions in Assumption~\ref{ass:nuisance}. On $\mathcal Y=[1/4,3/4]$, $E\{\varphi_{a,y}^2(O)\}=2y(1-y)$ is bounded away from zero.
\subsection{Conditions for Uniform Control across Outcome Thresholds}
\label{app:uniform-conditions}
For a class of functions $\mathcal F$, let
\[
N\{\epsilon,\mathcal F,L_2(\mathsf Q)\}
\]
denote the $L_2(\mathsf Q)$ covering number at radius $\epsilon$, where $\mathsf Q$ is a finitely supported probability measure.
\begin{assumption}[Bridge function-class conditions]
\label{ass:process}
For each treatment arm $a\in\{0,1\}$, the following conditions hold.
\begin{enumerate}[label=(\roman*)]
\item There exist fixed constants $H,Q<\infty$ such that
\[
\sup_{y\in\mathcal Y}|h_{a,y}(W,X)|\le H
\]
almost surely and
\[
|q_a(Z,X)|\le Q
\]
almost surely on $\{A=a\}$.
\item The class $\{h_{a,y}(W,X):y\in\mathcal Y\}$ is pointwise measurable and there exist fixed constants $A_{\mathrm{ent}},v<\infty$ such that, for every finitely supported probability measure $\mathsf Q$,
\begin{equation}
\log N\left\{\epsilon H,\{h_{a,y}(W,X):y\in\mathcal Y\},L_2(\mathsf Q)\right\}
\le v\log(A_{\mathrm{ent}}/\epsilon),
\qquad 0<\epsilon\le1.
\label{eq:app-true-bridge-entropy}
\end{equation}
\end{enumerate}
\end{assumption}
\begin{assumption}[Cross-fitted bridge estimation conditions]
\label{ass:nuisance}
The number of folds $K$ is fixed, the partition is independent of $O_1,\ldots,O_n$, and $n_k/n$ is bounded away from zero for every $k$. The estimated bridges are measurable with respect to the corresponding training sample. The following conditions hold for both treatment arms and uniformly over folds.
\begin{enumerate}[label=(\roman*)]
\item For the finite-dimensional bridge estimates in \eqref{eq:method-estimated-bridges}, the bridge estimation errors in \eqref{eq:method-uniform-bridge-rates} satisfy $r_{h,a,n}=o_p(1)$ and $r_{q,a,n}=o_p(1)$, and
\begin{equation}
\begin{aligned}
&\max_k\left\|\Sigma_a^\top
\{\widehat\alpha_{a,\lambda_q}^{(-k)}-\alpha_a^{(d_n)}\}\right\|_2
\max_k\sup_{y\in\mathcal Y}
\left\|\widehat\theta_{a,\lambda_h}^{(-k)}(y)-\theta_a^{(d_n)}(y)\right\|_2\\
&\qquad+b_{h,a,n}b_{q,a,n}=o_p(n^{-1/2}).
\end{aligned}
\label{eq:app-product-rate}
\end{equation}
\item There exist deterministic sequences $H_n\ge H$ and $Q_n\ge Q$ such that, on training-sample events whose probabilities tend to one,
\begin{equation}
\sup_{y\in\mathcal Y}
|\widehat h_{a,y}^{(-k)}(W,X)|
\le H_n,
\qquad
|\widehat q_a^{(-k)}(Z,X)|
\le Q_n
\label{eq:app-estimated-bridge-bounds}
\end{equation}
almost surely on the corresponding supports. For these estimates, the covering bound is
\begin{equation}
\begin{aligned}
&\sup_{\mathsf Q}\log N\!\left\{\epsilon H_n,
 \{\widehat h_{a,y}^{(-k)}(W,X):y\in\mathcal Y\},L_2(\mathsf Q)\right\}\\
&\qquad\le\min\{\log(N_k+1),\,d_n\log(1+4/\epsilon)\},
 \qquad 0<\epsilon\le1.
\end{aligned}
\label{eq:app-estimated-bridge-entropy}
\end{equation}
\item Define
\begin{equation}
B_n=(1+Q_n)(H_n+H).
\label{eq:app-Bn}
\end{equation}
There exists a deterministic sequence $0\le\delta_n\le\min\{1,B_n\}$ with $\delta_n\to0$ such that
\begin{equation}
\max_{a\in\{0,1\}}
(1+Q_n)r_{h,a,n}
\le\delta_n
\label{eq:app-delta}
\end{equation}
with probability tending to one. For a sufficiently large fixed constant $C\ge e$ and $\delta_n>0$, define
\begin{equation}
\Lambda_n=1+\min\!\left\{\log(n+1),\,
 d_n\log\left(\frac{CB_n}{\delta_n}\right)\right\}
 +v\log\left(\frac{CB_n}{\delta_n}\right).
\label{eq:app-Lambda}
\end{equation}
When $\delta_n=0$, set $\Lambda_n=0$.
The quantities $B_n$, $\delta_n$, and $\Lambda_n$ control the first term in \eqref{eq:app-estimated-minus-true}; the treatment bridge estimation error is controlled separately by $r_{q,a,n}=o_p(1)$ in (i).
The sequences satisfy
\begin{equation}
\delta_n\sqrt{\Lambda_n}
+\frac{B_n\Lambda_n}{\sqrt n}
\longrightarrow0.
\label{eq:app-empirical-rate}
\end{equation}
\end{enumerate}
\end{assumption}
Proposition~\ref{prop:method-bridge-rates} provides explicit sufficient rates for $r_{h,a,n}$ and $r_{q,a,n}$. The bound \eqref{eq:app-projected-product-rate} and Corollary~\ref{cor:method-weak-proxy-bias}, with its stated ridge condition, provide sufficient conditions for \eqref{eq:app-product-rate}, without requiring $r_{h,a,n}r_{q,a,n}=o_p(n^{-1/2})$. The deterministic bound \eqref{eq:app-ridge-product-bias-bound} is an additional finite-sample refinement. In addition, Lemma~\ref{lem:app-bridge-class-size} gives \eqref{eq:app-estimated-bridge-entropy} by combining the at most $N_k+1$ distinct functions with their at most $d_n$-dimensional linear span. Thus this covering bound follows from the specified estimates and \eqref{eq:app-estimated-bridge-bounds}.
\subsection{Derivation of the Estimated-Bridge Sampling Fluctuation}
\label{app:estimated-bridge-fluctuation}
The function $\varphi_{a,y}(O)$ in \eqref{eq:method-influence-function} has mean zero. Indeed,
\begin{align}
E\{\varphi_{a,y}(O)\}
&=E\{h_{a,y}(W,X)-F_a(y)\}\nonumber\\
&\quad+
E\Big[
\mathbf{1}\{A=a\}q_a(Z,X)
\{\mathbf{1}\{Y\le y\}-h_{a,y}(W,X)\}
\Big]
=0,
\label{eq:app-influence-mean-zero}
\end{align}
where the first term is zero by \eqref{eq:identification-Fa} and the second term is zero by \eqref{eq:outcome-bridge}.
\par
For a fixed fold $I_k$, define $d_{a,y}^{(-k)}(O)$ by subtracting \eqref{eq:method-influence-function} from the corresponding expression with estimated bridges:
\begin{align}
&\Big[
\widehat h_{a,y}^{(-k)}(W,X)-F_a(y)
+\mathbf{1}\{A=a\}\widehat q_a^{(-k)}(Z,X)
\{\mathbf{1}\{Y\le y\}-\widehat h_{a,y}^{(-k)}(W,X)\}
\Big]
-\varphi_{a,y}(O)\nonumber\\
&\quad=
\{\widehat h_{a,y}^{(-k)}(W,X)-h_{a,y}(W,X)\}\nonumber\\
&\qquad+
\mathbf{1}\{A=a\}
\Big[
\widehat q_a^{(-k)}(Z,X)
\{\mathbf{1}\{Y\le y\}-\widehat h_{a,y}^{(-k)}(W,X)\}
-q_a(Z,X)
\{\mathbf{1}\{Y\le y\}-h_{a,y}(W,X)\}
\Big]\nonumber\\
&\quad=
\{1-\mathbf{1}\{A=a\}\widehat q_a^{(-k)}(Z,X)\}
\{\widehat h_{a,y}^{(-k)}(W,X)-h_{a,y}(W,X)\}\nonumber\\
&\qquad+
\mathbf{1}\{A=a\}
\{\widehat q_a^{(-k)}(Z,X)-q_a(Z,X)\}
\{\mathbf{1}\{Y\le y\}-h_{a,y}(W,X)\}.
\label{eq:app-estimated-minus-true}
\end{align}
Equation~\eqref{eq:app-estimated-minus-true} shows separately how the two bridge estimation errors enter the sampling fluctuation.
\par
Under \eqref{eq:app-estimated-bridge-bounds} and Assumption~\ref{ass:process}(i), the absolute value of the first term on the right side of \eqref{eq:app-estimated-minus-true} is at most
\[
(1+Q_n)(H_n+H),
\]
which is $B_n$ in \eqref{eq:app-Bn}. The second term has the smaller, training-sample-dependent envelope
\[
(1+H)\mathbf{1}\{A=a\}
 |\widehat q_a^{(-k)}(Z,X)-q_a(Z,X)|
\le (Q_n+Q)(1+H).
\]
Thus the whole difference is bounded by $B_n+(Q_n+Q)(1+H)$, but only its first term will use the bounded-envelope local maximal inequality.
\par
The $L_2(P_0)$ norm of the first term on the right side of \eqref{eq:app-estimated-minus-true} satisfies
\begin{align}
&\left[
E\Big[
\{1-\mathbf{1}\{A=a\}\widehat q_a^{(-k)}(Z,X)\}^2
\{\widehat h_{a,y}^{(-k)}(W,X)-h_{a,y}(W,X)\}^2
\Big]
\right]^{1/2}\nonumber\\
&\quad\le
(1+Q_n)
\|\widehat h_{a,y}^{(-k)}(W,X)-h_{a,y}(W,X)\|_W.
\label{eq:app-first-L2-bound}
\end{align}
The $L_2(P_0)$ norm of the second term satisfies
\begin{align}
&\left[
E\Big[
\mathbf{1}\{A=a\}
\{\widehat q_a^{(-k)}(Z,X)-q_a(Z,X)\}^2
\{\mathbf{1}\{Y\le y\}-h_{a,y}(W,X)\}^2
\Big]
\right]^{1/2}\nonumber\\
&\quad\le
(1+H)
\|\widehat q_a^{(-k)}(Z,X)-q_a(Z,X)\|_{a,Z}.
\label{eq:app-second-L2-bound}
\end{align}
Combining \eqref{eq:app-first-L2-bound} and \eqref{eq:app-second-L2-bound}, then taking the maximum over folds and the supremum over thresholds, gives
\begin{equation}
\max_k\sup_{y\in\mathcal Y}
\left[E\{d_{a,y}^{(-k)}(O)^2\}\right]^{1/2}
\le
(1+Q_n)r_{h,a,n}+(1+H)r_{q,a,n}
\le\delta_n+(1+H)r_{q,a,n}
\label{eq:app-estimated-L2-radius}
\end{equation}
with probability tending to one.
\par
It remains to control the two terms in \eqref{eq:app-estimated-minus-true} separately as $y$ varies.
For the second term, the class
$\{\mathbf{1}\{Y\le y\}-h_{a,y}(W,X):y\in\mathcal Y\}$
contains only the threshold indicator and the true outcome bridge.
The VC covering bound for the threshold indicators and \eqref{eq:app-true-bridge-entropy} give
\begin{equation}
\begin{aligned}
&\sup_{\mathsf Q}\log N\!\left\{\epsilon(1+H),
 \{\mathbf{1}\{Y\le y\}-h_{a,y}(W,X):y\in\mathcal Y\},L_2(\mathsf Q)\right\}\\
&\qquad\le C(v+1)\log(C/\epsilon),\qquad 0<\epsilon\le1.
\end{aligned}
\label{eq:app-estimated-entropy-form}
\end{equation}
Conditionally on the training sample, multiplication by
$\mathbf{1}\{A=a\}\{\widehat q_a^{(-k)}(Z,X)-q_a(Z,X)\}$
preserves this polynomial entropy when the covering radius is normalized by the $L_2(\mathsf Q)$ norm of the second term's envelope.
Indeed, reweight any finitely supported $\mathsf Q$ by
$\mathbf{1}\{A=a\}\{\widehat q_a^{(-k)}(Z,X)-q_a(Z,X)\}^2$
and normalize to a probability measure; weighted $L_2(\mathsf Q)$ distances become unweighted distances under that measure, multiplied by the norm of the weight.
If the normalizing constant is zero, every weighted function is zero $\mathsf Q$-almost surely.
The resulting uniform entropy integral over $(0,1)$ is bounded by a fixed constant.
The uniform entropy maximal inequality \citep[Section~2.14]{vandervaartwellner1996} therefore gives
\[
\begin{aligned}
&E\!\Biggl[\sup_{y\in\mathcal Y}\Biggl|\sqrt{n_k}(P_{n,k}-E)
 \Bigl[\mathbf{1}\{A=a\}\{\widehat q_a^{(-k)}(Z,X)-q_a(Z,X)\}\\
&\hspace{27mm}\times\{\mathbf{1}\{Y\le y\}-h_{a,y}(W,X)\}\Bigr]\Biggr|
 \,\Bigm|\,\text{training sample}\Biggr]\\
&\qquad\le C\|\widehat q_a^{(-k)}(Z,X)-q_a(Z,X)\|_{a,Z}.
\end{aligned}
\]
In these maximal bounds, $C$ depends only on the fixed true-bridge envelope and entropy constants, not on the number or dimension of the estimated outcome bridges.
In particular, no estimated outcome bridge complexity multiplies the treatment bridge estimation error.
\par
For the first term, the fixed multiplier
$1-\mathbf{1}\{A=a\}\widehat q_a^{(-k)}(Z,X)$
is bounded in absolute value by $1+Q_n$.
Covering the estimated and true outcome bridges separately, using
\eqref{eq:app-estimated-bridge-entropy} and \eqref{eq:app-true-bridge-entropy}, yields
\begin{equation}
\begin{aligned}
&\sup_{\mathsf Q}\log N\!\Bigl\{\epsilon B_n,\,
 \bigl\{(1-\mathbf{1}\{A=a\}\widehat q_a^{(-k)}(Z,X))
 (\widehat h_{a,y}^{(-k)}(W,X)-h_{a,y}(W,X)):y\in\mathcal Y\bigr\},\,
 L_2(\mathsf Q)\Bigr\}\\
&\qquad\le\min\{\log(N_k+1),\,d_n\log(C/\epsilon)\}
 +v\log(C/\epsilon),\qquad 0<\epsilon\le1.
\end{aligned}
\label{eq:app-estimated-entropy-bound}
\end{equation}
\par
Condition on one training sample and use its event on which the estimated bridge bounds and
$(1+Q_n)\sup_{y\in\mathcal Y}\|\widehat h_{a,y}^{(-k)}-h_{a,y}\|_W\le\delta_n$
hold. The evaluation observations remain iid from $P_0$.
For $\delta_n>0$, apply the bounded-envelope local maximal inequality of \citet[Theorem~2.1]{vandervaartwellner2011local} only to the first term in \eqref{eq:app-estimated-minus-true}, using $B_n$, $\delta_n$, and \eqref{eq:app-estimated-entropy-bound}. Combining it with the preceding treatment bridge bound gives
\begin{align}
&E\!\left[\sup_{y\in\mathcal Y}
 \left|\sqrt{n_k}(P_{n,k}-E)d_{a,y}^{(-k)}\right|
 \,\middle|\,\text{training sample}\right]\nonumber\\
&\quad\le C\left\{
 \|\widehat q_a^{(-k)}(Z,X)-q_a(Z,X)\|_{a,Z}
 +\delta_n\sqrt{\Lambda_n}
 +\frac{B_n\Lambda_n}{\sqrt{n_k}}\right\}.
\label{eq:app-local-maximal-bound}
\end{align}
The difference $d_{a,y}^{(-k)}(O)$ is defined in \eqref{eq:app-estimated-minus-true}.
To see the dependence on the covering entropy explicitly, the normalized entropy integral appearing in the local maximal inequality is bounded by
\begin{equation}
\int_0^{\delta_n/B_n}
\sqrt{
1+\min\{\log(n+1),\,d_n\log(C/\epsilon)\}
+v\log(C/\epsilon)
}\,d\epsilon
\le
C\frac{\delta_n}{B_n}\sqrt{\Lambda_n}.
\label{eq:app-entropy-integral}
\end{equation}
To obtain \eqref{eq:app-entropy-integral}, bound the integral once using $\log(n+1)$ and once using $d_n\log(C/\epsilon)$, then take the smaller of the two bounds.
Substitution into the local maximal inequality, followed by addition of the treatment bridge bound, yields \eqref{eq:app-local-maximal-bound}.
If $\delta_n=0$, the first term is zero in $L_2(P_0)$ for every $y$ on the specified training-sample event; pointwise measurability makes it zero simultaneously on the evaluation sample almost surely. Its contribution is then omitted without evaluating a logarithm or dividing by $\delta_n$.
Because $r_{q,a,n}=o_p(1)$, $n_k\asymp n$, and \eqref{eq:app-empirical-rate} holds, the conditional bound tends to zero in probability on events whose probabilities tend to one.
Conditional Markov's inequality, with its probability bound truncated at one before averaging over the training sample, therefore gives
\begin{align}
\max_{a\in\{0,1\}}\max_{1\le k\le K}\sup_{y\in\mathcal Y}
\Bigg|
\sqrt{n_k}(P_{n,k}-E)
d_{a,y}^{(-k)}
\Bigg|
=o_p(1).
\label{eq:app-empirical-negligibility}
\end{align}
The maximum over folds in \eqref{eq:app-empirical-negligibility} is valid because $K$ is fixed. The argument conditions on one training sample at a time and does not require the overlapping training samples to be independent.
\subsection{Proof of Theorem~\ref{thm:uniform-dr}}
\label{app:uniform-expansion-proof}
\begin{proof}
Assumption~\ref{ass:nuisance}(i) gives \eqref{eq:app-product-rate} and $r_{q,a,n}=o_p(1)$, while \eqref{eq:app-delta} controls the outcome bridge estimation error.
Lemma~\ref{lem:app-bridge-class-size} verifies \eqref{eq:app-estimated-bridge-entropy} on the estimated bridge envelope events.
The separate bounds above establish \eqref{eq:app-empirical-negligibility} under \eqref{eq:app-empirical-rate}.
\noindent\textit{Step 1: Rewrite $\widehat F_a(y)$ fold by fold.}
By \eqref{eq:method-cdf-estimator} and \eqref{eq:app-evaluation-average},
\begin{align}
\widehat F_a(y)
=\sum_{k=1}^K\frac{n_k}{n}
P_{n,k}\Big[
\widehat h_{a,y}^{(-k)}(W,X)
+\mathbf{1}\{A=a\}\widehat q_a^{(-k)}(Z,X)
\{\mathbf{1}\{Y\le y\}-\widehat h_{a,y}^{(-k)}(W,X)\}
\Big].
\label{eq:app-estimator-fold-form}
\end{align}
By \eqref{eq:app-fold-average},
\begin{equation}
\frac{1}{n}\sum_{i=1}^n\varphi_{a,y}(O_i)
=
\sum_{k=1}^K\frac{n_k}{n}P_{n,k}\varphi_{a,y}(O).
\label{eq:app-influence-fold-form}
\end{equation}
Subtracting $F_a(y)$ and \eqref{eq:app-influence-fold-form} from \eqref{eq:app-estimator-fold-form} gives
\begin{align}
&\widehat F_a(y)-F_a(y)
-\frac{1}{n}\sum_{i=1}^n\varphi_{a,y}(O_i)\nonumber\\
&\quad=
\sum_{k=1}^K\frac{n_k}{n}
P_{n,k}
d_{a,y}^{(-k)}.
\label{eq:app-expansion-step1}
\end{align}
\noindent\textit{Step 2: Separate the population bias and sampling fluctuation.}
Add and subtract the population expectation inside each fold in \eqref{eq:app-expansion-step1}. Then
\begin{align}
&\widehat F_a(y)-F_a(y)
-\frac{1}{n}\sum_{i=1}^n\varphi_{a,y}(O_i)\nonumber\\
&\quad=
\sum_{k=1}^K\frac{n_k}{n}
E\{d_{a,y}^{(-k)}(O)\}\nonumber\\
&\qquad+
\sum_{k=1}^K\frac{n_k}{n}
(P_{n,k}-E)
d_{a,y}^{(-k)}.
\label{eq:app-exact-expansion}
\end{align}
\noindent\textit{Step 3: Control the population bias.}
Because $E\{\varphi_{a,y}(O)\}=0$, the first sum on the right side of \eqref{eq:app-exact-expansion} is exactly the fold-weighted population bias in \eqref{eq:app-exact-product-bias}. Hence \eqref{eq:app-projected-product-bias-bound} gives
\begin{align}
&\max_{a\in\{0,1\}}\sup_{y\in\mathcal Y}
\sqrt n
\left|
\sum_{k=1}^K\frac{n_k}{n}
E\{d_{a,y}^{(-k)}(O)\}
\right|\nonumber\\
&\quad\le
\sqrt n\max_{a\in\{0,1\}}\Bigg\{
\max_k\left\|\Sigma_a^\top
\{\widehat\alpha_{a,\lambda_q}^{(-k)}-\alpha_a^{(d_n)}\}\right\|_2\\
&\hspace{38mm}\times
\max_k\sup_{y\in\mathcal Y}
\left\|\widehat\theta_{a,\lambda_h}^{(-k)}(y)-\theta_a^{(d_n)}(y)\right\|_2
+b_{h,a,n}b_{q,a,n}\Bigg\}
=o_p(1)
\label{eq:app-bias-negligible}
\end{align}
by \eqref{eq:app-product-rate}.
\noindent\textit{Step 4: Control the sampling fluctuation caused by estimated bridges.}
For the second sum on the right side of \eqref{eq:app-exact-expansion},
\begin{align}
&\sqrt n
\left|
\sum_{k=1}^K\frac{n_k}{n}
(P_{n,k}-E)d_{a,y}^{(-k)}
\right|\le
\sum_{k=1}^K
\sqrt{\frac{n_k}{n}}
\left|
\sqrt{n_k}(P_{n,k}-E)d_{a,y}^{(-k)}
\right|.
\label{eq:app-fold-fluctuation-bound}
\end{align}
Since $K$ is fixed, $n_k/n$ is bounded, and \eqref{eq:app-empirical-negligibility} holds,
\begin{equation}
\max_{a\in\{0,1\}}\sup_{y\in\mathcal Y}
\sqrt n
\left|
\sum_{k=1}^K\frac{n_k}{n}
(P_{n,k}-E)d_{a,y}^{(-k)}
\right|
=o_p(1).
\label{eq:app-fluctuation-negligible}
\end{equation}
\noindent\textit{Step 5: Combine the two remainder terms.}
Combining \eqref{eq:app-exact-expansion}, \eqref{eq:app-bias-negligible}, and \eqref{eq:app-fluctuation-negligible} gives
\begin{equation}
\max_{a\in\{0,1\}}\sup_{y\in\mathcal Y}
\left|
\sqrt n\{\widehat F_a(y)-F_a(y)\}
-\frac{1}{\sqrt n}\sum_{i=1}^n\varphi_{a,y}(O_i)
\right|
=o_p(1),
\label{eq:app-final-uniform-expansion}
\end{equation}
which proves \eqref{eq:method-uniform-expansion}.
\par
\noindent\textit{Step 6: Establish simultaneous weak convergence.}
Under Assumption~\ref{ass:process}, the class
\[
\{\mathbf{1}\{Y\le y\}:y\in\mathcal Y\}
\]
is a bounded VC class, the class
\[
\{h_{a,y}(W,X):y\in\mathcal Y\}
\]
has the bounded polynomial covering entropy in \eqref{eq:app-true-bridge-entropy}, and $q_a(Z,X)$ is fixed and bounded for each treatment arm.
Products and sums of these bounded classes preserve polynomial covering entropy. The centering constants $F_a(y)\in[0,1]$ also form a bounded polynomial-entropy class. Therefore
\[
\{\varphi_{a,y}(O):a\in\{0,1\},y\in\mathcal Y\}
\]
is a bounded $P_0$-Donsker class. Its tight Gaussian limit is a random element of $\ell^\infty(\{0,1\}\times\mathcal Y)$; its paths need not be continuous in the usual distance on $\mathcal Y$.
The empirical-process central limit theorem then gives simultaneous weak convergence of
\[
\left\{
\frac{1}{\sqrt n}\sum_{i=1}^n\varphi_{a,y}(O_i):
a\in\{0,1\},y\in\mathcal Y
\right\}.
\]
The uniform remainder in \eqref{eq:app-final-uniform-expansion} is $o_p(1)$ in $\ell^\infty(\{0,1\}\times\mathcal Y)$. Slutsky's theorem therefore gives simultaneous weak convergence of $\{\sqrt n\{\widehat F_a(y)-F_a(y)\}:a\in\{0,1\},y\in\mathcal Y\}$.
\end{proof}

\section{Derivation and Proof for Simultaneous Inference for \texorpdfstring{$F_a(y)$ and $\Delta_Q(\tau)$}{F_a(y) and Delta_Q(tau)}}
\label{app:simultaneous-inference}
We derive the multiplier process and confidence-band inversion in Section~\ref{sec:method-simultaneous-inference}, then prove Theorem~\ref{thm:method-simultaneous-inference}.
The estimated bridges, $\widehat F_a(y)$, $\varphi_{a,y}(O)$, $\widehat\varphi_{a,y}^{(-k)}(O_i)$, and all simultaneous confidence bands retain the definitions in the main text.
The quantities $n_k$, $P_{n,k}$, $B_n$, $\delta_n$, and $\Lambda_n$ retain the definitions in Appendix~\ref{app:uniform-expansion}.
The multipliers $\xi_1,\ldots,\xi_n$ are mutually independent, independent of the observations and the fold partition, and satisfy $P(\xi_i=1)=P(\xi_i=-1)=1/2$.
Write $P_\xi$ and $E_\xi$ for probability and expectation over the multipliers conditional on all observations and the fold partition.
Population expectations containing estimated bridges integrate over $O\sim P_0$ with the corresponding training sample fixed.
Probabilities under $P_0$ for sample-dependent events include the independent fold partition.
\subsection{Derivation of the Multiplier Process and Simultaneous Confidence Bands}
\label{app:c-multiplier-derivation}
\noindent\textit{Centering $\widehat\varphi_{a,y}^{(-k)}(O_i)$.}
Substituting \eqref{eq:method-estimated-influence-function} and then using \eqref{eq:method-cdf-estimator} gives
\begin{equation}
\begin{aligned}
&\frac1n\sum_{k=1}^K\sum_{i\in I_k}\widehat\varphi_{a,y}^{(-k)}(O_i)\\
&\quad=\frac1n\sum_{k=1}^K\sum_{i\in I_k}\Big[\widehat h_{a,y}^{(-k)}(W_i,X_i)
+\mathbf{1}\{A_i=a\}\widehat q_a^{(-k)}(Z_i,X_i)\\
&\hspace{48mm}\times\{\mathbf{1}\{Y_i\le y\}-\widehat h_{a,y}^{(-k)}(W_i,X_i)\}\Big]-\widehat F_a(y)\\
&\quad=\widehat F_a(y)-\widehat F_a(y)=0.
\end{aligned}
\label{eq:app-c-estimated-centering}
\end{equation}
Equation~\eqref{eq:app-c-estimated-centering} establishes centering over the full sample.
The average of $\widehat\varphi_{a,y}^{(-k)}(O_i)$ within one fold need not be zero.
\par
\noindent\textit{Conditional mean and covariance.}
The multiplier distribution gives
\[
E_\xi(\xi_i)=0,
\qquad
E_\xi(\xi_i\xi_j)=
\begin{cases}
1,&i=j,\\
0,&i\ne j.
\end{cases}
\]
Conditionally on the observations and the fold partition, every $\widehat\varphi_{a,y}^{(-k)}(O_i)$ is fixed.
Consequently,
\begin{equation}
E_\xi\{\widehat{\mathbb G}_a^\xi(y)\}=0.
\label{eq:app-c-multiplier-mean}
\end{equation}
For $a,b\in\{0,1\}$ and $y,y'\in\mathcal Y$, multiplying the two sums in \eqref{eq:method-multiplier-process} gives
\begin{equation}
\begin{aligned}
&E_\xi\{\widehat{\mathbb G}_a^\xi(y)\widehat{\mathbb G}_b^\xi(y')\}\\
&\quad=\frac1n\sum_{k=1}^K\sum_{\ell=1}^K
\sum_{i\in I_k}\sum_{j\in I_\ell}
E_\xi(\xi_i\xi_j)
\widehat\varphi_{a,y}^{(-k)}(O_i)\widehat\varphi_{b,y'}^{(-\ell)}(O_j)\\
&\quad=\frac1n\sum_{k=1}^K\sum_{i\in I_k}
\widehat\varphi_{a,y}^{(-k)}(O_i)\widehat\varphi_{b,y'}^{(-k)}(O_i).
\end{aligned}
\label{eq:app-c-multiplier-covariance}
\end{equation}
The last equality uses the disjoint folds and the identity $E_\xi(\xi_i\xi_j)=0$ for $i\ne j$.
Using the same $\xi_i$ for all $a$ and $y$ retains the products of $\widehat\varphi_{a,y}^{(-k)}(O_i)$ and $\widehat\varphi_{b,y'}^{(-k)}(O_i)$ evaluated at the same observation in \eqref{eq:app-c-multiplier-covariance}.
The conditional distribution of $\widehat{\mathbb G}_a^\xi(y)$ need not be Gaussian at a finite sample size.
The Gaussian approximation is established below as a limit.
\par
\noindent\textit{The conditional critical value.}
The definition of $\widehat c_{1-\alpha_{\mathrm{sig}}}$ in the main text is equivalently
\begin{equation}
\widehat c_{1-\alpha_{\mathrm{sig}}}
=\inf\left\{t\ge0:
P_\xi\left(\max_{a\in\{0,1\}}\sup_{y\in\mathcal Y}
|\widehat{\mathbb G}_a^\xi(y)|\le t\right)\ge1-\alpha_{\mathrm{sig}}\right\}.
\label{eq:app-c-conditional-critical-value}
\end{equation}
The conditional critical value is nonnegative because the maximum of absolute values is nonnegative.
For every $a$ and $y$, $0\le F_a(y)\le1$ implies
\begin{equation}
\begin{aligned}
L_a(y)\le F_a(y)\le U_a(y)
&\quad\Longleftrightarrow\quad
\widehat F_a(y)-\frac{\widehat c_{1-\alpha_{\mathrm{sig}}}}{\sqrt n}
\le F_a(y)\le
\widehat F_a(y)+\frac{\widehat c_{1-\alpha_{\mathrm{sig}}}}{\sqrt n}\\
&\quad\Longleftrightarrow\quad
|\sqrt n\{\widehat F_a(y)-F_a(y)\}|\le\widehat c_{1-\alpha_{\mathrm{sig}}}.
\end{aligned}
\label{eq:app-c-band-equivalence}
\end{equation}
The first equivalence follows directly from the maximum with zero defining $L_a(y)$ and the minimum with one defining $U_a(y)$.
Thus proving simultaneous coverage of $F_a(y)$ reduces to proving that \eqref{eq:app-c-conditional-critical-value} consistently estimates the corresponding Gaussian critical value.
\par
\noindent\textit{Confidence-band inversion.}
On the simultaneous coverage event for $F_a(y)$, the pointwise inequalities $L_a(y)\le F_a(y)\le U_a(y)$ imply
\begin{equation}
\{y\in\mathcal Y:L_a(y)\ge\tau\}
\subseteq\{y\in\mathcal Y:F_a(y)\ge\tau\}
\subseteq\{y\in\mathcal Y:U_a(y)\ge\tau\}.
\label{eq:app-c-crossing-sets}
\end{equation}
For nonempty crossing sets, taking infima reverses the inclusions in \eqref{eq:app-c-crossing-sets}.
The infimum using $U_a(y)$ therefore gives the lower bound $Q_a^L(\tau)$, while the infimum using $L_a(y)$ gives the upper bound $Q_a^U(\tau)$.
Subtracting the upper bound for $Q_0(\tau)$ from the lower bound for $Q_1(\tau)$ gives $\Delta_Q^L(\tau)$.
Subtracting the lower bound for $Q_0(\tau)$ from the upper bound for $Q_1(\tau)$ gives $\Delta_Q^U(\tau)$.
The proof of Theorem~\ref{thm:method-simultaneous-inference} below verifies the coverage implication for every $\tau\in\mathcal T$, using the convention $\inf\varnothing=y_U$ for the infima over $\mathcal Y$ in \eqref{eq:method-quantile-bands}.
\subsection{Derivation of the Bridge Replacement Error}
\label{app:c-replacement-derivation}
Recall the training-sample-dependent difference
\begin{equation}
\begin{aligned}
d_{a,y}^{(-k)}(O)
={}&\widehat h_{a,y}^{(-k)}(W,X)-F_a(y)\\
&+\mathbf{1}\{A=a\}\widehat q_a^{(-k)}(Z,X)
\{\mathbf{1}\{Y\le y\}-\widehat h_{a,y}^{(-k)}(W,X)\}
-\varphi_{a,y}(O).
\end{aligned}
\label{eq:app-c-difference}
\end{equation}
The function $d_{a,y}^{(-k)}(O)$ is the difference already expanded in \eqref{eq:app-estimated-minus-true}.
The definition uses $F_a(y)$ rather than $\widehat F_a(y)$ so that $d_{a,y}^{(-k)}(O)$ depends only on the training sample when viewed as a function of $O$.
In contrast, $\widehat\varphi_{a,y}^{(-k)}(O_i)$ also uses $\widehat F_a(y)$, which depends on the full sample.
Subtracting \eqref{eq:method-influence-function} from \eqref{eq:method-estimated-influence-function} therefore gives the exact identity
\begin{equation}
\widehat\varphi_{a,y}^{(-k)}(O_i)-\varphi_{a,y}(O_i)
=d_{a,y}^{(-k)}(O_i)-\{\widehat F_a(y)-F_a(y)\}.
\label{eq:app-c-centering-difference}
\end{equation}
Multiply \eqref{eq:app-c-centering-difference} by $\xi_i/\sqrt n$ and sum over all folds:
\begin{equation}
\begin{aligned}
&\widehat{\mathbb G}_a^\xi(y)-\frac1{\sqrt n}\sum_{i=1}^n\xi_i\varphi_{a,y}(O_i)\\
&\quad=\sum_{k=1}^K\sqrt{\frac{n_k}{n}}
\frac1{\sqrt{n_k}}\sum_{i\in I_k}\xi_i d_{a,y}^{(-k)}(O_i)\\
&\qquad-\{\widehat F_a(y)-F_a(y)\}\frac1{\sqrt n}\sum_{i=1}^n\xi_i.
\end{aligned}
\label{eq:app-c-multiplier-decomposition}
\end{equation}
The first term on the right of \eqref{eq:app-c-multiplier-decomposition} is caused by replacing the true bridges by estimated bridges.
The second term is caused by replacing $F_a(y)$ by $\widehat F_a(y)$ in the centering.
The two terms require separate bounds.
\par
Expanding \eqref{eq:app-c-difference} using \eqref{eq:method-influence-function} yields
\begin{equation}
\begin{aligned}
d_{a,y}^{(-k)}(O)
={}&\{1-\mathbf{1}\{A=a\}\widehat q_a^{(-k)}(Z,X)\}
\{\widehat h_{a,y}^{(-k)}(W,X)-h_{a,y}(W,X)\}\\
&+\mathbf{1}\{A=a\}\{\widehat q_a^{(-k)}(Z,X)-q_a(Z,X)\}
\{\mathbf{1}\{Y\le y\}-h_{a,y}(W,X)\}.
\end{aligned}
\label{eq:app-c-bridge-difference-decomposition}
\end{equation}
The bounds in Appendix~\ref{app:uniform-expansion} consequently give, on the specified training-sample events,
\begin{equation}
\begin{aligned}
\sup_{y\in\mathcal Y}|d_{a,y}^{(-k)}(O)|&\le B_n+(Q_n+Q)(1+H)
\quad\text{almost surely},\\
\sup_{y\in\mathcal Y}\left[E\{d_{a,y}^{(-k)}(O)^2\}\right]^{1/2}
&\le(1+Q_n)r_{h,a,n}+(1+H)r_{q,a,n}
\le\delta_n+(1+H)r_{q,a,n}.
\end{aligned}
\label{eq:app-c-difference-radius}
\end{equation}
Write $\mathcal D_{a,n}^{(-k)}=\{d_{a,y}^{(-k)}(O):y\in\mathcal Y\}$.
Combining the separate covering arguments in \eqref{eq:app-estimated-entropy-form} and \eqref{eq:app-estimated-entropy-bound} also gives the whole-class bound
\begin{equation}
\begin{aligned}
&\sup_{\mathsf Q}\log N\!\left(\epsilon\{B_n+(Q_n+Q)(1+H)\},
 \mathcal D_{a,n}^{(-k)},L_2(\mathsf Q)\right)\\
&\qquad\le\min\{\log(N_k+1),\,d_n\log(C/\epsilon)\}
 +C(v+1)\log(C/\epsilon),\qquad 0<\epsilon\le1.
\end{aligned}
\label{eq:app-c-difference-entropy}
\end{equation}
The proof below uses the two terms separately, rather than applying the whole-class entropy bound to the treatment bridge estimation error.
\subsection{Proof of Multiplier Control for the Bridge Replacement Error}
\label{app:c-multiplier-error-proof}
\begin{lemma}[Multiplier control for $d_{a,y}^{(-k)}(O)$]
\label{lem:app-c-multiplier-error}
Under the conditions of Theorem~\ref{thm:uniform-dr} and pointwise measurability of the function classes, for every $\epsilon>0$,
\begin{equation}
P_\xi\left\{
\max_{a\in\{0,1\}}\max_{1\le k\le K}\sup_{y\in\mathcal Y}
\left|\frac1{\sqrt{n_k}}\sum_{i\in I_k}\xi_i d_{a,y}^{(-k)}(O_i)\right|
>\epsilon\right\}\longrightarrow0
\label{eq:app-c-multiplier-error}
\end{equation}
in $P_0$-probability.
\end{lemma}
\begin{proof}
\noindent\textit{Step 1: Condition on one training sample.}
Fix $a$ and $k$ and condition on the fold partition and $\{O_i:i\notin I_k\}$.
Let $\mathcal E_{a,n}^{(-k)}$ be the training-sample event on which the estimated bridge bounds, the covering entropy bound, and the following inequality hold:
\[
\begin{aligned}
&(1+Q_n)\sup_{y\in\mathcal Y}
\|\widehat h_{a,y}^{(-k)}(W,X)-h_{a,y}(W,X)\|_W
\le\delta_n.
\end{aligned}
\]
Assumption~\ref{ass:nuisance} implies $P_0\{(\mathcal E_{a,n}^{(-k)})^c\}\to0$.
Only fold-specific errors appear in $\mathcal E_{a,n}^{(-k)}$, so the conditioning does not fix the errors from other training samples.
Conditionally on the training sample, the pairs $(O_i,\xi_i)$ for $i\in I_k$ are iid, and the functions $d_{a,y}^{(-k)}(O)$ are fixed.
\par
\noindent\textit{Step 2: Check the multiplier function class.}
Multiply each of the two terms in \eqref{eq:app-c-bridge-difference-decomposition} by $\xi$ separately.
Since $|\xi|=1$, their respective envelopes remain $B_n$ and
$(1+H)\mathbf{1}\{A=a\}|\widehat q_a^{(-k)}(Z,X)-q_a(Z,X)|$.
Independence and $E_\xi\xi=0$ make every multiplier function mean zero.
Their $L_2$ norms are unchanged: the first is at most $\delta_n$ on $\mathcal E_{a,n}^{(-k)}$, and the second is at most
$(1+H)\|\widehat q_a^{(-k)}(Z,X)-q_a(Z,X)\|_{a,Z}$.
For a finitely supported probability measure $\mathsf Q$ on $(O,\xi)$, let $\mathsf Q_O$ denote its marginal measure on $O$.
For any $d,d'\in\mathcal D_{a,n}^{(-k)}$,
\[
\|\xi d-\xi d'\|_{L_2(\mathsf Q)}^2
=\int\xi^2\{d(O)-d'(O)\}^2\,d\mathsf Q(O,\xi)
=\|d-d'\|_{L_2(\mathsf Q_O)}^2.
\]
The same identity holds for either term separately.
Thus the first multiplier class inherits \eqref{eq:app-estimated-entropy-bound}, and the second inherits the fixed polynomial entropy from the reweighting argument following \eqref{eq:app-estimated-entropy-form}.
Pointwise measurability is preserved by multiplication by $\xi$ and by squaring.
\par
\noindent\textit{Step 3: Apply the two maximal inequalities separately.}
For $0<u\le1$, define the deterministic entropy-integral bound
\[
J_n(u)=\int_0^u
\sqrt{1+\min\{\log(n+1),\,d_n\log(C/t)\}+v\log(C/t)}\,dt.
\]
Substituting $t=us$, using $\int_0^1\sqrt{\log(1/s)}\,ds<\infty$, and bounding the minimum by each of its two entries gives
\[
\begin{aligned}
J_n(u)
&\le Cu\min\!\left\{
 \sqrt{1+\log(n+1)+v\log(C/u)},\,
 \sqrt{1+(d_n+v)\log(C/u)}\right\}.
\end{aligned}
\]
For $\delta_n>0$, it follows that
$J_n(\delta_n/B_n)\le C(\delta_n/B_n)\sqrt{\Lambda_n}$.
Normalize only the first multiplier class by $B_n$ and apply the local maximal inequality of \citet[Theorem~2.1]{vandervaartwellner2011local}, with normalized radius $\delta_n/B_n$.
For the second multiplier class, the uniform entropy maximal inequality \citep[Section~2.14]{vandervaartwellner1996} uses its fixed entropy integral and the $L_2(P_0)$ norm of its envelope, giving the separate bound
$C\|\widehat q_a^{(-k)}(Z,X)-q_a(Z,X)\|_{a,Z}$.
Every multiplier function has mean zero, so no additional centering term is needed.
Adding the bounds gives, on $\mathcal E_{a,n}^{(-k)}$,
\begin{equation}
\begin{aligned}
&E\left[E_\xi\left\{\sup_{y\in\mathcal Y}
\left|\frac1{\sqrt{n_k}}\sum_{i\in I_k}\xi_i d_{a,y}^{(-k)}(O_i)\right|\right\}
\,\middle|\,\{I_j\}_{j=1}^K,\{O_i:i\notin I_k\}\right]\\
&\quad\le C\|\widehat q_a^{(-k)}(Z,X)-q_a(Z,X)\|_{a,Z}
 +CB_n\left\{J_n(\delta_n/B_n)
 +\frac{J_n(\delta_n/B_n)^2}{(\delta_n/B_n)^2\sqrt{n_k}}\right\}\\
&\quad\le C\left\{\|\widehat q_a^{(-k)}(Z,X)-q_a(Z,X)\|_{a,Z}
 +\delta_n\sqrt{\Lambda_n}
 +\frac{B_n\Lambda_n}{\sqrt{n_k}}\right\}.
\end{aligned}
\label{eq:app-c-multiplier-maximal}
\end{equation}
If $\delta_n=0$, the first multiplier class vanishes simultaneously on the evaluation sample almost surely, as in Appendix~\ref{app:estimated-bridge-fluctuation}.
Its contribution is omitted, and the last bound in \eqref{eq:app-c-multiplier-maximal} remains valid under the convention $\Lambda_n=0$.
\par
\noindent\textit{Step 4: Obtain convergence conditional on all observations.}
Markov's inequality, \eqref{eq:app-c-multiplier-maximal}, and $n_k\asymp n$ imply
\begin{equation}
\begin{aligned}
&E\left[P_\xi\left\{\sup_{y\in\mathcal Y}
\left|\frac1{\sqrt{n_k}}\sum_{i\in I_k}\xi_i d_{a,y}^{(-k)}(O_i)\right|>\epsilon\right\}\right]\\
&\quad\le P_0\{(\mathcal E_{a,n}^{(-k)})^c\}
 +E\!\left[\min\!\left\{1,\frac{C}{\epsilon}
 \|\widehat q_a^{(-k)}(Z,X)-q_a(Z,X)\|_{a,Z}\right\}\right]\\
&\qquad+\frac{C}{\epsilon}\left\{\delta_n\sqrt{\Lambda_n}
 +\frac{B_n\Lambda_n}{\sqrt n}\right\}\longrightarrow0.
\end{aligned}
\label{eq:app-c-simultaneous-multiplier-error}
\end{equation}
The bounded expectation involving the treatment bridge error tends to zero because $r_{q,a,n}=o_p(1)$; no $L_1$ convergence of this error is required.
The rate condition in \eqref{eq:app-empirical-rate} makes the last term vanish.
A union bound over the two treatment arms and fixed $K$ gives convergence of the expectation of the conditional probability in \eqref{eq:app-c-multiplier-error} to zero.
That conditional probability lies in $[0,1]$.
A second application of Markov's inequality shows that the conditional probability converges to zero in $P_0$-probability, proving \eqref{eq:app-c-multiplier-error}.
The argument conditions on one training sample at a time and never assumes independence of overlapping training samples.
\end{proof}
\subsection{Proof of Conditional Multiplier Convergence}
\label{app:c-multiplier-convergence-proof}
Let $\{\mathbb G_a:a\in\{0,1\}\}$ denote the centered tight Gaussian limit of the leading process $\{n^{-1/2}\sum_{i=1}^n\varphi_{a,y}(O_i):a\in\{0,1\},y\in\mathcal Y\}$ established in Appendix~\ref{app:uniform-expansion}.
The covariance is
\begin{equation}
E\{\mathbb G_a(y)\mathbb G_b(y')\}
=E\{\varphi_{a,y}(O)\varphi_{b,y'}(O)\},
\qquad a,b\in\{0,1\},\quad y,y'\in\mathcal Y.
\label{eq:app-c-gaussian-covariance}
\end{equation}
The notation $\rightsquigarrow_\xi$ denotes conditional weak convergence given all observations and the fold partition in the bounded-Lipschitz sense, in $P_0$-probability.
The bounded-Lipschitz sense means uniform convergence of conditional expectations over real functionals bounded in absolute value by one with Lipschitz constant at most one for the supremum norm.
\begin{lemma}[Conditional multiplier convergence]
\label{lem:app-c-multiplier-convergence}
Under the conditions of Lemma~\ref{lem:app-c-multiplier-error},
\begin{equation}
\{\widehat{\mathbb G}_a^\xi:a\in\{0,1\}\}
\rightsquigarrow_\xi\{\mathbb G_a:a\in\{0,1\}\}
\quad\text{in }\ell^\infty(\{0,1\}\times\mathcal Y),
\label{eq:app-c-conditional-convergence}
\end{equation}
in $P_0$-probability.
\end{lemma}
\begin{proof}
\noindent\textit{Step 1: Apply the multiplier central limit theorem to the mean-zero functions $\varphi_{a,y}(O)$.}
Appendix~\ref{app:uniform-expansion} establishes that $\{\varphi_{a,y}(O):a\in\{0,1\},y\in\mathcal Y\}$ is a bounded, pointwise measurable $P_0$-Donsker class and that $E\{\varphi_{a,y}(O)\}=0$.
The multipliers have mean zero and variance one.
The multiplier integrability condition holds because $\int_0^\infty\sqrt{P(|\xi_i|>t)}\,dt=1$.
The conditional multiplier central limit theorem in \citet[Section~2.9]{vandervaartwellner1996} therefore gives
\begin{equation}
\begin{aligned}
&\Bigg\{\frac1{\sqrt n}\sum_{i=1}^n\xi_i
\left[\varphi_{a,y}(O_i)-\frac1n\sum_{j=1}^n\varphi_{a,y}(O_j)\right]:
a\in\{0,1\},\ y\in\mathcal Y\Bigg\}\\
&\qquad\rightsquigarrow_\xi\{\mathbb G_a(y):a\in\{0,1\},\ y\in\mathcal Y\}.
\end{aligned}
\label{eq:app-c-true-multiplier-limit}
\end{equation}
The empirical centering in \eqref{eq:app-c-true-multiplier-limit} changes the process by
\[
-\left\{\frac1n\sum_{j=1}^n\varphi_{a,y}(O_j)\right\}
\frac1{\sqrt n}\sum_{i=1}^n\xi_i.
\]
The simultaneous empirical-process limit in Appendix~\ref{app:uniform-expansion} gives
\[
\max_{a\in\{0,1\}}\sup_{y\in\mathcal Y}
\left|\frac1n\sum_{j=1}^n\varphi_{a,y}(O_j)\right|=O_p(n^{-1/2}).
\]
Also, independence and unit variance of the multipliers give the exact identity
\begin{equation}
E_\xi\left\{\left(\frac1{\sqrt n}\sum_{i=1}^n\xi_i\right)^2\right\}=1.
\label{eq:app-c-multiplier-sum}
\end{equation}
Conditional Chebyshev's inequality therefore makes the empirical-centering difference negligible uniformly over $a$ and $y$.
Thus \eqref{eq:app-c-true-multiplier-limit} remains valid with $n^{-1/2}\sum_i\xi_i\varphi_{a,y}(O_i)$ in place of the empirically centered process.
\par
\noindent\textit{Step 2: Bound the bridge replacement term.}
The triangle inequality gives
\[
\begin{aligned}
&\max_{a\in\{0,1\}}\sup_{y\in\mathcal Y}
\left|\sum_{k=1}^K\sqrt{\frac{n_k}{n}}
\frac1{\sqrt{n_k}}\sum_{i\in I_k}\xi_i d_{a,y}^{(-k)}(O_i)\right|\\
&\quad\le\left(\sum_{k=1}^K\sqrt{\frac{n_k}{n}}\right)
\max_{a\in\{0,1\}}\max_{1\le k\le K}\sup_{y\in\mathcal Y}
\left|\frac1{\sqrt{n_k}}\sum_{i\in I_k}\xi_i d_{a,y}^{(-k)}(O_i)\right|.
\end{aligned}
\]
Cauchy--Schwarz and $\sum_k n_k/n=1$ imply $\sum_k\sqrt{n_k/n}\le\sqrt K$.
Lemma~\ref{lem:app-c-multiplier-error} therefore makes the bridge replacement term in \eqref{eq:app-c-multiplier-decomposition} negligible conditionally on all observations.
\par
\noindent\textit{Step 3: Bound the replacement of $F_a(y)$ in the centering.}
Theorem~\ref{thm:uniform-dr} and tightness of the leading process imply
\begin{equation}
\max_{a\in\{0,1\}}\sup_{y\in\mathcal Y}
|\widehat F_a(y)-F_a(y)|=O_p(n^{-1/2}).
\label{eq:app-c-uniform-cdf-rate}
\end{equation}
By \eqref{eq:app-c-multiplier-sum}, for every $\epsilon>0$,
\[
\begin{aligned}
&P_\xi\left\{\max_{a\in\{0,1\}}\sup_{y\in\mathcal Y}
\left|\{\widehat F_a(y)-F_a(y)\}\frac1{\sqrt n}\sum_{i=1}^n\xi_i\right|>\epsilon\right\}\\
&\quad\le\epsilon^{-2}
\left[\max_{a\in\{0,1\}}\sup_{y\in\mathcal Y}|\widehat F_a(y)-F_a(y)|\right]^2
\longrightarrow0
\end{aligned}
\]
in $P_0$-probability.
The centering error is controlled after conditioning on all observations, not on a single training sample.
\par
\noindent\textit{Step 4: Combine the conditional limits.}
Using \eqref{eq:app-c-multiplier-decomposition}, Steps 2 and 3 yield
\begin{equation}
P_\xi\left\{\max_{a\in\{0,1\}}\sup_{y\in\mathcal Y}
\left|\widehat{\mathbb G}_a^\xi(y)-\frac1{\sqrt n}\sum_{i=1}^n\xi_i\varphi_{a,y}(O_i)\right|
>\epsilon\right\}\longrightarrow0
\label{eq:app-c-process-replacement}
\end{equation}
in $P_0$-probability for every $\epsilon>0$.
For a functional bounded by one with Lipschitz constant one, the difference of conditional expectations under the two processes in \eqref{eq:app-c-process-replacement} is bounded by $\epsilon$ plus twice the conditional probability in \eqref{eq:app-c-process-replacement}.
The bound is uniform over the bounded-Lipschitz functionals.
Combining Step 1 with \eqref{eq:app-c-process-replacement} and then letting $\epsilon\downarrow0$ proves \eqref{eq:app-c-conditional-convergence}.
\end{proof}
\subsection{Proof of Continuity and Separation of the Gaussian Critical Value}
\label{app:c-critical-value-proof}
Define the Gaussian critical value, corresponding to \eqref{eq:app-c-conditional-critical-value}, by
\begin{equation}
c_{1-\alpha_{\mathrm{sig}}}=\inf\left\{t\ge0:
P\left(\max_{a\in\{0,1\}}\sup_{y\in\mathcal Y}|\mathbb G_a(y)|\le t\right)
\ge1-\alpha_{\mathrm{sig}}\right\}.
\label{eq:app-c-population-critical-value}
\end{equation}
Probabilities and expectations of $\mathbb G_a$ are taken under the Gaussian limit law determined by $P_0$.
\begin{lemma}[Continuity and separation of $c_{1-\alpha_{\mathrm{sig}}}$]
\label{lem:app-c-critical-value}
Under the conditions of Lemma~\ref{lem:app-c-multiplier-convergence} and \eqref{eq:method-nondegenerate-variance}, the distribution of $\max_{a\in\{0,1\}}\sup_{y\in\mathcal Y}|\mathbb G_a(y)|$ is continuous and strictly increasing on $[0,\infty)$.
The Gaussian critical value $c_{1-\alpha_{\mathrm{sig}}}$ is finite and positive, and for every $0<\epsilon<c_{1-\alpha_{\mathrm{sig}}}$,
\begin{equation}
\begin{aligned}
P\left\{\max_{a\in\{0,1\}}\sup_{y\in\mathcal Y}|\mathbb G_a(y)|
\le c_{1-\alpha_{\mathrm{sig}}}-\epsilon\right\}
&<1-\alpha_{\mathrm{sig}}\\
&<P\left\{\max_{a\in\{0,1\}}\sup_{y\in\mathcal Y}|\mathbb G_a(y)|
\le c_{1-\alpha_{\mathrm{sig}}}+\epsilon\right\}.
\end{aligned}
\label{eq:app-c-critical-separation}
\end{equation}
\end{lemma}
\begin{proof}
\noindent\textit{Step 1: Establish boundedness and nondegeneracy of the Gaussian maximum.}
By Assumption~\ref{ass:process} and $0\le F_a(y)\le1$,
\[
|\varphi_{a,y}(O)|\le(1+Q)(1+H)
\quad\text{almost surely}.
\]
Equation~\eqref{eq:app-c-gaussian-covariance} and the condition \eqref{eq:method-nondegenerate-variance} imply
\[
0<\max_{a\in\{0,1\}}\sup_{y\in\mathcal Y}E\{\mathbb G_a(y)^2\}
\le(1+Q)^2(1+H)^2<\infty.
\]
The bounded polynomial covering entropy established in Appendix~\ref{app:uniform-expansion} has a finite entropy integral. The Gaussian process therefore has a separable tight version with finite expected absolute supremum. At least one Gaussian coordinate has positive variance, so
\[
P\left\{\max_{a\in\{0,1\}}\sup_{y\in\mathcal Y}|\mathbb G_a(y)|=0\right\}=0.
\]
No positive lower bound is required for the variances of all coordinates.
\par
\noindent\textit{Step 2: Prove continuity.}
Apply the Gaussian supremum continuity result of \citet{gaenssler2007continuity} to the centered Gaussian process containing both $\mathbb G_a(y)$ and $-\mathbb G_a(y)$ for every $a$ and $y$. Its supremum equals $\max_a\sup_y|\mathbb G_a(y)|$, is almost surely finite, and is not degenerate at zero by Step 1. Its distribution is therefore continuous. Equivalently, uniform continuity of the resulting real distribution function gives
\begin{equation}
\lim_{\epsilon\downarrow0}\sup_{x\in\mathbb R}
P\left\{\left|\max_{a\in\{0,1\}}\sup_{y\in\mathcal Y}|\mathbb G_a(y)|-x\right|\le\epsilon\right\}=0.
\label{eq:app-c-gaussian-continuity}
\end{equation}
This qualitative continuity is sufficient for critical-value consistency.
\par
\noindent\textit{Step 3: Prove strict increase and separation.}
View the separable tight version as a centered Gaussian random element in a separable closed subspace of $\ell^\infty(\{0,1\}\times\mathcal Y)$. The support of a centered Gaussian measure is a closed linear subspace \citep[Theorem~1]{vakhania1975support}. Step 1 ensures that the support contains a nonzero function. For any $0\le r<s<\infty$, scaling a nonzero support function gives a support function with supremum norm $(r+s)/2$. An open ball of radius less than $(s-r)/2$ around that function has positive probability and contains only functions whose norms lie between $r$ and $s$. Thus
\[
P\left\{r<\max_{a\in\{0,1\}}\sup_{y\in\mathcal Y}|\mathbb G_a(y)|<s\right\}>0.
\]
Every open ball about zero also has positive probability, and a nondegenerate Gaussian coordinate exceeds every finite absolute bound with positive probability. Together with continuity and almost sure finiteness, these facts imply that the distribution function is continuous and strictly increasing on $[0,\infty)$. For $0<\alpha_{\mathrm{sig}}<1$, the critical value is finite and positive and satisfies
\[
P\left\{\max_{a\in\{0,1\}}\sup_{y\in\mathcal Y}|\mathbb G_a(y)|\le c_{1-\alpha_{\mathrm{sig}}}\right\}=1-\alpha_{\mathrm{sig}}.
\]
Strict increase proves \eqref{eq:app-c-critical-separation}.
\end{proof}
\subsection{Proof of Theorem~\ref{thm:method-simultaneous-inference}}
\label{app:c-coverage-proof}
\begin{proof}
\noindent\textit{Step 1: Prove consistency of $\widehat c_{1-\alpha_{\mathrm{sig}}}$.}
The maximum of absolute values is a Lipschitz functional for the supremum norm.
Lemma~\ref{lem:app-c-multiplier-convergence}, the conditional continuous mapping theorem, and Lemma~\ref{lem:app-c-critical-value} therefore imply that, for each fixed $t\in\mathbb R$,
\begin{equation}
\begin{aligned}
&P_\xi\left\{\max_{a\in\{0,1\}}\sup_{y\in\mathcal Y}|\widehat{\mathbb G}_a^\xi(y)|\le t\right\}\\
&\qquad\longrightarrow
P\left\{\max_{a\in\{0,1\}}\sup_{y\in\mathcal Y}|\mathbb G_a(y)|\le t\right\}
\end{aligned}
\label{eq:app-c-maximum-distribution}
\end{equation}
in $P_0$-probability.
Fix $0<\epsilon<c_{1-\alpha_{\mathrm{sig}}}$ and apply \eqref{eq:app-c-maximum-distribution} at $c_{1-\alpha_{\mathrm{sig}}}-\epsilon$ and $c_{1-\alpha_{\mathrm{sig}}}+\epsilon$.
By \eqref{eq:app-c-critical-separation}, with probability tending to one,
\[
\begin{aligned}
P_\xi\left\{\max_{a\in\{0,1\}}\sup_{y\in\mathcal Y}|\widehat{\mathbb G}_a^\xi(y)|
\le c_{1-\alpha_{\mathrm{sig}}}-\epsilon\right\}&<1-\alpha_{\mathrm{sig}},\\
P_\xi\left\{\max_{a\in\{0,1\}}\sup_{y\in\mathcal Y}|\widehat{\mathbb G}_a^\xi(y)|
\le c_{1-\alpha_{\mathrm{sig}}}+\epsilon\right\}&>1-\alpha_{\mathrm{sig}}.
\end{aligned}
\]
The definition \eqref{eq:app-c-conditional-critical-value} then places $\widehat c_{1-\alpha_{\mathrm{sig}}}$ between $c_{1-\alpha_{\mathrm{sig}}}-\epsilon$ and $c_{1-\alpha_{\mathrm{sig}}}+\epsilon$.
Hence
\begin{equation}
\widehat c_{1-\alpha_{\mathrm{sig}}}\longrightarrow c_{1-\alpha_{\mathrm{sig}}}
\quad\text{in }P_0\text{-probability}.
\label{eq:app-c-critical-consistency}
\end{equation}
No continuity of the finite-sample conditional multiplier distribution is required.
\par
\noindent\textit{Step 2: Prove simultaneous coverage of $F_a(y)$.}
Theorem~\ref{thm:uniform-dr} and the continuous mapping theorem give
\begin{equation}
\max_{a\in\{0,1\}}\sup_{y\in\mathcal Y}|\sqrt n\{\widehat F_a(y)-F_a(y)\}|
\rightsquigarrow
\max_{a\in\{0,1\}}\sup_{y\in\mathcal Y}|\mathbb G_a(y)|.
\label{eq:app-c-estimation-maximum}
\end{equation}
For every $\epsilon>0$, the probability of the event with the conditional critical value is bounded by
\begin{equation}
\begin{aligned}
&P_0\left\{\max_{a\in\{0,1\}}\sup_{y\in\mathcal Y}
|\sqrt n\{\widehat F_a(y)-F_a(y)\}|\le c_{1-\alpha_{\mathrm{sig}}}-\epsilon\right\}
-P_0\{|\widehat c_{1-\alpha_{\mathrm{sig}}}-c_{1-\alpha_{\mathrm{sig}}}|>\epsilon\}\\
&\quad\le P_0\left\{\max_{a\in\{0,1\}}\sup_{y\in\mathcal Y}
|\sqrt n\{\widehat F_a(y)-F_a(y)\}|\le\widehat c_{1-\alpha_{\mathrm{sig}}}\right\}\\
&\quad\le P_0\left\{\max_{a\in\{0,1\}}\sup_{y\in\mathcal Y}
|\sqrt n\{\widehat F_a(y)-F_a(y)\}|\le c_{1-\alpha_{\mathrm{sig}}}+\epsilon\right\}
+P_0\{|\widehat c_{1-\alpha_{\mathrm{sig}}}-c_{1-\alpha_{\mathrm{sig}}}|>\epsilon\}.
\end{aligned}
\label{eq:app-c-coverage-sandwich}
\end{equation}
Take $n\to\infty$ using \eqref{eq:app-c-critical-consistency} and \eqref{eq:app-c-estimation-maximum}, then let $\epsilon\downarrow0$.
Lemma~\ref{lem:app-c-critical-value} gives
\[
P_0\left\{\max_{a\in\{0,1\}}\sup_{y\in\mathcal Y}
|\sqrt n\{\widehat F_a(y)-F_a(y)\}|\le\widehat c_{1-\alpha_{\mathrm{sig}}}\right\}
\longrightarrow1-\alpha_{\mathrm{sig}}.
\]
The pointwise equivalence in \eqref{eq:app-c-band-equivalence}, applied to every $a$ and $y$, proves \eqref{eq:method-cdf-coverage}.
The argument does not require independence between $\widehat c_{1-\alpha_{\mathrm{sig}}}$ and $\widehat F_a(y)$.
\par
\noindent\textit{Step 3: Derive simultaneous coverage of $Q_a(\tau)$.}
Fix a realization for which $L_a(y)\le F_a(y)\le U_a(y)$ for every $a\in\{0,1\}$ and $y\in\mathcal Y$.
Let $\tau\in\mathcal T$.
The definition of $Q_a(\tau)$ and right-continuity of the marginal counterfactual distribution $F_a(y)$ give
\[
F_a\{Q_a(\tau)\}\ge\tau,
\qquad
F_a(y)<\tau\quad\text{for every }y<Q_a(\tau).
\]
The assumed interior location of $Q_a(\tau)$ ensures $Q_a(\tau)\in\mathcal Y$.
Since $U_a\{Q_a(\tau)\}\ge F_a\{Q_a(\tau)\}\ge\tau$, the set defining $Q_a^L(\tau)$ is nonempty and contains $Q_a(\tau)$.
Taking the infimum gives $Q_a^L(\tau)\le Q_a(\tau)$.
For every $y\in\mathcal Y$ with $y<Q_a(\tau)$, the inequality $L_a(y)\le F_a(y)<\tau$ excludes $y$ from the set defining $Q_a^U(\tau)$.
When $\{y\in\mathcal Y:L_a(y)\ge\tau\}$ is nonempty, the infimum defining $Q_a^U(\tau)$ is at least $Q_a(\tau)$.
When $\{y\in\mathcal Y:L_a(y)\ge\tau\}$ is empty, the convention $\inf\varnothing=y_U$ gives $Q_a^U(\tau)=y_U>Q_a(\tau)$.
Therefore,
\begin{equation}
Q_a^L(\tau)\le Q_a(\tau)\le Q_a^U(\tau)
\quad\text{for every }a\in\{0,1\},\ \tau\in\mathcal T
\label{eq:app-c-quantile-containment}
\end{equation}
on the simultaneous coverage event for the marginal counterfactual distribution $F_a(y)$.
The argument is deterministic on the simultaneous coverage event and applies to all quantile levels without a union bound over $\mathcal T$.
Monotonicity of the lower band $L_a(y)$ and upper band $U_a(y)$ is not required for \eqref{eq:app-c-quantile-containment}.
The argument is the confidence-band inversion used in \citet{chernozhukov2020generic}.
\par
\noindent\textit{Step 4: Derive simultaneous coverage of $\Delta_Q(\tau)$.}
On the same simultaneous coverage event, \eqref{eq:app-c-quantile-containment} yields
\[
\begin{aligned}
\Delta_Q^L(\tau)
&=Q_1^L(\tau)-Q_0^U(\tau)
\le Q_1(\tau)-Q_0(\tau)=\Delta_Q(\tau),\\
\Delta_Q^U(\tau)
&=Q_1^U(\tau)-Q_0^L(\tau)
\ge Q_1(\tau)-Q_0(\tau)=\Delta_Q(\tau).
\end{aligned}
\]
Consequently,
\begin{equation}
\begin{aligned}
&P_0\{\Delta_Q^L(\tau)\le\Delta_Q(\tau)\le\Delta_Q^U(\tau)
\text{ for all }\tau\in\mathcal T\}\\
&\quad\ge P_0\{L_a(y)\le F_a(y)\le U_a(y)
\text{ for all }a\in\{0,1\},\ y\in\mathcal Y\}.
\end{aligned}
\label{eq:app-c-qte-event-containment}
\end{equation}
The right side converges to $1-\alpha_{\mathrm{sig}}$ by Step 2.
Taking the lower limit of the left side proves \eqref{eq:method-qte-coverage}.
The inequality can be strict because simultaneous coverage of $\Delta_Q(\tau)$ can occur without simultaneous coverage of the marginal counterfactual distributions $F_0(y)$ and $F_1(y)$.
No continuity of $F_a$ on $\mathcal Y$, uniqueness of a solution to $F_a(y)=\tau$, counterfactual density, or derivative of the counterfactual quantile function $Q_a(\tau)$ is used in Steps 3 and 4.
\par
\noindent\textit{Step 5: Bound the widths without density conditions.}
Fix any realization of the data and any number
\[
\epsilon\ge\max_{a\in\{0,1\}}\sup_{y\in\mathcal Y}|\widehat F_a(y)-F_a(y)|
+\frac{\widehat c_{1-\alpha_{\mathrm{sig}}}}{\sqrt n}.
\]
The definitions in \eqref{eq:method-cdf-bands} imply $U_a(y)\le F_a(y)+\epsilon$ and $L_a(y)\ge F_a(y)-\epsilon$. Consequently,
\[
\begin{aligned}
\{y\in\mathcal Y:U_a(y)\ge\tau\}
&\subseteq\{y\in\mathcal Y:F_a(y)\ge\tau-\epsilon\},\\
\{y\in\mathcal Y:F_a(y)\ge\tau+\epsilon\}
&\subseteq\{y\in\mathcal Y:L_a(y)\ge\tau\}.
\end{aligned}
\]
For all crossing sets in this step, retain the convention that the infimum of an empty set is $y_U$. Taking infima gives
\[
Q_a^L(\tau)\ge\inf\{y\in\mathcal Y:F_a(y)\ge\tau-\epsilon\},
\qquad
Q_a^U(\tau)\le\inf\{y\in\mathcal Y:F_a(y)\ge\tau+\epsilon\}.
\]
These inequalities also hold for empty crossing sets, since all other infima lie in $[y_L,y_U]$. Subtracting yields the deterministic bound
\begin{equation}
\begin{aligned}
Q_a^U(\tau)-Q_a^L(\tau)
\le{}&\inf\{y\in\mathcal Y:F_a(y)\ge\tau+\epsilon\}\\
&-\inf\{y\in\mathcal Y:F_a(y)\ge\tau-\epsilon\}.
\end{aligned}
\label{eq:app-c-quantile-width}
\end{equation}
Because $\Delta_Q^U(\tau)-\Delta_Q^L(\tau)=\sum_{a=0}^1\{Q_a^U(\tau)-Q_a^L(\tau)\}$, the same argument gives
\begin{equation}
\begin{aligned}
\Delta_Q^U(\tau)-\Delta_Q^L(\tau)
\le\sum_{a=0}^1\Big[&\inf\{y\in\mathcal Y:F_a(y)\ge\tau+\epsilon\}\\
&-\inf\{y\in\mathcal Y:F_a(y)\ge\tau-\epsilon\}\Big].
\end{aligned}
\label{eq:app-c-qte-width}
\end{equation}
The right-hand sides are nonnegative, so the bounds also apply to the positive parts of the displayed endpoint differences. No CDF-coverage event has been assumed in deriving \eqref{eq:app-c-quantile-width} and \eqref{eq:app-c-qte-width}.
\par
On the CDF-coverage event already established in Step 2, \eqref{eq:app-c-band-equivalence} gives $\max_a\sup_y|\widehat F_a(y)-F_a(y)|\le\widehat c_{1-\alpha_{\mathrm{sig}}}/\sqrt n$. Thus $\epsilon=2\widehat c_{1-\alpha_{\mathrm{sig}}}/\sqrt n$ is valid on that event. More generally, choosing $\epsilon$ equal to the error-plus-half-width expression above gives $\epsilon=O_p(n^{-1/2})$ by \eqref{eq:app-c-uniform-cdf-rate} and \eqref{eq:app-c-critical-consistency}. When $\tau-\epsilon,\tau+\epsilon\in(0,1)$ and the corresponding true quantiles lie in the interior of $\mathcal Y$, the right-hand side of \eqref{eq:app-c-quantile-width} equals $Q_a(\tau+\epsilon)-Q_a(\tau-\epsilon)$. The crossing-set version remains valid without these interior restrictions. The bounds describe how changes in the probability level affect band width, rather than imposing smoothness of the counterfactual density. No universal root-$n$ rate for the quantile-band widths follows without additional restrictions on the counterfactual distribution.
\end{proof}
The coverage statements use the conditional critical value defined in \eqref{eq:app-c-conditional-critical-value}.
A finite number of multiplier draws approximates $\widehat c_{1-\alpha_{\mathrm{sig}}}$ numerically.
The coverage statements do not assert that a fixed number of multiplier draws has negligible simulation error.

\section{More Details for the Numerical Studies}
\label{app:simulation-common}

This section specifies the observed data $O=(Y,A,Z,W,X)$ and settings shared by the three simulations, followed by the settings specific to each simulation.
The observed covariates $X$, unobserved latent variable $U$,
treatment-inducing proxy $Z$, outcome-inducing proxy $W$, treatment indicator
$A$, outcome $Y$, and potential outcome $Y(a)$ retain their definitions in
the main text. We write $X=(X_1,X_2)$ to identify the two components of the
observed covariates $X$.

\noindent\textit{Observed data $O=(Y,A,Z,W,X)$.}
We draw $X_1$ and $X_2$ independently from Bernoulli$(1/2)$.
Conditional on $X$, the unobserved latent variable $U$ is binary, with
\begin{equation}
P(U=1\mid X)
=\frac{1}{1+\exp\{-(-0.25+0.55X_1-0.35X_2)\}}.
\label{eq:sim-common-u}
\end{equation}
Conditional on $(U,X)$, the treatment-inducing proxy $Z$ and outcome-inducing
proxy $W$ are independent binary variables, with
\begin{equation}
\begin{aligned}
P(Z=1\mid U,X)
&=\frac{1}{1+\exp\{-(-0.15+2.2\rho(2U-1)+0.25X_1-0.15X_2)\}},\\
P(W=1\mid U,X)
&=\frac{1}{1+\exp\{-(0.10+2.2\rho(2U-1)-0.20X_1+0.20X_2)\}}.
\end{aligned}
\label{eq:sim-common-proxies}
\end{equation}
The coefficient $2.2\rho$ controls the dependence of each proxy on $U$. The resulting $\kappa_{a,n}$ is the smallest singular value of the population matrix in \eqref{eq:method-population-matrix} and depends on the chosen basis. The values of $\rho$ are specified separately below. We then draw the treatment indicator $A$ from the binary distribution
\begin{equation}
P(A=1\mid U,Z,W,X)
=\frac{1}{1+\exp\{-(-0.20+0.80(2U-1)+0.45Z+0.25X_1-0.20X_2)\}}.
\label{eq:sim-common-a}
\end{equation}
The conditional probability in \eqref{eq:sim-common-a} does not depend on $W$.
Together with the conditional independence of $Z$ and $W$, this gives
$W\perp(A,Z)\mid U,X$. The conditional probabilities are strictly between
zero and one, so latent positivity $0<P(A=a\mid U,X)<1$ holds for both
treatment arms $a\in\{0,1\}$.

\noindent\textit{Potential outcome $Y(a)$ and outcome $Y$.} For each treatment arm $a\in\{0,1\}$, the conditional mean of the
potential outcome $Y(a)$ is
\begin{equation}
\begin{aligned}
E\{Y(a)\mid U,W,X\}
={}&0.35a+0.35X_1-0.25X_2+0.75U+0.25W+0.20aU+0.15aX_1.
\end{aligned}
\label{eq:sim-common-outcome-mean}
\end{equation}
Draw from the noise distribution specified for each simulation and standardize the draw to mean zero and variance one. Generate the potential outcome by adding this noise, multiplied by $0.70+0.10X_1+0.05X_2$, to the conditional mean:
\begin{equation}
\begin{aligned}
Y(a)
={}&E\{Y(a)\mid U,W,X\}+(0.70+0.10X_1+0.05X_2)
  \times\text{standardized noise}.
\end{aligned}
\label{eq:sim-common-standardized-noise}
\end{equation}
The draws used for $Y(0)$ and $Y(1)$ are mutually independent and independent of $(U,Z,W,A,X)$ before multiplication by this scale. Thus, the additive noise has conditional mean zero and conditional variance $(0.70+0.10X_1+0.05X_2)^2$ in every simulation.
The observed outcome is
\begin{equation}
Y=AY(1)+(1-A)Y(0).
\label{eq:sim-common-observed-outcome}
\end{equation}
Thus, $Y(1)$ is observed when $A=1$, and $Y(0)$ is observed when $A=0$. The generation above satisfies $Y(a)\perp A\mid U,X$,
$Y\perp Z\mid U,A,X$, and $W\perp(A,Z)\mid U,X$. Each replication consists of $n$ independent observations
$O_1,\ldots,O_n$. The unobserved latent variable $U$ is used to generate the data
and is excluded from bridge estimation.

\noindent\textit{Prespecified basis vectors $b_W(W,X)$ and $b_Z(Z,X)$.}
All three simulations use $d_n=8$. For the Walsh construction, the prespecified basis vector $b_W(W,X)$ contains the constant $1$,
the three variables $2X_1-1$, $2X_2-1$, and $2W-1$, their three pairwise
products, and their product $(2X_1-1)(2X_2-1)(2W-1)$.
The prespecified basis vector $b_Z(Z,X)$ has the same form, with $Z$
replacing $W$. These eight functions span all functions on the eight possible values of $(W,X)$ or $(Z,X)$.
Simulations~1 and~3 use this Walsh construction. Simulation~2 also considers ReLU and GB-leaf constructions, fixed before the simulation replications as described in Appendix~\ref{app:simulation2-bases}.
We use the bridge estimators in \eqref{eq:method-regularized-bridges}--\eqref{eq:method-estimated-bridges}, with the ridge parameters specified separately below.
The basis vectors and ridge parameters are fixed across outcome thresholds $y$ within each training sample.
We randomly divide each replication into $K=3$ folds of sizes differing by at most one, and use 500 replications for each setting.
All training-sample moments use the denominator $N_k$, including those containing $\mathbf{1}\{A=a\}$.

\noindent\textit{Marginal counterfactual distribution $F_a(y)$ and population expectations.}
We calculate the marginal counterfactual distribution $F_a(y)$ from
\begin{equation}
F_a(y)=E\bigl[P\{Y(a)\le y\mid U,W,X\}\bigr],
\label{eq:sim-common-population-cdf}
\end{equation}
using the specified noise distribution for the conditional probability
and summing over the 16 possible values of $(U,W,X)$.
For the bridge estimation errors and population bias in Simulations~1 and~2, we likewise evaluate population moments by summing over the binary variables and integrating the outcome indicator using its conditional probability.
The true bridges $h_{a,y}(W,X)$ and $q_a(Z,X)$ are obtained from the population moment equations in \eqref{eq:app-population-coefficient-equations}.
These population quantities are used only to evaluate the simulation results and are excluded from bridge estimation.
Since all retained basis constructions span the complete eight-dimensional spaces, $b_{h,a,n}=b_{q,a,n}=0$.
All numerical suprema over outcome thresholds are evaluated on the grid specified for each simulation.

\subsection{More Details for Simulation 1}
\label{app:simulation1-dgp}

We use $n\in\{1600,3200,4800,5400\}$ and
$\rho\in\{0.9,0.75,0.6,0.45,0.3,0.2\}$ in \eqref{eq:sim-common-proxies}.
The standardized noise in \eqref{eq:sim-common-standardized-noise} follows $N(0,1)$.
Only $\rho$ changes the conditional distributions of $Z$ and $W$; the other conditional distributions, the conditional mean of $Y(a)$, and the noise scale remain fixed.
We evaluate the bridges at 401 equally spaced outcome thresholds in $\mathcal Y=[-1.5,3.5]$.
We compare $\lambda_h=\lambda_q=0$ with $\lambda_h=\lambda_q=0.1/N_k$ on each training sample, using the same observations and folds for both choices.
The unregularized estimates solve the sample moment equations in \eqref{eq:sample-outcome-bridge}--\eqref{eq:sample-treatment-bridge} whenever the sample matrix is numerically nonsingular.

The first two panels of Figure~\ref{fig:simulation1-bridges-bias} report
$\max_{a\in\{0,1\}}r_{h,a,n}$ and $\max_{a\in\{0,1\}}r_{q,a,n}$.
The population norms are calculated as defined in Appendix~\ref{app:bridge-error-definitions}, and $\kappa_{a,n}$ is calculated from \eqref{eq:method-weak-proxy-strength}.
Figure~\ref{fig:simulation1-bridges-bias} holds $n=3200$ fixed and uses $\min_{a\in\{0,1\}}\kappa_{a,n}$ as the horizontal coordinate in all four panels.

The last two panels report the $\sqrt n$-scaled absolute population bias in \eqref{eq:method-bridge-product-bias}, maximized over treatment arms, training folds, and outcome thresholds.
The Bound underlying Corollary~\ref{cor:method-weak-proxy-bias} is calculated from the right-hand side of \eqref{eq:app-projected-product-bias-bound}, maximized over treatment arms and multiplied by $\sqrt n$. The additional bound \eqref{eq:app-ridge-product-bias-bound} is not calculated in these panels.
The Cauchy--Schwarz bound is $\sqrt n\max_{a\in\{0,1\}}r_{h,a,n}r_{q,a,n}$.
Each point represents the median for one proxy strength and ridge choice at $n=3200$; its upper bar is the 90th percentile, conditional on successful fits.
Numerically singular unregularized fits are recorded as failures and excluded from these summaries.

\subsection{More Details for Simulation 2}
\label{app:simulation2-bases}

We use $\rho\in\{0.75,0.60,0.45\}$, corresponding to $2.2\rho\in\{1.65,1.32,0.99\}$ in \eqref{eq:sim-common-proxies}, and $n\in\{1600,3200,4800,5400\}$. We compare all three basis constructions at every value of $\rho$.
The standardized noise follows $N(0,1)$.
We set $\lambda_h=\lambda_q=0.1/N_k$ on each training sample and use 401 equally spaced outcome thresholds in $\mathcal Y=[-1.5,3.5]$.
Within each replication at each value of $\rho$, Walsh, ReLU, and GB-leaf use the same observations and folds.

\noindent\textit{Prespecified basis vectors $b_W(W,X)$ and $b_Z(Z,X)$.}
Walsh uses the construction above.
For ReLU and GB-leaf, we first generate one independent pilot sample of 5,400 observations from the distribution with $\rho=0.75$.
For each construction, we separately regress the observed outcome $Y$ on $(X_1,X_2,W)$ and on $(X_1,X_2,Z)$, pooling the two treatment arms.
The three inputs are centered and scaled using their pilot means and standard deviations.
Neither the unobserved latent variable $U$ nor the simulation replications are used to train these regressions.

ReLU uses one hidden layer with 32 ReLU units, fitted by Adam with an initial learning rate of $0.001$, an $L_2$ coefficient regularization parameter of $0.001$, batch size 200, and at most 300 iterations.
Early stopping uses 10\% of the pilot sample for validation, tolerance $10^{-4}$, and 20 iterations without improvement.
The candidate basis functions are the fitted hidden-unit activations.
GB-leaf uses 36 regression trees with squared-error loss, maximum depth three, learning rate $0.05$, and at least 20 observations per leaf.
Its candidate basis functions are the indicators of the fitted leaves across the trees.

For each input side, we evaluate the candidate functions on all eight possible binary input values, and center and scale each function using equal weights on those values.
After removing numerically constant functions, column-pivoted QR selects seven linearly independent centered and scaled functions; adding the constant $1$ gives $d_n=8$.
We retain the selected functions themselves, without replacing them by the orthogonal functions from QR, and check that the resulting evaluation matrix has rank eight.
The two basis vectors are then held fixed across all sample sizes, proxy strengths, replications, training folds, treatment arms, and outcome thresholds.
Thus, the three constructions span the same complete function spaces, while the ridge penalty is applied to the coefficients in each construction.
The pilot sample is excluded from $n$, and the comparison is conditional on the fixed pilot sample.
Only the bridge coefficients are estimated on each training sample.
The sample matrix and ridge parameters remain fixed across $y$, retaining the shared-threshold structure in Lemma~\ref{lem:app-bridge-class-sizemain}.

\noindent\textit{Estimated-bridge sampling fluctuation $V_n$ and uniform expansion remainder $R_n$.}
We use $d_{a,y}^{(-k)}(O)$ as defined in \eqref{eq:app-c-difference}.
For each training sample, its population expectation holds the estimated bridges fixed and is evaluated using the population moments above.
We subtract this expectation before combining the evaluation-fold averages:
\[
V_n=\sqrt n\max_{a\in\{0,1\}}\sup_{y\in\mathcal Y}
\left|\sum_{k=1}^K\frac{n_k}{n}
\{P_{n,k}d_{a,y}^{(-k)}-E d_{a,y}^{(-k)}\}\right|.
\]
This removes the population-bias contribution from $V_n$.
The uniform expansion remainder is
\[
R_n=\max_{a\in\{0,1\}}\sup_{y\in\mathcal Y}
\left|\sqrt n\{\widehat F_a(y)-F_a(y)\}
-\frac{1}{\sqrt n}\sum_{i=1}^n\varphi_{a,y}(O_i)\right|,
\]
using the same observations to evaluate the estimated bridges and the true bridges.
Both quantities use the counterfactual CDF estimate in \eqref{eq:method-cdf-estimator}, without clipping or monotonicity adjustment; $R_n$ is not divided by a critical value.
The signed sampling fluctuation and population bias sum to the signed remainder before taking the maximum absolute value.
Figure~\ref{fig:simulation2} reports the median $V_n$ and $R_n$ for each sample size and basis construction in columns corresponding to $\rho=0.75$, $0.60$, and $0.45$. The top and bottom rows show $V_n$ and $R_n$, respectively, and each panel uses its own logarithmic vertical range. The annotated $\min_{a\in\{0,1\}}\kappa_{a,n}$ values refer to the Walsh basis.

\subsection{More Details for Simulation 3}
\label{app:simulation3-dgp}

We fix $\rho=0.75$ and use $n\in\{1600,3200,4800\}$, with
$\lambda_h=\lambda_q=0.01/\sqrt{N_k}$ on each training sample.
Only the noise distribution changes across the four settings; the conditional mean in \eqref{eq:sim-common-outcome-mean}, the conditional distributions of $(U,Z,W,A,X)$, and the noise scale in \eqref{eq:sim-common-standardized-noise} remain fixed.
The Beta settings have zero influence-function variance at thresholds outside the outcome support, which is allowed by \eqref{eq:method-nondegenerate-variance}.
Draw from the selected noise distribution in Table~\ref{tab:sim3-noise-standardization} and standardize the draw by subtracting the listed mean and dividing by the listed standard deviation.

\begin{table}[!htbp]
\centering
\caption{Noise distributions before standardization in Simulation 3.}
\label{tab:sim3-noise-standardization}
\small
\setlength{\tabcolsep}{5pt}
\renewcommand{\arraystretch}{1.08}
\begin{tabular}{@{}lcc@{}}
\toprule
Noise distribution & Mean & Standard deviation \\
\midrule
Double Weibull$(1/2)$ & $0$ & $\sqrt{24}$ \\
Beta$(1/4,1/4)$ & $1/2$ & $1/\sqrt{6}$ \\
Beta$(1/4,3/4)$ & $1/4$ & $\sqrt{3/32}$ \\
$\chi^2_1$ & $1$ & $\sqrt{2}$ \\
\bottomrule
\end{tabular}
\end{table}

Before standardization, Double Weibull$(1/2)$ assigns an independent, equally likely positive or negative sign to a nonnegative draw whose cumulative probability at $y\ge0$ is $1-\exp(-\sqrt y)$.
We calculate the counterfactual quantile function $Q_a(\tau)$ by numerical inversion of $F_a(y)$ in \eqref{eq:sim-common-population-cdf} and obtain the quantile treatment effect $\Delta_Q(\tau)=Q_1(\tau)-Q_0(\tau)$.

The counterfactual CDF estimate $\widehat F_a(y)$ is evaluated
at 501 equally spaced outcome thresholds in $\mathcal Y=[-3,5]$.
We use $\alpha_{\mathrm{sig}}=0.05$, 999 multiplier bootstrap draws, and
$\mathcal T=\{0.10,0.12,\ldots,0.90\}$.
Proposed obtains $[\Delta_Q^L(\tau),\Delta_Q^U(\tau)]$ by confidence-band inversion of $[L_a(y),U_a(y)]$.
Estimated-density delta uses the same $\widehat F_a(y)$ and additionally estimates the counterfactual density.
Coverage requires the band to contain $\Delta_Q(\tau)$ for every $\tau\in\mathcal T$; mean band width averages the widths over $\mathcal T$ and successful replications. Failed fits count as noncoverage.
Oracle 95\% is the width of a symmetric, constant-width band: twice the
4,751st smallest maximum absolute error in estimating $\Delta_Q(\tau)$ over
$\mathcal T$ among 5,000 replications independent of the 500 replications above.
\end{document}